\newif\ifabstract
\abstracttrue
 \abstractfalse 
\newif\iffull
\ifabstract \fullfalse \else \fulltrue \fi

\documentclass[11pt]{article}
\usepackage{amsfonts}
\usepackage{amssymb}
\usepackage{amstext}
\usepackage{amsmath}
\usepackage{xspace}
\usepackage{ntheorem}
\usepackage{graphicx}
\usepackage{url}
\usepackage{graphics}
\usepackage{colordvi}
\usepackage{colordvi}
\usepackage{subfigure}
\usepackage{hyperref,url}

\advance \oddsidemargin by -0.8in
\newcommand{\myparskip}{3pt}
\parskip \myparskip

\newtheorem*{definition}{Definition.}[section]

\newenvironment{proofof}[1]{\noindent{\bf Proof of #1.}}%
        {\hspace*{\fill}$\Box$\par\vspace{4mm}}

\def\etal{et al.\xspace}

\def\floor#1{\lfloor {#1} \rfloor}
\def\ceil#1{\lceil {#1} \rceil}
\def\script#1{\mathcal{#1}}
\def\card#1{|#1|}

\def\mP{\script{P}}

\def\mX{\script{X}}

\DeclareMathOperator*{\argmax}{arg\,max}

\newcommand{\tc}{\tilde c}
\def\cKRV{\gamma_{\mathrm{CMG}}}
\newcommand{\gkrv}{\ensuremath{\gamma_{\mbox{\tiny{\sc CMG}}}}}

\newcommand{\gammaKRV}{\ensuremath{\gamma_{\mbox{\tiny{\sc CMG}}}}}

\newcommand{\alphaKRV}{\ensuremath{\alpha_{\mbox{\tiny{\sc CMG}}}}}

\newcommand{\alphawl}{\ensuremath{\alpha_{\mbox{\tiny{\sc BW}}}}}
\newcommand{\alphaWL}{\ensuremath{\alpha_{\mbox{\tiny{\sc BW}}}}}
\newcommand{\alphabw}{\ensuremath{\alpha_{\mbox{\tiny{\sc BW}}}}}
\newcommand{\alphaBW}{\ensuremath{\alpha_{\mbox{\tiny{\sc BW}}}}}

\newcommand{\partition}{\mathsf{PARTITION}}
\newcommand{\separate}{\mathsf{SEPARATE}}

\newcommand{\sconnect}{\overset{\mbox{\tiny{1:1}}}{\leadsto}}

\newcommand{\tw}{\mathrm{tw}}

\newcommand{\alphasc}{\ensuremath{\beta_{\mbox{\tiny{\sc ARV}}}}}
\newcommand{\algsc}{\ensuremath{{\mathcal{A}}_{\mbox{\textup{\scriptsize{ARV}}}}}\xspace}

\newcommand{\algSC}{\algsc}
\newcommand{\tH}{\tilde{H}}
\newcommand{\tpset}{\tilde{\pset}}
\newcommand{\nset}{\mathcal{N}}
\newcommand{\tU}{\tilde{U}}
\newcommand{\tT}{\tilde{T}}
\newcommand{\G}{{\mathbf{G}}}
\renewcommand{\H}{{\mathbf{H}}}

\newcommand{\event}{{\cal{E}}}

\newcommand{\polylog}[1]{\mathrm{polylog(#1)}}

\newcommand{\set}[1]{\left\{ #1 \right\}}
\newcommand{\sse}{\subseteq}

\newcommand{\tset}{{\mathcal T}}

\newcommand{\uset}{{\mathcal U}}
\newcommand{\gset}{{\mathcal G}}

\newcommand{\iset}{{\mathcal{I}}}
\newcommand{\pset}{{\mathcal{P}}}
\newcommand{\qset}{{\mathcal{Q}}}
\newcommand{\lset}{{\mathcal{L}}}
\newcommand{\bset}{{\mathcal{B}}}
\newcommand{\aset}{{\mathcal{A}}}
\newcommand{\cset}{{\mathcal{C}}}
\newcommand{\fset}{{\mathcal{F}}}
\newcommand{\mset}{{\mathcal M}}

\newcommand{\xset}{{\mathcal{X}}}
\newcommand{\wset}{{\mathcal{W}}}
\newcommand{\ttset}{\tilde{\mathcal T}}
\newcommand{\tZ}{\tilde{Z}}

\newcommand{\tA}{\tilde A}
\newcommand{\tB}{\tilde B}

\newcommand{\yset}{{\mathcal{Y}}}
\newcommand{\rset}{{\mathcal{R}}}

\newcommand{\hset}{{\mathcal{H}}}
\newcommand{\sset}{{\mathcal{S}}}

\newcommand{\nots}{\overline S}

\newcommand{\be}{\begin{enumerate}}
\newcommand{\ee}{\end{enumerate}}
\newcommand{\bd}{\begin{description}}
\newcommand{\ed}{\end{description}}
\newcommand{\bi}{\begin{itemize}}
\newcommand{\ei}{\end{itemize}}

\newtheorem{theorem}{Theorem}[section]
\newtheorem{lemma}[theorem]{Lemma}
\newtheorem{observation}[theorem]{Observation}
\newtheorem{corollary}[theorem]{Corollary}
\newtheorem{conjecture}{Conjecture}

\newtheorem{claim}[theorem]{Claim}

\newenvironment{proof}{\par \smallskip{\bf Proof:}}{\hfill\stopproof}
\def\stopproof{\square}
\def\square{\vbox{\hrule height.2pt\hbox{\vrule width.2pt height5pt \kern5pt
\vrule width.2pt} \hrule height.2pt}}

\renewcommand{\phi}{\varphi}
\newcommand{\eps}{\epsilon}

\newcommand{\half}{\ensuremath{\frac{1}{2}}}

\newcommand{\poly}{\operatorname{poly}}

\newcommand{\R}{\ensuremath{\mathbb R}}
\newcommand{\Z}{\ensuremath{\mathbb Z}}

\newcommand{\prob}[2][]{\text{\bf Pr}_{#1}\left [#2\right]}

\newenvironment{properties}[2][0]
{
\begin{enumerate} \setcounter{enumi}{#1}}{\end{enumerate}}

\newcommand{\mynote}[1]{{\sc\bf{[#1]}}}

\newcommand{\out}{\operatorname{out}}

\begin{document}

\title{
Polynomial Bounds for the Grid-Minor Theorem\footnote{A preliminary version of this paper appeared in {\em Proceedings of ACM STOC, 2014}.}
%
}
\author{Chandra Chekuri\thanks{Dept.\ of Computer Science, University
of Illinois, Urbana, IL 61801. {\tt chekuri@illinois.edu}.
Supported in part by NSF grants CCF-1016684 \& CCF-1319376,
and by TTI Chicago during a sabbatical visit in Fall 2013.}
\and 
Julia Chuzhoy\thanks{Toyota Technological Institute, Chicago, IL
60637. Email: {\tt cjulia@ttic.edu}. Supported in part by NSF CAREER 
grant CCF-0844872, NSF grant CCF-1318242 and Sloan Research Fellowship.}
}

\maketitle

\begin{abstract}
  One of the key results in Robertson and Seymour's seminal work on
  graph minors is the Grid-Minor Theorem (also called the Excluded
  Grid Theorem). The theorem states that for every grid $H$, every
  graph whose treewidth is large enough relative to $|V(H)|$
  contains $H$ as a minor. This theorem has found many applications in
  graph theory and algorithms. Let $f(k)$ denote the largest value
  such that every graph of treewidth $k$ contains a grid minor of size
  $(f(k)\times f(k))$. The best previous quantitative bound, due to
  recent work of Kawarabayashi and
  Kobayashi~\cite{KawarabayashiK-grid}, and Leaf and
  Seymour~\cite{LeafS12}, shows that $f(k)=\Omega(\sqrt{\log k/\log
    \log k})$. In contrast, the best known upper bound implies that
  $f(k) = O(\sqrt{k/\log k})$ \cite{RobertsonST94}. In this paper we
  obtain the first polynomial relationship between treewidth and grid
  minor size by showing that $f(k)=\Omega(k^{\delta})$ for some fixed
  constant $\delta > 0$, and describe a randomized algorithm, whose
  running time is polynomial in $|V(G)|$ and $k$, that with high
  probability finds a model of such a grid minor in $G$.
\end{abstract}

\label{-------------------------------------------------Intro-------------------------------------}
\section{Introduction}\label{sec: intro}

The seminal work of Roberston and Seymour on graph minors makes
essential use of the notions of tree-decompositions and treewidth. A
key structural result in their work is the Grid-Minor theorem (also
called the Excluded Grid theorem), which states that for every grid
$H$, every graph whose treewidth is large enough relative to $|V(H)|$
contains $H$ as a minor. This theorem has found many applications in
graph theory and algorithms.  Let $f(k)$ denote the largest value,
such that every graph of treewidth k contains a grid minor of size
$(f(k) \times f(k))$.
The quantitative estimate for $f$ given in the original
proof of Robertson and Seymour \cite{RS-grid} was substantially
improved by Robertson, Seymour and Thomas \cite{RobertsonST94} who
showed that $f(k) = \Omega(\log^{1/5} k)$; see
\cite{DiestelJGT99,Diestel-book} for a simpler proof with a slightly
weaker bound.  There have been recent improvements by Kawarabayashi
and Kobayashi \cite{KawarabayashiK-grid}, and by Leaf and Seymour
\cite{LeafS12},  giving the best previous bound of
$f(k) = \Omega(\sqrt{\log k/\log \log k})$. On the other hand, the known upper
bounds on $f$ are polynomial in $k$. It is easy to see, for example by
considering the complete graph on $n$ nodes, whose treewidth is $n-1$, that
$f(k) = O(\sqrt{k})$.  This can be slightly improved to $f(k) =
O(\sqrt{k/\log k})$ by considering sparse random graphs (or
$\Omega(\log n)$-girth constant-degree expanders)
\cite{RobertsonST94}. Robertson \etal \cite{RobertsonST94}
 suggest that this value may
be sufficient, and Demaine \etal \cite{DemaineHK09} conjecture that the bound of $f(k)=\Theta(k^{1/3})$ is both necessary and sufficient. It has been  an important open problem to prove a polynomial relationship between a graph's treewidth and the 
size of the largest grid minor in it.
 In this paper we prove the following theorem, which accomplishes this goal,
while also giving a polynomial-time randomized algorithm to find a model of the grid minor. Given a function $f: {\Z}^+\rightarrow {\Z}^+$, we say that $f(m)=O(\poly(m))$, if  $f(m)=O(m^c)$ for some constant $c$ independent of $m$. Similarly, we say that $f(m)=O(\poly\log m)$, if $f(m)=O(\log^cm)$ for some constant $c>0$ independent of $m$.
We use notation $\Omega(\poly m)$ and $\Omega(\poly\log m)$ analogously.

\begin{theorem}\label{thm: main}
There is a universal constant $\delta > 0$, such that for every $k\geq 1$, every graph $G$ of treewidth $k$ contains a grid of size
  $\left (\Omega(k^{\delta}/\poly\log k)\times\Omega(k^{\delta}/\poly\log k)\right )$ as a minor. Moreover, there is a randomized algorithm that, given $G$, with high probability outputs a model of the grid minor in time $O(\poly(|V(G)|\cdot k))$.

\end{theorem}

Our proof shows that $\delta$ is at least $1/98$ in the preceding
theorem.
 We note that the relationship between grid minors and treewidth is much tighter in some special classes of graphs. In planar graphs $f(k) = \Omega(k)$ \cite{RobertsonST94}; a similar linear relationship is known in bounded-genus graphs \cite{DemaineFHT05} and graphs that exclude a fixed graph $H$ as a minor \cite{DemaineH08} (see also \cite{KawarabayashiK-grid}).
  
We obtain the following corollary by observing that every simple planar graph $H$ is a minor of a grid of size
$(k' \times k')$, for $k' = O(|V(H)|)$ \cite{RobertsonST94}.
\begin{corollary}
  \label{cor:exclude-planar}
  There is a universal constant $c$ such that, if $G$ excludes a simple planar graph $H$ as a minor, then the treewidth of $G$ is $O(|V(H)|^c)$.
\end{corollary}

The Grid-Minor Theorem has several important applications in graph
theory and algorithms, and also in proving lower bounds. The quantitative
bounds in some of these applications can be directly improved by our
main theorem. We anticipate that there will be other applications for
our main theorem, and also for the algorithmic and graph-theoretic tools
that we develop here.

Our proof and algorithm are based on a combinatorial object, called a
path-of-sets system that we informally describe now; see
Figure~\ref{fig: path-set-system}. A path-of-sets system of width $w$
and length $\ell$ consists of a collection of $\ell$ disjoint sets of nodes
$S_1,\ldots,S_{\ell}$ together with collections of paths
$\pset_1,\ldots,\pset_{\ell-1}$ that are disjoint, which
connect the sets in a path-like fashion. The number of paths in each
set $\pset_i$ is $w$. Moreover, for each $i$, the induced graph $G[S_i]$
satisfies the following connectivity properties for the endpoints of the
paths $\pset_{i-1}$ and $\pset_{i}$ (sets $A_i$ and $B_i$ of vertices in the
figure): for every pair $A\subseteq A_i$, $B \subseteq B_i$ of vertex
subsets with $|A|=|B|$, there are $|A|$ node-disjoint paths connecting
$A$ to $B$ in $G[S_i]$.

\begin{figure}[h]
\begin{center}
  \scalebox{0.6}{\includegraphics{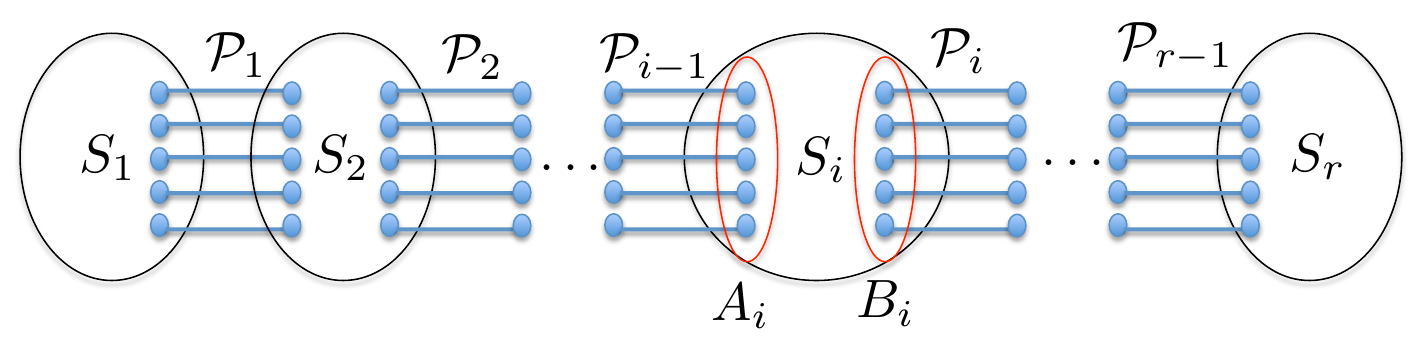}}\caption{A
    path-of-sets system of width $w$ and length $\ell$. Each set $\pset_i$ contains $w$ paths. All paths in
    $\bigcup_{i=1}^{\ell-1}\pset_i$ are node-disjoint and internally
    disjoint from $\bigcup_{i=1}^{\ell}S_i$. \label{fig:
      path-set-system}}
\end{center}
\end{figure}

Given a path-of-sets system of width $w$ and length $w$, we can
efficiently find a model of a grid minor of size $\left(\Omega(w^{1/2})\times
\Omega(w^{1/2})\right )$ in $G$, slightly strengthening a similar recent
result of Leaf and Seymour \cite{LeafS12}, who use a related
combinatorial object that they call a $(w,\ell)$-grill.
Our main contribution is to show that there is a randomized algorithm, that, given a graph $G$ of
treewidth $k$, with high probability constructs a path-of-sets system of
width $w$ and length $w$ in $G$, if $w^{c} \le O(k/\polylog k)$, where $c$ is a
fixed constant. The running time of the algorithm is polynomial in $|V(G)|$ and $k$. The central ideas for the construction build upon and
extend recent work on approximation algorithms for the Maximum Edge-Disjoint Paths
problem with constant congestion \cite{Chuzhoy11,ChuzhoyL12}, and
connections to treewidth \cite{ChekuriE13,ChekuriC13}.  In order to
construct the path-of-sets system, we use a closely related object,
called a tree-of-sets system. The definition of the tree-of-sets
system is very similar to the definition of the path-of-sets system,
except that, instead of connecting the clusters $S_i$ into a single
long path, we connect them into a tree whose maximum vertex degree is
at most $3$. We extend and strengthen the results of
\cite{Chuzhoy11,ChuzhoyL12,ChekuriE13}, by showing an efficient randomized 
algorithm, that, given a graph of treewidth $k$, with high probability constructs a large
tree-of-sets system. We then show how to construct a large
path-of-sets system, given a large tree-of-sets system.  We believe
that the tree-of-sets system is an interesting combinatorial object of
independent interest and hope that future work will yield simpler and
faster algorithms for constructing it, as well as improved parameters. This could lead to 
improvements in algorithms for related routing problems.

\paragraph{Subsequent Work.} 
Building on this work, the authors recently showed
in~\cite{tw-sparsifiers} an efficient randomized algorithm that, given
any graph $G$ of treewidth $k$, with high probability produces a
topological minor $H$ of $G$ (called a \emph{treewidth sparsifier}),
whose treewidth is $\Omega(k/\poly\log k)$, maximum vertex degree is
$3$, and $|V(H)|=O(k^4\poly\log k)$.

More recently, Chuzhoy~\cite{GMT2} has improved our bound on $\delta$
in Theorem~\ref{thm: main}, proving the theorem for $\delta=1/36$,
using a different construction of the path-of-sets system. By
combining some results and techniques from this work with this new
construction, she further improved the constant $\delta$ to
$1/19$. Her results use the treewidth sparsifier
from~\cite{tw-sparsifiers} as a starting point. We note that her proof
is non-constructive, and does not provide an algorithm to find a model
of the grid minor (though it is likely that it can be turned into an
algorithm whose running time is polynomial in $n$ and exponential in
$k$).

\label{-------------------------------------------------Prelims-------------------------------------}
\section{Preliminaries}
\label{sec: prelims}

In this paper we use the term ``efficient algorithm'' to refer
to a (possibly randomized)  algorithm that runs in time polynomial in
the length of its input.

All graphs in this paper are finite, and they do not have loops. We
say that a graph is simple to indicate that it does not have parallel
edges; otherwise, parallel edges are allowed.  Given a graph $G=(V,E)$
and a set $A\subseteq V$ of its vertices, we denote by $\out_G(A)$ the
set of all edges with exactly one endpoint in $A$ and by $E_G(A)$ the
set of all edges with both endpoints in $A$.  For disjoint sets of
vertices $A$ and $B$, the set of edges with one endpoint in $A$ and
the other in $B$ is denoted by $E_G(A,B)$.  For a vertex $v\in V$, we
denote the degree of $v$ by $d_G(v)$.  We may omit the subscript $G$
if it is clear from the context. Given a set $\pset$ of paths in $G$,
we denote by $V(\pset)$ the set of all vertices participating in paths
in $\pset$, and similarly $E(\pset)$ is the set of all edges that
participate in paths in $\pset$. We sometimes refer to sets of
vertices as \emph{clusters}. All logarithms are to the base of $2$.
We say that an event $\event$ holds with high probability, if the
probability of $\event$ is at least $1-1/n^c$ for some constant $c>1$,
where $n$ is the cardinality of vertex set of the graph in
question. We use the following simple claims several times.

\begin{claim}\label{claim: simple partition} 
  There is an efficient algorithm, that, given a set
  $\set{x_1,\ldots,x_n}$ of non-negative integers, with
  $\sum_{i}x_i=N$, and $x_i\leq 2N/3$ for all $i$, computes a
  partition $(A,B)$ of $\set{1,\ldots,n}$, such that $\sum_{i\in A}x_i
  \geq N/3$ and $\sum_{i\in B}x_i\geq N/3$.
\end{claim}

\begin{proof}
  We assume without loss of generality that $x_1\geq x_2\geq\cdots\geq x_n$, and process
  the integers in this order. When $x_i$ is processed, we add $i$ to
  $A$ if $\sum_{j\in A}x_j\leq \sum_{j\in B}x_j$, and we add it to $B$
  otherwise. We claim that at the end of this process, $\sum_{i\in
    A}x_i,\sum_{i\in B}x_i\geq N/3$. Indeed, 
  $1$ is always added to $A$. If if $x_1\geq N/3$, then, since $x_1\leq 2N/3$, it is easy to see
  that both subsets of integers sum up to at least $N/3$.  Otherwise,
  $|\sum_{i\in A}x_i-\sum_{i\in B}x_i|\leq \max_i\set{x_i}\leq x_1\leq
  N/3$.
\end{proof}

\begin{claim}
  \label{claim:path-leaves-in-tree}
  Let $T$ be a rooted tree, and $\ell,p\geq 1$ integers, such that $|V(T)| \ge \ell p$. Then either $T$ has at least $\ell$
  leaves, or there is a root-to-leaf path containing at least $p$ vertices in $T$.
\end{claim}

\begin{proof}
Suppose $T$ has fewer than $\ell$ leaves, and each root-to-leaf path
has fewer than $p$ vertices. Then, since every node belongs to some root-to-leaf
path of $T$, $|V(T)| < \ell p$, contradicting our assumption.
\end{proof}

The treewidth of a graph $G=(V,E)$ is typically defined via tree-decompositions.  A tree-decomposition of a graph $G$ consists of a tree
$T=(V(T),E(T))$ and a collection of vertex sets $\{X_v \subseteq V\}_{v \in
  V(T)}$ called \emph{bags}, such that the following two properties are
satisfied: (i) for each edge $(a,b) \in E$, there is some node $v \in
V(T)$ with both $a,b \in X_v$ and (ii) for each vertex $a \in V$, the
set of all nodes of $T$ whose bags contain $a$ induces a non-empty 
(connected) subtree of $T$. The {\em width} of a given tree-decomposition is
$\max_{v \in V(T)} |X_v| - 1$, and the \emph{treewidth} of a graph $G$,
denoted by $\tw(G)$, is the width of a minimum-width tree-decomposition of $G$.

We say that a simple graph $H$ is a \emph{minor} of a graph $G$, if $H$ can be
obtained from $G$ by a sequence of edge deletion, vertex deletion, and edge contraction
operations. 
Equivalently, a simple graph $H$ is a minor of $G$ if there is a map $\phi$, assigning to each vertex $v\in V(H)$ a subset $\phi(v)$ of vertices of $G$, and to each edge $e=(u,v)\in E(H)$ a path $\phi(e)$ connecting a vertex of $\phi(u)$ to a vertex of $\phi(v)$, such that:
\begin{itemize}
\item For every vertex $v\in V(H)$, the subgraph of $G$ induced by $\phi(v)$ is connected;
\item  If $u,v\in V(H)$ and $u\neq v$, then $\phi(u)\cap \phi(v)=\emptyset$; and
\item The paths in set $\set{\phi(e)\mid e\in E(H)}$ are internally node-disjoint, and they are internally disjoint from $\bigcup_{v\in V(H)}\phi(v)$.
\end{itemize}

A map $\phi$ satisfying these conditions is called \emph{a model of $H$ in $G$}. (We note that this definition is slightly different from the standard one, that requires that for each $e\in E(H)$, path $\phi(e)$ consists of a single edge; but it is immediate to verify that both definitions are equivalent, and it is more convenient for us to work with the above definition.) For convenience, we may sometimes refer to the map $\phi$ as the \emph{embedding} of $H$ into $G$, and specifically to $\phi(v)$ and $\phi(e)$ as the embeddings of the vertex $v\in V(H)$ and the edge $e\in E(H)$, respectively.

The $(g\times g)$-grid is a graph, whose vertex set is: $\set{v(i,j)\mid
  1\leq i,j\leq g}$.  The edge set consists of two subsets: a set of
\emph{horizontal edges} $E_1=\set{(v(i,j),v(i,j+1))\mid 1\leq i\leq g;
  1\leq j<g}$; and a set of \emph{vertical edges}
$E_2=\set{(v(i,j),v(i+1,j))\mid 1\leq i<g; 1\leq j\leq g}$. The
subgraph induced by $E_1$ consists of $g$ disjoint paths, that we
refer to as \emph{the rows of the grid}; the $i$th row is the row
incident with $v(i,1)$. Similarly, the subgraph induced by $E_2$
consists of $g$ disjoint paths, that we refer to as \emph{the columns
  of the grid}; the $j$th column is the column incident with $v(1,j)$.
We say that graph $G$ contains a $(g\times g)$-grid minor if some
minor $H$ of $G$ is isomorphic to the $(g\times g)$-grid.

\subsection{Flows and Cuts}
In this section we define standard single-commodity flows and discuss their relationships with the corresponding notions of cuts. Most definitions and results from this section can be found in standard textbooks; we refer the reader to~\cite{Schrijver} for more details.

Let $G=(V,E)$ be an edge-capacitated graph with $c(e)> 0$ denoting the
capacity of edge $e \in E$. Given two disjoint vertex subsets
$S,T\subseteq V$, let $\pset$ be the set of all paths that start at
$S$ and terminate at $T$. An $S$--$T$ flow $f: \pset \rightarrow
\mathbb{R}_+$ is an assignment of non-negative values to paths in
$\pset$. The \emph{value} of the flow is $\sum_{P\in
  \pset}f(P)$. Given a flow $f$, for each edge $e\in E$, we define a
flow through $e$ to be: $f'(e)=\sum_{P \in \pset: e \in P} f(P)$.  The
\emph{edge-congestion} of the flow is $\max_{e\in E}\set{f'(e)/c(e)}$.
We say that the flow $f$ is \emph{valid}, or that it causes no
edge-congestion, if its edge-congestion is at most $1$.  We note that
even though $|\pset|$ may be exponential in $|V|$, there are known
efficient algorithms to compute a valid flow of a specified value $F$
(if it exists), and to compute a flow of maximum value. Moreover, in
both cases, the number of paths in $\pset$ with non-zero flow value
$f(P)$ is guaranteed to be at most $|E|$. Such flows can be computed,
for example, by using an equivalent edge-based flow formulation
together with Linear Programming, and a flow-path decomposition of the
resulting solution (see~\cite{Schrijver} for more details). It is also
well known that if all edge capacities are integral, then whenever a
valid $S$--$T$ flow of an integral value $F$ exists in $G$, there is
also a valid $S$--$T$ flow $\tilde f$ of the same value, where $\tilde
f(P)$ is integral for all $P \in \pset$, and the number of paths $P$
with $\tilde f(P) > 0$ is at most $|E|$. Moreover, such a flow can be
found efficiently. Throughout the paper, whenever the edge capacities
of a given graph $G$ are not specified, we assume that they are all
unit.

A \emph{cut} in a graph $G$ is a bipartition $(A,B)$ of its vertices,
with $A,B\neq \emptyset$. We sometimes use $\overline{A}$ to denote
$V\setminus A$. The \emph{value} of the cut is the total capacity of
all edges in $E(A,B)$ (if the edge capacities of $G$ are not
specified, then the value of the cut is $|E(A,B)|$). We say that a cut
$(A,B)$ \emph{separates} $S$ from $T$ if $S\subseteq A$ and
$T\subseteq B$. The well-known max-flow min-cut theorem states that
for a graph $G$ and disjoint vertex subsets $S$, $T$, the value of
the maximum $S$--$T$ flow in $G$ is equal to the value of the minimum
cut separating $S$ from $T$ in $G$. Notice that if all edges of $G$
have unit capacities, and the value of the maximum flow from $S$ to
$T$ is $F$, then the maximum number of edge-disjoint paths connecting
the vertices of $S$ to the vertices of $T$ is also $F$, and if $E'$ is
a minimum-cardinality set of edges, such that $G\setminus E'$ contains
no path connecting a vertex of $S$ to a vertex of $T$, then $|E'|=F$.
When $S=\set{s}$ and $T=\set{t}$, then we sometimes refer to the
$S$-$T$ flow and $S$-$T$ cut as $s$-$t$ flow and $s$-$t$ cut
respectively.

Given a subset $\pset'\subseteq \pset$ of paths connecting vertices of
$S$ to vertices of $T$ in $G$, we say that the paths in $\pset'$ cause
edge-congestion at most $\eta$, if for every edge $e\in E$, the total
number of paths in $\pset'$ containing $e$ is at most $\eta \cdot c(e)$.

A variant of the $S$--$T$ flow that we sometimes use is when the
capacities are given on the graph vertices and not edges. Such a flow
$f$ is defined exactly as before, except that now, for every vertex
$v\in V$, we let $f'(v)=\sum_{\stackrel{P\in \pset:}{v\in P}}f(P)$,
and we define the congestion of the flow to be $\max_{v\in
  V}\set{f'(v)/c(v)}$. If the congestion of the flow is at most $1$,
then we say that it is a valid flow, or that the flow causes no
vertex-congestion. When all vertex capacities are integral, there is a
maximum flow $f$, such that all values $f(P)$ for all $P\in \pset$ are
integral. In particular, if all vertex-capacities are $1$, and there
is a valid $S$--$T$ flow of value $F$, then there are $F$
node-disjoint paths connecting vertices of $S$ to vertices of $T$, and
this set of paths can be found efficiently.

All the definitions and results about single-commodity flows mentioned
above carry over to directed graphs as well, except that cuts are
defined slightly differently. As before, a cut in $G$ is a
bipartition $(A,B)$ of the vertices of $G$. The value of the cut is
the total capacity of edges connecting vertices of $A$ to vertices of
$B$. The max-flow min-cut theorem remains valid in directed graphs,
with this definition of cuts. For every directed flow network, there
exists a maximum $S$--$T$ flow, in which for every pair $(e,e')$ of
anti-parallel edges, at most one of these edges carries non-zero flow;
if all edge capacities are integral, then there is a maximum flow that
is integral and has this property. This follows from the equivalent
edge-based definition of flows.  Flows in directed graphs with
capacities on vertices are defined similarly.

We will repeatedly use the following simple claim.

\begin{claim}\label{claim: large matching}
There is an efficient algorithm, that, given a bipartite graph $G=(V_1,V_2,E)$ with maximum vertex degree at most $\Delta$, computes a matching $E'\subseteq E$ of cardinality at least $|E|/\Delta$.
\end{claim}

\begin{proof}
We set up a directed flow network: start with graph $G$, assign each of its vertices capacity $1$ and direct its edges from $V_1$ to $V_2$. Add a source $s$ of infinite capacity that connects to every vertex in $V_1$ with a directed edge, and add a destination vertex $t$ of infinite capacity to which every vertex of $V_2$ connects with a directed edge. It is immediate to see that this network has a valid $s$-$t$ flow of value $|E|/\Delta$, by sending $1/\Delta$ flow units on each edge $e\in E$. From the integrality of flow, there is a valid integral flow of the same value, which defines the desired matching.
\end{proof}

\subsection {Sparsest Cut} 
Suppose we are given a graph $G=(V,E)$ and a subset $\tset\sse V$ of
$k$ vertices, called terminals. Given a cut $(S,\nots)$ in $G$ with
$S\cap \tset,\nots\cap \tset\neq \emptyset$, the \emph{sparsity} of
$(S,\nots)$ is
$\Phi_{\tset}(S,\nots)=\frac{|E(S,\nots)|}{\min\set{|S\cap
    \tset|,|\nots\cap \tset|}}$, and the value of the sparsest cut in
$G$ with respect to $\tset$ is:
$\Phi_{\tset}(G)=\min_{\stackrel{S\subseteq V:}{S\cap \tset,\nots\cap
    \tset\neq \emptyset}}\set{\Phi_{\tset}(S,\nots)}$.  In the
sparsest cut problem, the input is a graph $G$ with a set $\tset$ of
terminals, and the goal is to find a cut of minimum sparsity. Arora,
Rao and Vazirani~\cite{ARV} have shown an $O(\sqrt {\log
  k})$-approximation algorithm for the sparsest cut problem, where
$k=|\tset|$.  We use $\algsc$ to refer to their algorithm, and we denote by
$\alphasc(k)=O(\sqrt{\log k})$ its approximation factor.
%
%
%
We will repeatedly use the following observation.

\begin{observation}\label{obs: sparsest cut to flow}
Let $G$ be a graph, and $\tset \subseteq V(G)$ a subset of its vertices called terminals, where $|\tset|=k$ for some $k>0$. Assume further that for some $0<\alpha\leq 1$, $\Phi_{\tset}(G)\geq \alpha$. Then for every pair $\tset',\tset''\subseteq \tset$ of disjoint equal-sized subsets of terminals, there is a flow $f$ in $G$, where every terminal in $\tset'$ sends one flow unit, every terminal in $\tset''$ receives one flow unit, and the edge-congestion is bounded by $1/\alpha$.
\end{observation}

\begin{proof}
Let $\tset',\tset''\subseteq \tset$ be a pair of disjoint equal-sized sets of terminals.
We construct a directed flow network $H$ from $G$, by replacing each edge of $G$ with a pair of bi-directed edges, and setting the capacity $c(e)$ of each such edge $e$ to be $1/\alpha$. We then add two special vertices to the graph: the source $s$, that connects with a capacity-$1$ edge to every vertex of $\tset'$, and the destination $t$, to which every vertex of $\tset''$ connects with a capacity-$1$ edge.  Let $k'=|\tset'|=|\tset''|$, and let $f$ be the maximum $s$-$t$ flow in $H$. If the value of $f$ is at least $k'$, then we can use $f$ to define a flow $f'$, where  every terminal in $\tset'$ sends one flow unit, every terminal in $\tset''$ receives one flow unit, and the edge-congestion is bounded by $1/\alpha$ (as we can assume without loss of generality that for every pair $e',e''$ of anti-parallel edges, only one of these edges carries non-zero flow). Therefore, we assume from now on that the value of $f$ is less than $k'$. We will reach a contradiction by showing a cut whose sparsity with respect to $\tset$ is less than $\alpha$.

Let $(A',B')$ be the minimum $s$-$t$ cut in $H$, and let $E'$ be the set of all edges of $H$ from $A'$ to $B'$, so $\sum_{e\in E'}c(e)<k'$. Let $A=A'\setminus\set{s}$ and $B=B'\setminus\set{t}$, and  assume that $|\tset\cap A|\leq |\tset\cap B|$ - the other case is symmetric. Let $k_1=|\tset'\cap A'|$ and $k_2=|\tset'\cap B'|$. Then $\sum_{e\in E'}c(e)\geq k_2+|E_G(A,B)|/\alpha$. Therefore, $|E_G(A,B)|< \alpha (k'-k_2)= \alpha k_1= \alpha |\tset'\cap A|\leq \alpha|\tset\cap A|$, and the cut $(A,B)$ has sparsity less than $\alpha$, a contradiction.
\end{proof}

\subsection{Linkedness and Well-Linkedness}
We define the notion of linkedness and the different notions of
well-linkedness that we use.

\begin{definition}
  We say that a set $\tset$ of vertices is
  $\alpha$-well-linked\footnote{This notion of well-linkedness is
    based on edge-cuts and we distinguish it from node-well-linkedness
    that is directly related to treewidth. For technical reasons it is
    easier to work with edge-cuts and hence we use the term well-linked to mean
    edge-well-linkedness, and explicitly use the term
    node-well-linkedness when necessary.} in $G$, if for every
  partition $(A,B)$ of the vertices of $G$ into two subsets,
  $|E(A,B)|\geq \alpha\cdot \min\set{|A\cap \tset|,|B\cap \tset|}$.
\end{definition}

The following simple observation immediately follows from the definition of well-linkedness.

\begin{observation}\label{obs: wl-properties}
Let $G$ be a graph and $\tset\subseteq V(G)$ a subset of its vertices, so that $\tset$ is $\alpha$-well-linked in
$G$, for some $0<\alpha\leq 1$. Then:

\begin{itemize}
\item $\Phi_{\tset}(G)\geq \alpha$;

\item for every subset $\tset'\subseteq\tset$, $\tset'$ is $\alpha$-well-linked in $G$; and

\item $\tset$ is $\alpha'$-well-linked in $G$ for all $0<\alpha'<\alpha$.
\end{itemize}
\end{observation}

\begin{definition}
  We say that a set $\tset$ of vertices is \emph{node-well-linked} in
  $G$, if for every pair $(\tset_1,\tset_2)$ of equal-sized subsets of
  $\tset$, there is a collection $\pset$ of $|\tset_1|$ {\bf
    node-disjoint} paths, connecting the vertices of $\tset_1$ to the
  vertices of $\tset_2$. (Note that $\tset_1$, $\tset_2$ are not
  necessarily disjoint, and we allow paths consisting of a single vertex).
\end{definition}

\begin{definition}
  We say that two disjoint vertex subsets $A,B$ are \emph{linked} in
  $G$ if for every pair of equal-sized subsets $A'\subseteq A$,
  $B'\subseteq B$ there is a set $\pset$ of $|A'|=|B'|$ node-disjoint
  paths connecting $A'$ to $B'$ in $G$.
\end{definition}

Our algorithm starts with a graph $G$ of treewidth $k$, and then reduces its degree to $\poly\log(k)$, while preserving the treewidth to within a factor of $\poly\log(k)$. As we show below, in bounded-degree graphs, the notions of edge- and node-well-linkedness are closely related to each other, and we exploit this connection throughout the algorithm.

\begin{theorem}\label{thm: linkedness from node-well-linkedness}
  Suppose we are given a graph $G$ with maximum vertex degree at most $\Delta$, and two disjoint subsets $\tset_1,\tset_2$ of
  its vertices, such that
  $\tset_1\cup \tset_2$ is $\alpha$-well-linked in $G$ for some $0<\alpha\leq 1$, and each one
  of the sets $\tset_1,\tset_2$ is node-well-linked in $G$. Let
  $\tset'_1\subsetneq \tset_1$, $\tset_2'\subsetneq\tset_2$, be a pair of
  subsets with $|\tset_1'|\leq \frac{\alpha |\tset_1|}{2\Delta}$ and $|\tset_2'|\leq \frac{\alpha |\tset_2|}{2\Delta}$. Then $\tset'_1$ and $\tset_2'$ are
  linked in $G$.
\end{theorem}

\begin{proof}
Let $\tset=\tset_1\cup \tset_2$. We refer to the vertices of $\tset$ as terminals. Denote $|\tset_1|=\kappa_1$, $|\tset_2|=\kappa_2$, and assume without loss of generality that $\kappa_1\leq \kappa_2$.
Assume for contradiction that $\tset'_1$ and $\tset_2'$ are not linked in $G$. Then there are two sets $A\subseteq\tset_1',B\subseteq \tset_2'$, with $|A|=|B|=\kappa'$ for some $\kappa'\leq \frac{\alpha\kappa_1}{2\Delta}$, and a set $S$ of $\kappa'-1$ vertices, separating $A$ from $B$ in $G$. 

Let $A'\subseteq \tset_1$ be the set of all terminals $t\in \tset_1$, such that $t$ lies in the same component of $G\setminus S$ as some vertex of $A$. We claim that $|A'|\geq \kappa_1-\kappa'$. Indeed, assume otherwise, and let $A''\subseteq \tset_1\setminus A'$ be a set of $\kappa'$ vertices. Since $\tset_1$ is node-well-linked in $G$, there is a set $\pset$ of $\kappa'$ node-disjoint paths, connecting the vertices of $A$ to the vertices of $A''$ in $G$. At most $\kappa'-1$ of these paths may contain the vertices of $S$, and so at least one vertex of $\tset_1\setminus A'$ is connected to some vertex of $A$ in $G\setminus S$, a contradiction.

Similarly, we let $B'\subseteq \tset_2$ be the set of all terminals $t\in \tset_2$, such that $t$ lies in the same component of $G\setminus S$ as some vertex of $B$. From the same arguments as above, $|B'|\geq \kappa_2-\kappa'$. Finally, we show that there is some pair $a\in A'$, $b\in B'$ of vertices that lie in the same connected component of $G\setminus S$. Indeed, since the terminals of $\tset$ are $\alpha$-well-linked in $G$, there is a set $\qset$ of at least $\kappa_1-\kappa'$ paths in $G$, where each path originates at a distinct vertex of $A'$ and terminates at a distinct vertex of $B'$, and every edge of $G$ participates in at most $1/\alpha$ paths. At most $(\kappa'-1)\Delta/\alpha$ of the paths in $\qset$ may contain the vertices of $S$. Since $|\qset|=\kappa_1-\kappa'>(\kappa'-1)\Delta/\alpha$, at least one path of $\qset$ belongs to $G\setminus S$. Therefore, there is a path in $G\setminus S$ from a vertex of $A$ to a vertex of $B$, a contradiction.
%
%
%
\end{proof}

\subsection{Boosting Well-Linkedness}
\label{subsec:boosting}
Suppose we are given a graph $G$ and a set $\tset$ of vertices of $G$
called terminals, where $\tset$ is $\alpha$-well-linked in
$G$. Boosting theorems allow us to boost the well-linkedness by
selecting an appropriate subset of the terminals, whose well-linkedness is greater than $\alpha$. 
We start with the following simple claim, that has been extensively used in past work to boost well-linkedness of terminals.

\begin{claim}\label{claim: better well-linkedness for trees}
  Suppose we are given a graph $G$ and a set $\tset$ of vertices of
  $G$, called terminals, such that $\tset$ is $\alpha$-well-linked for
  some $0<\alpha<1$. Assume further that we are given a collection
  $\sset$ of trees in $G$, and for every tree $T\in \sset$ we are
  given a subset $\lambda_{T}\subseteq V(T)\cap \tset$ of at least
  $\ceil{1/\alpha}$ terminals, such that for every pair $T\neq T'$ of
  the trees, $\lambda_{T}\cap \lambda_{T'}=\emptyset$. Assume further
  that each edge of $G$ belongs to at most $c$ trees, and let
  $\tset'\subseteq \tset$ be a subset of terminals, containing exactly
  one terminal from each set $\lambda_{T}$ for $T\in \sset$. Then
  $\tset'$ is $1/(c+1)$-well-linked in $G$.
\end{claim}
\begin{proof}
  The proof provided here was suggested by an anonymous referee, and
  it is somewhat simpler than our original proof.  Let $(A,B)$ be a
  partition of the vertices of $G$, and let $\tset_A=\tset'\cap A$,
  $\tset_B=\tset'\cap B$. Assume without loss of generality that
  $|\tset_A|\leq |\tset_B|$ and denote $|\tset_A|=\kappa$. Our goal is
  to show that $|E(A,B)|\geq \kappa/(c+1)$. Assume for contradiction
  that $|E(A,B)|< \kappa/(c+1)$.

Let $\sset_1\subseteq \sset$ be the set of trees $T$ with $E(T)\cap E(A,B)\neq \emptyset$. Since each edge of $G$ belongs to at most $c$ trees, $|\sset_1|<\kappa c/(c+1)$. Let $\sset_A\subseteq \sset$ be the set of all trees $T$ with $V(T)\cap \tset_A\neq \emptyset$, and define $\sset_B\subseteq \sset$ similarly for $\tset_B$. Then $|\sset_A|,|\sset_B|\geq \kappa$. Let $\sset'_A=\sset_A\setminus \sset_1$. Then $|\sset'_A|> \kappa-\kappa c/(c+1)=\kappa/(c+1)$, and every tree in $\sset_A'$ is contained in $A$. Similarly, let $\sset'_B=\sset_B\setminus \sset_1$, so $|\sset'_B|>\kappa/(c+1)$, and every tree in $\sset_B'$ is contained in $B$.

Let $\tset_A'\subseteq \tset\cap A$ be the set of all terminals participating in the trees of $\sset_A'$, and define $\tset_B'\subseteq \tset\cap B$ similarly for $\sset_B'$. Then $|\tset_A'|\geq \ceil{1/\alpha}\cdot |\sset'_1|\geq \frac{\kappa}{\alpha (c+1)}$, and similarly $|\tset_B'|\geq \frac{\kappa}{\alpha( c+1)}$. From the $\alpha$-well-linkedness of the terminals in $\tset$, $|E(A,B)|\geq \alpha\cdot \min\set {|\tset_A'|,|\tset_B'|}\geq \kappa/(c+1)$ must hold, a contradiction.
\end{proof}

This claim is already sufficient to boost the well-linkedness of a given set of terminals, as follows.

\begin{corollary}\label{cor: weak boosting w subsets}
There is an efficient algorithm, that, given a connected graph $G$,  a subset $\tset$ of its vertices, such that for some $0<\alpha\leq 1$, $\tset$ is $\alpha$-well-linked in $G$, and a partition $\tset_1,\ldots,\tset_{\ell}$ of $\tset$, computes,  for each $1\leq i\leq \ell$ a subset $\tset_i'\subseteq \tset_i$ of at least $\floor{\frac{|\tset_i|}{ 3\ceil{1/\alpha}}}$ vertices, so that $\bigcup_{i=1}^{\ell}\tset_i'$ is $1/2$-well-linked in $G$.
\end{corollary}

\begin{proof}
Throughout the proof, we refer to the vertices of $\tset$ as terminals. For $1\leq i\leq \ell$, denote $|\tset_i|=\kappa_i$. We start with the following simple observation.

\begin{observation} There is an efficient algorithm to compute a collection $\fset$ of trees in $G$, and for every tree $T\in \fset$, a subset $\lambda_T\subseteq V(T)\cap \tset$ of its vertices, such that:

\begin{itemize}
\item every edge of $G$ belongs to at most one tree;
\item for every tree $T\in \fset$, $\ceil{1/\alpha}\leq |\lambda_T|\leq 3\ceil{1/\alpha}$; and
\item the sets $\set{\lambda_T}_{T\in \fset}$ define a partition of $\tset$.
\end{itemize}
\end{observation}

\begin{proof}
Let $T^*$ be a spanning tree of $G$, that we root at some vertex $r$. We perform a number of iterations, where in every iteration we delete some edges and vertices from $T^*$. For each vertex $v$ of the tree $T^*$, let $T^*_v$ denote the sub-tree rooted at $v$, and let $w(T^*_v)$ denote the total number of terminals in $T^*_v$. We build the set $\fset$ of the trees gradually. At the beginning, $\fset=\emptyset$. While $w(T^*_r)>3\ceil{1/\alpha}$, we perform the following iteration:

\begin{itemize}
\item Let $v$ be the lowest vertex in the tree $T^*$, such that $w(T^*_v)>\ceil{1/\alpha}$.
\item If $w(T^*_v)\leq 2\ceil{1/\alpha}$, then we add the tree $T^*_v$ to $\fset$, set $\lambda_{T^*_v}=V(T^*_v)\cap \tset$, and delete all vertices and edges of $T^*_v$ from the tree $T^*$.

\item Otherwise, let $u_1,\ldots,u_k$ be the children of $v$, and let $j$ be the smallest index, such that $\sum_{i=1}^jw(T^*_{u_i})\geq \ceil{1/\alpha}$. We add a new tree $T'$ to $\fset$ --- a subtree of $T^*$ induced by $\set{v}\cup (\bigcup_{i=1}^jV(T^*_{u_i}))$, setting $\lambda_{T'}=\bigcup_{i=1}^j (V(T^*_{u_i})\cap \tset)$. We delete all edges of $T'$, and all vertices of $V(T')\setminus\set{v}$ from  the tree $T^*$.
\end{itemize}

Notice that since at the beginning of the current iteration $w(T^*_r)>3\ceil{1/\alpha}$,  at the end of the current iteration, $w(T^*_r)>\ceil{1/\alpha}$ must hold. In the last iteration, when $w(T^*_r)\leq 3\ceil{1/\alpha}$, we add the tree $T^*_r$ to $\fset$ and set $\lambda_{T^*_r}=V(T^*_r)\cap \tset$. It is easy to verify that all conditions of the observation hold for the final collection $\fset$ of trees.
\end{proof}

Next, we show that we can select at most one terminal from each set $\lambda_T$, for $T\in \fset$, such that enough terminals from every subset $\tset_i$ is selected.

\begin{observation}\label{obs: select the terminals}
There is an efficient algorithm that computes, for each $1\leq i\leq  \ell$, a subset $\tset_i'\subseteq \tset_i$ of at least $\floor{\frac{\kappa_i}{ 3\ceil{1/\alpha}}}$ vertices, so that, if we denote $\tset'=\bigcup_{i=1}^{\ell}\tset'_i$, then for every tree $T\in \fset$, $|\lambda_T\cap \tset'|\leq 1$.
\end{observation}

From Claim~\ref{claim: better well-linkedness for trees}, the resulting set $\tset'$ of terminals is $1/2$-well-linked in $G$. It now remains to prove Observation~\ref{obs: select the terminals}.

\begin{proof}
We build a node-capacitated directed flow network $\nset$, as follows. We start from a source vertex $s$ and a destination vertex $s'$ that have infinite capacity. We then add $\ell$ vertices $u_1,\ldots,u_{\ell}$, each of capacity $\kappa_i'=\frac{\kappa_i}{ 3\ceil{1/\alpha}}$, and connect $s$ to each of these vertices. Each vertex $u_i$ will represent the set $\tset_i$ of the terminals. For each terminal $t\in \tset$, we add a unit-capacity vertex $v_t$ to $\nset$, and, if $t\in \tset_i$, then we connect $u_i$ to $v_t$ with a directed edge.

For every tree $T\in \fset$, we add a unit-capacity vertex $x_{T}$, that connects to the destination vertex $s'$ with a directed edge. Finally, for every tree $T\in \fset$, and for every terminal $t\in \lambda_T$, we add a directed edge $(v_t,x_{T})$. We claim that there is a valid flow of value $\sum_{i=1}^{\ell}\kappa_i'$ from $s$ to $s'$ in $\nset$. Indeed, consider a directed $s$-$s'$ path, and assume that the path is $(s,u_i,v_t,x_{T},s')$. We send $\frac{1}{ 3\ceil{1/\alpha}}$ flow units along this path. Since for all $T\in \fset$, $|\lambda_T|\leq  3\ceil{1/\alpha}$, we obtain a valid $s$-$s'$ flow of value $\sum_{i=1}^{\ell}\kappa_i'$. If we reduce the capacity of every vertex $u_i$, for $1\leq i\leq \ell$, to $\floor{\kappa_i'}$, we can still obtain a valid $s$-$s'$ flow of value   $\sum_{i=1}^{\ell}\floor{\kappa_i'}$, by appropriately reducing flows on some paths. Since all vertex capacities are now integral, there is an integral flow $f$ of the same value. We are now ready to define the set $\tset'_i$ of terminals for each $1\leq i\leq \ell$: it contains all terminals $t\in \tset_i$, such that the edge $(s_i,v_t)$ carries one flow unit in $f$. It is immediate to verify that $|\tset'_i|=\floor{\kappa_i'}\geq \floor{\frac{\kappa_i}{ 3\ceil{1/\alpha}}}$, and since the capacities of all vertices $\set{x_{T'}}_{T'\in \fset}$ are unit, if we denote by $\tset'=\bigcup_{i=1}^{\ell}\tset'_i$, then for every tree $T\in \fset$, $|\lambda_T\cap \tset'|\leq 1$.
\end{proof}
\end{proof}

The above claim gives a way to boost the well-linkedness of a given
set $\tset$ of terminals to $1/2$-well-linkedness. This type of
argument has been used before extensively, usually under the name of
the ``grouping technique''~\cite{ANF,CKS, RaoZhou, Andrews,
  Chuzhoy11}.  However, we need a stronger result: given a set $\tset$
of terminals, that are $\alpha$-well-linked in $G$, we would like to
find a large subset $\tset'\subseteq \tset$, such that $\tset'$ is
{\bf node-well-linked} in $G$. The following theorem allows us to
achieve this, generalizing a similar theorem for edge-disjoint routing
in~\cite{ANF}. The proof\footnote{Some of our theorems on well-linked
  sets including this one appear to have alternate proofs via tangles
  and related matroids from graph minor theory \cite{graphminorsx};
  this was suggested to us by a reviewer. However, it is unclear
  whether the alternate proofs yield polynomial-time algorithms.}
appears in the Appendix.

\begin{theorem}\label{thm: grouping}
  Suppose we are given a connected graph $G=(V,E)$ with maximum vertex
  degree at most $\Delta$, where $\Delta\geq 3$, and a subset $\tset$
  of $\kappa$ vertices called terminals, such that $\tset$ is
  $\alpha$-well-linked in $G$, for some $0<\alpha\leq 1$. Then there
  is a subset $\tset'\subseteq \tset$ of $\ceil{\frac{3\alpha\kappa
    }{10\Delta}}$ terminals, such that $\tset'$ is node-well-linked in
  $G$. Moreover, there is an algorithm whose running time is
  polynomial in $|V|$ and $\kappa$, that computes a subset
  $\tset'\subseteq \tset$ of at least
  $\frac{\alpha}{32\Delta^4\alphasc(\kappa)} \cdot \kappa$ terminals,
  such that $\tset'$ is node-well-linked in
  $G$. 
\end{theorem}

\begin{corollary}\label{cor: ultimate boosting}
There is an efficient algorithm, that, given a connected graph $G=(V,E)$ with maximum vertex degree $\Delta\geq 3$, a subset $\tset$ of $\kappa$ vertices of $G$ that are $\alpha$-well-linked in $G$ for some $0<\alpha\leq 1$, and a partition $\tset_1,\ldots,\tset_{\ell}$ of $\tset$ into disjoint subsets, computes, for each $1\leq i\leq \ell$, a subset $\tset'_i\subseteq \tset_i$ of $\floor{\frac{\alpha \cdot |\tset_i|}{2^{10}\Delta^5\alphasc(\kappa)}}$ terminals, such that: (i) for all $1\leq i\leq \ell$, set $\tset'_i$ is node-well-linked in $G$; (ii) $\bigcup_{i=1}^{\ell}\tset'_i$ is $1/2$-well-linked in $G[S_i]$; and (iii) for all $1\leq i<j\leq \ell$, $\tset'_i$ and $\tset'_j$ are linked in $G$.
\end{corollary}
\begin{proof}
For all $1\leq i\leq \ell$, let $\kappa_i=|\tset_i|$. We use Corollary~\ref{cor: weak boosting w subsets} to compute, for each $1\leq i\leq \ell$, a subset $\tset^1_i\subseteq \tset_i$ of at least $\floor{\frac{\kappa_i}{3\ceil{1/\alpha}}}\geq \floor{\frac{\alpha \kappa_i}{8}}$ terminals, such that $\bigcup_{i=1}^{\ell}\tset^1_i$ is $1/2$-well-linked in $G$.

Next, for each $1\leq i\leq \ell$, we apply Theorem~\ref{thm: grouping} to $\tset^1_i$, to compute a subset $\tset^2_i\subseteq \tset^1_i$ of at least $\frac{ |\tset^1_i|}{32\Delta^4\alphasc(\kappa_i)}\geq \frac{\alpha \kappa_i}{256\Delta^4\alphasc(\kappa)}$ terminals, so that $\tset^2_i$ is node-well-linked in $G$. We then let $\tset'_i\subseteq \tset^2_i$ be a subset of $\floor{\frac{|\tset^2_i|}{4\Delta}}\geq \floor{\frac{\alpha\kappa_i}{2^{10}\Delta^5\alphasc(\kappa)}}$ terminals. Since the terminals of $\bigcup_{j=1}^{\ell}\tset^1_j$ are $1/2$-well-linked in $G$, from Theorem~\ref{thm: linkedness from node-well-linkedness}, for all $1\leq j<j'\leq \ell$, $\tset_j'$ and $\tset_{j'}'$ are linked in $G$.
\end{proof}

\subsection{Treewidth and Well-Linkedness}

The following lemma summarizes an important connection between the
graph treewidth, and the size of the largest node-well-linked set of
vertices.

\begin{lemma}\cite{Reed-chapter}
  \label{lem:tw-wl}
  Let $k$ be the size of the largest node-well-linked vertex set in $G$. Then
  $\frac{k} 4 -1 \le \tw(G) \le k-1$.
\end{lemma}

Lemma~\ref{lem:tw-wl} guarantees that a graph $G$ of treewidth $k$
contains a set $X$ of $\Omega(k)$ vertices, that is node-well-linked
in $G$. Kreutzer and Tazari~\cite{KreutzerT10} give a constructive
version of this lemma, obtaining a set $X$ with slightly weaker
properties. Lemma~\ref{lem: find well-linked set} below rephrases, in
terms convenient to us, Lemma 3.7 in~\cite{KreutzerT10}\footnote{Lemma~\ref{lem: find well-linked set} is slightly weaker than what was shown in~\cite{KreutzerT10}. We use it since it suffices for our purposes and avoids the introduction of additional notation.}.
 
 \begin{lemma}\label{lem: find well-linked set} There is an
   efficient algorithm, that, given a graph $G$ of treewidth $k$,
   finds a set $X$ of $\Omega(k)$ vertices, such that $X$ is
   $\alpha^*=\Omega(1/\log k)$-well-linked in $G$ and $|X|$ is even. Moreover, for
   every partition $(X_1,X_2)$ of $X$ into two equal-sized subsets,
   there is a collection $\pset$ of paths connecting every vertex of
   $X_1$ to a distinct vertex of $X_2$, such that every vertex of $G$
   participates in at most $1/\alpha^*$ paths in $\pset$.
 \end{lemma}


\subsection{A Tree with Many Leaves or a Long $2$-Path}
Suppose we are given a connected $n$-vertex graph $Z$. A path
$P$ in $Z$ is called a \emph{$2$-path} if every vertex $v\in P$ has degree $2$ in $Z$. The following theorem, due to Leaf and
Seymour~\cite{LeafS12} states that we can find either a spanning tree
with many leaves or a long $2$-path in $Z$. For completeness, the
proof appears in the Appendix.

\begin{theorem}\label{thm: many leaves or a long 2-path}
  There is an efficient algorithm, that, given a connected $n$-vertex graph $Z$, and integers $L\geq 1,p\geq 1$ 
  with $\frac n {2L}\geq p+5$, either finds a spanning tree $T$ with at least $L$
  leaves in $Z$, or a $2$-path containing at least $p$ vertices in $Z$.
\end{theorem}


\subsection{Re-Routing Two Sets of Disjoint Paths}
Suppose we are given a directed graph $\hat G$, a set $U\subseteq V(\hat G)$ of its vertices, and an additional vertex $s\in V(\hat G)\setminus U$. A set $\xset$ of directed paths that originate at the vertices of $U$ and terminate at $s$ is called a set of $U$-$s$ paths. We say that the paths in $\xset$ are \emph{nearly disjoint}, if except for vertex $s$ they do not share other vertices.
We need the following lemma, that was proved by Conforti, Hassin and Ravi~\cite{CHR}. We provide a simpler proof, suggested to us by Paul Seymour \cite{PS-comm} in the Appendix.

\begin{lemma}
\label{lemma: re-routing of vertex-disjoint paths} 
There is an efficient algorithm, that, given a directed graph $\hat G$, two subsets $U_1,U_2$ of its vertices, and an additional vertex $s\in V(\hat G)\setminus(U_1\cup U_2)$, together with a set $\xset_1$ of $\ell_1$ nearly disjoint $U_1$-$s$ paths and a set $\xset_2$ of $\ell_2$ nearly disjoint $U_2$-$s$ paths in $\hat G$, where $\ell_1>\ell_2\geq 1$, finds a set $\xset'$ of $\ell_1$ nearly-disjoint $(U_1\cup U_2)$-$s$ paths, and  a partition $(\xset'_1,\xset_2')$ of $\xset'$, such that $|\xset_2'|=\ell_2$, the paths of $\xset_2'$ originate from $U_2$, and $\xset_1'\subseteq \xset_1$.\end{lemma}

\subsection{Cut-Matching Game and Degree Reduction}
We say that a graph $G=(V,E)$ is an $\alpha$-expander, if
$\min_{\stackrel{S\sse V:}{0<|S|\leq |V|/2}}\set{\frac{|E(S,\nots)|}{|S|}}\geq \alpha$. Equivalently, $G$ is an $\alpha$-expander if $\Phi_V(G)\geq \alpha$.

We use the cut-matching game of Khandekar, Rao and Vazirani~\cite{KRV}. In this game, we are given a set $V$ of $N$ vertices, where $N$ is even, and two players: a cut player, whose goal is to construct an expander $X$ on the set $V$ of vertices, and a matching player, whose goal is to delay its construction. The game is played in iterations. We start with the graph $X$ containing the set $V$ of vertices, and no edges.
In each iteration $j$, the cut player computes a bipartition $(A_j,B_j)$ of $V$ into two equal-sized sets, and the matching player returns some perfect matching $M_j$ between the two sets. The edges of $M_j$ are then added to $X$. Khandekar, Rao and Vazirani have shown that there is a strategy for the cut player, guaranteeing that after $O(\log^2N)$ iterations we obtain a $\half$-expander with high probability. Subsequently, Orecchia et al.~\cite{better-CMG} have shown the following improved bound:


\begin{theorem}[\cite{better-CMG}]\label{thm: CMG}
There is a randomized algorithm for the cut player, such that, no matter how the matching player plays, after $\gkrv(N)=O(\log^2N)$ iterations, graph $X$ is an $\alphaKRV(N)=\Omega(\log N)$-expander, with constant probability.
\end{theorem}

\subsection{Starting Point}
Let $G$ be a graph with $\tw(G) = k$. The proof of Theorem~\ref{thm:
  main} uses the notion of edge-well-linkedness as well as
node-well-linkedness. In order to be able to translate between both
types of well-linkedness and the treewidth, we need to reduce the
maximum vertex degree of the input graph $G$.  Using the cut-matching
game, one can reduce the maximum vertex degree to $O(\log^3 k)$, while
only losing a $\poly\log k$ factor in the treewidth, as was noted in
\cite{ChekuriE13} (see Remark 2.2). The following theorem, whose proof
appears in the Appendix, provides the starting point for our
algorithm.

\begin{theorem}\label{thm: starting point}
  There is an efficient randomized algorithm, that, given a graph $G$ with $\tw(G)=k$, computes a subgraph $G'$ of $G$ with maximum
  vertex degree $\Delta=O(\log^3k)$, and a subset $Z$ of
  $\Omega(k/\poly\log k)$ vertices of $G'$, such that $Z$ is
  node-well-linked in $G'$, with high probability.
\end{theorem}

We note that one can also reduce the degree to a constant with an
additional $\polylog k$ factor loss in the treewidth \cite{ChekuriE13},
however that result also relies on the preceding theorem as a starting
point. The constant can be made $4$ with a polynomial factor loss in
treewidth \cite{KreutzerT10} which we would not wish to lose.  We also
note that in~\cite{tw-sparsifiers} the authors have shown that the
degree can be reduced to $3$, and a set $X$ of cardinality
$\Omega(k/\poly\log k)$ as in the theorem can be computed efficiently,
but that proof builds on the present work.


\label{---------------------------------sec: from Path of Sets to grid----------------------------------------}
\section{A Path-of-Sets System}\label{sec: from path of sets to grid}

In this section we define our main combinatorial object, called a
path-of-sets system. We start with a few definitions.

Suppose we are given a collection $\sset=\set{S_1,\ldots,S_{\ell}}$ of
disjoint vertex subsets of $V(G)$. Let $S_i,S_j\in \sset$ be two such
subsets. We say that a path $P = (v_1,\ldots,v_h)$ \emph{connects} $S_i$
to $S_j$ if and only if the first vertex $v_1$ belongs to $S_i$ and
the last vertex $v_h$ belongs to $S_j$. We say that $P$ connects $S_i$
to $S_j$ \emph{directly}, if additionally $P$ does not contain
vertices of $\bigcup_{S \in \sset}S$ as inner vertices.

\begin{definition}
A path-of-sets system of width $w$ and length $\ell$ consists of:

\begin{itemize}
\item A sequence\footnote{We also interpret $\sset$ as a collection of
    sets for notational ease.} $\sset= (S_1,\ldots,S_{\ell})$ of
  $\ell$ disjoint vertex subsets of $G$, where for each $i$, $G[S_i]$
  is connected;

\item For each $1\leq i\leq \ell$, two disjoint sets $A_i,B_i\subseteq
  S_i$ of vertices of cardinality $w$ each, such that sets $A_i$ and
  $B_i$ are linked in $G[S_i]$; and

\item For each $1\leq i<\ell$, a set $\pset_i$ of $w$ disjoint paths,
  connecting the vertices of $B_i$ to the vertices of $A_{i+1}$ {\bf
    directly} (that is, paths in $\pset_i$ do not contain the vertices
  of $\bigcup_{S\in \sset}S$ as inner vertices), such that all paths
  in $\bigcup_i\pset_i$ are mutually disjoint.  (See Figure~\ref{fig:
    path-set-system}).
\end{itemize}

We say that it is a {\bf strong} path-of-sets system, if additionally
for each $1\leq i\leq \ell$, $A_i$ is node-well-linked in $G[S_i]$,
and so is $B_i$.
\end{definition}

Notice that a path-of-sets system is completely determined by the
sequence $\sset$ of vertex subsets; the collection
$\bigcup_{i=1}^{\ell-1}\pset_i$ of paths; and the sets $A_1\subseteq
S_1, B_{\ell}\subseteq S_{\ell}$ of vertices. In the following we will
denote path-of-sets systems by
$(\sset,\bigcup_i\pset_i,A_1,B_{\ell})$.

We note that Leaf and Seymour~\cite{LeafS12} have defined a very
similar object, called a $(w,\ell)$-grill, and they showed that the two
objects are roughly equivalent. Namely, a path-of-sets system with
parameters $w$ and $\ell$ contains a $(w,\ell)$-grill as a minor, while a
$(w,\ell)$-grill contains a path-of-sets system of width $w$ and length
$\Omega(\ell/w)$. They also show an efficient algorithm, that, given a $(w,\ell)$-grill with
$w=\Omega(g^3)$ and $\ell=\Omega(g^4)$, finds a model of the
$(g\times g)$-grid minor in the grill\footnote{In fact \cite{LeafS12}
  shows a slightly stronger result that a $(w,\ell)$-grill with $w \ge
  (2g+1)(2r-5)+2$ and $\ell \ge r (2g+r-2)$ contains a $(g \times
  g)$-grid minor or a bipartite-clique $K_{r,r}$ as a minor. This
  can give slightly improved bounds on the grid minor size if the
  given graph excludes bipartite-clique minors for small $r$.}.

Our goal is to show that a graph containing
a large enough path-of-sets system must also contain a large grid
minor. The following theorem is a starting point. 
The proof appears in Appendix.
\begin{theorem}\label{thm: find grid minor or good linkage}
  There is an efficient algorithm, that, given a connected graph
  $G=(V,E)$, two disjoint subsets $A,B\subseteq V$ of its vertices
  with $|A|=|B|=w$, such that $A,B$ are linked in $G$, and integers
  $h_1,h_2>1$ with $(16h_1+10)h_2\leq w$, either returns a model of
  the $(h_1\times h_1)$-grid minor in $G$, or computes a collection
  $\pset$ of $h_2$ node-disjoint paths, connecting vertices of $A$ to
  vertices of $B$, such that for every pair $P,P'\in \pset$ of paths
  with $P\neq P'$, there is a path $\beta_{P,P'}\subseteq G$,
  connecting a vertex of $P$ to a vertex of $P'$, where $\beta_{P,P'}$
  is internally disjoint from $\bigcup_{P''\in \pset}V(P'')$.
\end{theorem}

Given a path-of-sets system $(\sset,\bigcup_{i=1}^{\ell-1}\pset_i, A_1,B_{\ell})$ in $G$, we say that $G'$ is a subgraph of $G$ \emph{spanned} by the path-of-sets system, if $G'$ is the union of $G[S_i]$ for all $1\leq i\leq \ell$ and all paths in $\bigcup_{i=1}^{\ell-1}\pset_i$. 

The following corollary of Theorem~\ref{thm: find grid minor or good linkage} allows us to obtain a grid minor from a Path-of-Sets system.  Its proof appears in Appendix.

\begin{corollary}\label{cor: paths from the path-set system}
There is an efficient algorithm, that, given a  graph $G$, a path-of-sets system $(\sset,\bigcup_{i=1}^{\ell-1}\pset_i, A_1,B_{\ell})$ of length $\ell\geq 2$ and width $w$ in $G$, and integers $h_1,h_2$ with  $(16h_1+10)h_2\leq w$, either returns a model of the $(h_1\times h_1)$-grid minor in the subgraph $G'$ of $G$ spanned by the path-of-sets system, or returns a collection $\qset$ of $h_2$ node-disjoint paths in $G'$, connecting vertices of $A_1$ to vertices of $B_{\ell}$, such that for all $1\leq i\leq \ell$, for every path $Q\in \qset$, $S_i\cap Q$ is a path, and $S_1\cap Q,S_2\cap Q,\ldots,S_{\ell}\cap Q$ appear on $Q$ in this order. Moreover, for every $1\leq j\leq \floor{\ell/2}$, for every pair $Q,Q'\in \qset$ of paths, there is a path $\beta_{2i}(Q,Q')\subseteq G'[S_{2i}]$, connecting a vertex of $Q$ to a vertex of $Q'$, such that $\beta_{2i}(Q,Q')$ is internally disjoint from all paths in $\qset$.
\end{corollary}

The following corollary completes the construction of the grid minor, slightly improving upon a similar result of~\cite{LeafS12}. The proof is included in Appendix.

\begin{corollary}\label{cor: from path-set system to grid minor}
There is an efficient algorithm, that, given a graph $G$, an integer $g>1$ and a path-of-sets system $\left(\sset,\bigcup_{i=1}^{\ell-1}\pset_i, A_1,B_{\ell}\right)$  of width  $w=16g^2+10g$ and length $\ell=2g(g-1)$ in $G$, computes a model of the $\left (g\times g\right )$-grid minor in $G$.
\end{corollary}

The main technical contribution of our paper is summarized in the following theorem.

\begin{theorem}\label{thm: path-of-sets: main}
There are constants $c', c''>0$ and an efficient randomized algorithm, that, given a graph $G$ of treewidth $k$ and integral parameters $w^*,\ell^*>2$, such that $k/\log^{c'}k>c''w^*(\ell^*)^{48}$, with high probability returns a strong path-of-sets system of width $w^*$ and length $\ell^*$ in $G$.
\end{theorem}

Choosing $w^*,\ell^*=\Omega(k^{1/49}/\poly\log k)$, from Theorem~\ref{thm: path-of-sets: main}, we can efficiently construct a path-of-sets system of width $w^*$ and length $\ell^*$ in $G$ with high probability. From Corollary~\ref{cor: from path-set system to grid minor}, we can then efficiently construct a model of a grid minor of size $\left (\Omega(k^{1/98}/\poly\log k)\times \Omega(k^{1/98}/\poly\log k)\right )$. The rest of this paper is mostly dedicated to proving Theorem~\ref{thm: path-of-sets: main}. In Section~\ref{sec: extensions} we provide some extensions to this theorem, that we believe may be useful in various applications, such as, for example, algorithms for routing problems.

\label{------------------------------------sec: getting path of sets from tree------------------------------------}
\section{Constructing a Path-of-Sets System}\label{sec: getting path of sets from tree}

We can view a path-of-sets system as a meta-path, whose vertices
$v_1,\ldots,v_{\ell}$ correspond to the sets $S_1,\ldots,S_{\ell}$,
and each edge $e=(v_i,v_{i+1})$ corresponds to the collection
$\pset_i$ of $w$ disjoint paths. Unfortunately, we do not know how to
find such a meta-path directly (except for $\ell=O(\log k)$, which is
not enough for us). As we show below, a generalization of the work
of~\cite{ChuzhoyL12}, combined with some ideas from~\cite{ChekuriE13}
gives a construction of a meta-tree of degree at most $3$, instead of
the meta-path. We define the corresponding object that we call a
\emph{tree-of-sets system}. We start with the following definitions.

\begin{definition}
  Given a set $S$ of vertices in graph $G$, the \emph{interface} of
  $S$ is $\Gamma_G(S)=\set{v\in S\mid \exists e=(u,v)\in \out_G(S)}$. We
  say that $S$ has the $\alpha$-bandwidth property in $G$ if its
  interface $\Gamma_G(S)$ is $\alpha$-well-linked in $G[S]$.
\end{definition}


\begin{definition}
A tree-of-sets system with parameters $\ell,w,\alphawl$ ($\ell,w\geq 1$ are
integers and $0<\alphawl <1$ is real-valued) consists of:

\begin{itemize}
\item A collection $\sset=\set{S_1,\ldots,S_{\ell}}$ of $\ell$ disjoint vertex
  subsets of $G$, where for each $1\leq i\leq \ell$, $G[S_i]$ is connected;

\item A tree $T$ with $V(T)=\set{v_1,\ldots,v_{\ell}}$, whose
  maximum vertex degree is at most $3$;

\item For each edge $e=(v_i,v_j)$ of $T$, a set $\pset_e$
  of $w$ disjoint paths, connecting $S_i$ to $S_{j}$ {\bf directly}
  (that is, paths in $\pset_e$ do not contain the vertices of
  $\bigcup_{S\in \sset}S$ as inner vertices). Moreover, all paths in
  $\pset=\bigcup_{e\in E(T)}\pset_e$ are pairwise disjoint,
\end{itemize}

and has the following additional property.
Let $G'$ be the subgraph of $G$ obtained by the union of $G[S_i]$ for
all $S_i\in \sset$ and $\bigcup_{e\in E(T)}\pset(e)$. Then each
$S_i\in \sset$ has the $\alphawl$-bandwidth property in $G'$.
\end{definition}

We say that the graph $G'$ defined above is a subgraph of $G$ \emph{spanned} by the tree-of-sets system.

We remark that a tree-of-sets system is closely related to the path-of-sets system: a path-of-sets system is a tree-of-sets system where the tree is restricted to be a path; it is easy to verify that the linkedness property of sets $A_i,B_i$ inside every cluster $S_i$ guarantee the $1/4$-bandwidth property of $S_i$ in $G'$.

The following theorem describes our construction of a tree-of-sets
system.  It strengthens the results of~\cite{ChuzhoyL12} and its
proof appears in Section~\ref{sec: finding the tree}.

\begin{theorem}\label{thm: meta-tree}
  There is a constant $c$ and an efficient randomized
  algorithm that takes as input (i) a graph $G$ of maximum degree
  $\Delta$; (ii) a subset $\tset$ of $k$ vertices in $G$ called
  terminals, such that $\tset$ is node-well-linked in $G$ and the
  degree of every vertex in $\tset$ is $1$; and (iii) two integer
  parameters $\ell>1, w>4\log k$, such that
  $k/\log^{4}k>cw\ell^{19}\Delta^8$, and with high probability outputs a
  tree-of-sets system $(\sset,T,\bigcup_{e\in E(T)}\pset_e)$ in $G$,
  with parameters $w,\ell$ and $\alphawl=\Omega\left
    (\frac 1 {\ell^2\log^{1.5}k}\right )$. Moreover, for all $S_i\in \sset$,
  $S_i\cap \tset=\emptyset$.
\end{theorem}

We prove Theorem~\ref{thm: meta-tree} in the following section, and
show how to construct a path-of-sets system using this theorem
here.

Suppose we are given a tree-of-sets system $(\sset,T,\bigcup_{e\in
  E(T)}\pset(e))$, and an edge $e\in E(T)$, incident on a vertex
$v_i\in V(T)$. We denote by $\delta_{S_i}(e)\subseteq S_i$ the set of all vertices of $S_i$ that serve as endpoints of paths in $\pset(e)$.

\begin{definition}
A tree-of-sets system $(\sset,T,\bigcup_{e\in E(T)}\pset(e))$ with parameters $w,\ell,\alphawl$ is a {\bf strong} tree-of-sets system, if and only if for each $S_i \in \sset$:

\begin{itemize}
\item for each edge $e\in E(T)$ incident to $v_i$, the set $\delta_{S_i}(e)\subseteq S_i$ of vertices is node well-linked in $G[S_i]$; and
\item for every pair $e,e' \in E(T)$ of edges incident
    to $v_i$, the sets $\delta_{S_i}(e),\delta_{S_i}(e')\subseteq S_i$ of vertices are linked
    in $G[S_i]$.
\end{itemize}
\end{definition}

The following lemma allows us to transform an arbitrary tree-of-sets
system into a strong one.

\begin{lemma}
  \label{lem:strong-tree-of-set-system} There is an efficient
  algorithm, that, given a graph $G$ with maximum vertex degree at
  most $\Delta$ and a tree-of-sets system $(\sset, T, \bigcup_{e \in
    E(T)} \pset_e)$ with parameters $\ell,w,\alphabw$ in $G$, outputs
  a strong tree-of-sets system $(\sset, T, \bigcup_{e \in E(T)}
  \pset^*_e)$ with parameters $\ell,\tilde{w},\half$ such that for
  each $e \in E(T)$ $\pset^*_e \subseteq \pset_e$, and $\tilde{w} =
  \Omega\left(\frac{\alphabw^2}{\Delta^{10} (\alphasc(w))^2} \cdot
    w\right)$.
\end{lemma}

\begin{proof}
%
  We assume that tree $T$ is rooted at some vertex whose degree is
  greater than $1$.  We process the vertices of the tree $T$ in the
  bottom-up fashion: that is, we only process a vertex $v_i$ after all
  its descendants have been processed. Assume first that $v_i$ is a
  leaf vertex, and let $e$ be the unique edge incident to $v$ in
  $T$. We use Corollary~\ref{cor: ultimate boosting} to compute a
  subset $\delta'\subseteq \delta_{S_i}(e)$ of at least
  $\floor{\frac{\alphawl w}{2^{11}\Delta^5\alphasc(w)}}$ vertices,
  such that the vertices of $\delta'$ are node-well-linked in
  $G[S_i]$. We then discard from $\pset(e)$ all paths except those
  whose endpoint lies in $\delta'$.

  Consider now some non-leaf vertex $v_i$ of the tree, and assume that
  it has degree $3$ (the case where $v_i$ has degree $2$ is dealt with
  similarly). Let $e_1,e_2,e_3$ be the edges incident to $v_i$ in $T$,
  and for each $1\leq j\leq 3$, let $\delta_j=\delta_{S_i}(e_j)$; note
  that $\delta_j$ is based on the current set of paths $\pset(e)$ as
  we process the tree.  Recall that the set $\delta_1\cup \delta_2\cup
  \delta_3$ of vertices is $\alphawl$-well-linked in $G[S_i]$. We use
  Corollary~\ref{cor: ultimate boosting} to compute, for each $1\leq
  j\leq 3$ a subset $\delta'_j$ of at least $\floor{\frac{\alphawl
      |\delta_j|}{2^{11}\Delta^5\alphasc(w)}}$ vertices, so that each
  of the sets $\delta_j'$ is node-well-linked in $G[S_i]$, every pair
  of such sets is linked in $G[S_i]$, and $\delta_1'\cup \delta_2'\cup
  \delta_3'$ is $1/2$-well-linked in $G[S_i]$. For $1 \le j \le 3$ we
  discard paths from $\pset(e_j)$ that do not have an endpoint in
  $\delta'_j$.  Once all vertices of $T$ are processed, we claim that
  for every edge $e\in E(T)$ the resulting set $\pset(e)$ contains at
  least $\Omega\left(\frac{\alphabw^2}{\Delta^{10} (\alphasc(w))^2}
    \cdot w\right )$ paths, and the new tree-of-sets system is
  guaranteed to be strong. The latter property is easy to see. For the
  former, consider an edge $e = (v_i,v_{i'}) \in E(T)$ where
  $v_{i'}$ is the parent of $v_i$. Before processing $v_i$ there
  are $w$ paths in $\pset(e)$. After processing $v_i$ 
  there are least $\delta' = \floor{\frac{\alphawl
      w}{2^{11}\Delta^5\alphasc(w)}}$ paths left in $\pset(e)$.  
  After processing $v_{i'}$ there are at least 
  $\delta'' = \floor{\frac{\alphawl
      |\delta'|}{2^{11}\Delta^5\alphasc(w)}}$ paths that remain in $\pset(e)$.
  Paths in $\pset(e)$ are only eliminated when processing $v_i$ and $v_{i'}$
  and this gives us the desired claim.
\end{proof}

The following theorem allows us to obtain a strong path-of-sets system from a strong tree-of-sets system.

\begin{theorem}
  \label{thm:tree-of-set-to-path-of-set}
  There is an efficient algorithm, that, given a graph $G$ and a strong tree-of-sets system $(\sset, T, \bigcup_{e \in E(T)} \pset^*_e)$ with parameters $\ell,\tilde{w},\half$, and integers $w^*, \ell^*>1$, such that  $(\ell^*)^2 \leq \ell$ and $\tilde{w} > 16 w^* (\ell^*)^2+1$,  outputs a
  strong path-of-sets system $(\sset',\bigcup_{i=1}^{\ell^*-1}\pset_i,A_1,B_{\ell^*})$ of length $\ell^*$ and width $w^*$, with $\sset'\subseteq \sset$.
\end{theorem}

Before we prove the preceding theorem we use the results stated so far
to complete the proof of Theorem~\ref{thm: path-of-sets: main}.

\begin{proofof}{Theorem~\ref{thm: path-of-sets: main}}
  We assume that $k$ is large enough, so, e.g. $k^{1/30}>c^*\log k$
  for some large enough constant $c^*$.  Given a graph $G=(V,E)$ with
  treewidth $k$, we use Theorem~\ref{thm: starting point} to compute a
  subgraph $G'$ of $G$ with maximum vertex degree $\Delta=O(\log^3k)$,
  and a set $X$ of $\Omega(k/\poly\log k)$ vertices, such that $X$ is
  node-well-linked in $G'$. We add a new set $\tset$ of $|X|$
  vertices, each of which connects to a distinct vertex of $X$ with an
  edge. For convenience, we denote this new graph by $G$, and
  $|\tset|$ by $k$, and we refer to the vertices of $\tset$ as
  \emph{terminals}. Clearly, the maximum vertex degree of $G$ is at
  most $\Delta=O(\log^3k)$, the degree of every terminal is $1$, and
  $\tset$ is node-well-linked in $G$. 
  We can now assume that $\frac{k}{\Delta^{19}\log^8k}>\hat c w^*(\ell^*)^{48}$ for some large enough constant $\hat c$.
    
 We set $\ell=(\ell^*)^2$ and $w=\frac{\hat c}{c}\cdot w^*(\ell^*)^{10}\Delta^{11}\log^4k$, so $w>4\log k$ holds, where $c$ is the constant from Theorem~\ref{thm: meta-tree}.
 Clearly:
 
 \[cw\ell^{19}\Delta^{8}=(\hat c w^*(\ell^*)^{10}\Delta^{11}\log^4k)\cdot (\ell^*)^{38}\Delta^{8}=\hat c w^*(\ell^*)^{48}\Delta^{19}\log^4k.\]
 
 Therefore, $\frac{k}{\log^4k}>cw\ell^{19}\Delta^{8}$.
  We then apply Theorem~\ref{thm:
    meta-tree} to $G$ and $\tset$ to obtain a tree-of-sets system
  $(\sset,T,\bigcup_{e\in E(T)}\pset_e)$, with parameters
  $\ell$, $w$ and $\alphabw =
  \Omega(\frac{1}{\ell^2 \log^{1.5} k})$. 

  We use Lemma~\ref{lem:strong-tree-of-set-system} to convert
  $(\sset,T,\bigcup_{e\in E(T)}\pset_e)$ into a strong tree-of-set
  system $(\sset,T,\bigcup_{e\in E(T)}\pset^*_e)$ with parameters $\ell$
  and $\tilde{w} = \Omega(\frac{\alphabw^2}{\Delta^{10}
    (\alphasc(w))^2} \cdot w)$. If $\hat c$ is chosen to be large enough, 
  $\tilde w> 16w^*(\ell^*)^2+1$ must hold. We then apply
  Theorem~\ref{thm:tree-of-set-to-path-of-set} to obtain a path-of-set
  system with width $w^*$ and length
  $\ell^*$.
  \end{proofof}

We now prove Theorem~\ref{thm:tree-of-set-to-path-of-set}. 

\begin{proofof}{Theorem~\ref{thm:tree-of-set-to-path-of-set}}
Let
$(\sset, T, \bigcup_{e \in E(T)}\pset^*_e)$ be the tree-of-set
system with parameters $\ell$ and $\tilde{w}$. 
%
For convenience, for each set $S\in \sset$, we denote the
corresponding vertex of tree $T$ by $v_S$.  If tree $T$ contains a
root-to-leaf path of length at least $\ell^*$, then we are done, as this
path gives a path-of-sets system of width $\tilde{w} \ge w^*$ and length
$\ell^*$. The path-of-sets system is strong, since for every edge $e=(v_i,v_{i'})\in E(T)$, $\delta_{S_i}(e)$ is node-well-linked in $G[S_i]$.

Otherwise, since $|V(T)| = \ell \ge (\ell^*)^2$, $T$ must contain at
least $\ell^*+1$ leaves (see Claim~\ref{claim:path-leaves-in-tree}). Let
$L$ be a subset of $\ell^*$ leaves of $T$, and let $\lset\subseteq \sset$ be the
collection of $\ell^*$ clusters, whose corresponding vertices belong to
$L$, so $\lset=\set{S\in \sset\mid v_S\in L}$. We next show how to build
a path-of-sets system $(\sset',\bigcup_{i=1}^{\ell^*-1}\pset_i,A_1,B_{\ell^*})$, whose collection of clusters is $\sset'=\lset$.

Intuitively, we would like to perform a depth-first-search (DFS) tour on our
meta-tree $T$. This should be done with many paths in parallel. In
other words, we want to build $w^{*}$ disjoint paths, that visit the
clusters in $\sset$ in the same order --- the order of the tour. The
clusters in $\lset$ will then serve as the sets $\sset'$ in our final
path-of-sets system, and the collection of $w^*$ paths that we build
will be used for the path sets $\pset_i$. In order for this to work, we
need to route up to three sets of paths across clusters $S\in
\sset$. For example, if the vertex $v_S$ corresponding to the cluster
$S$ is a degree-3 vertex in $T$, then for the DFS tour, we need to
route three sets of paths across $S$: one set connecting the paths
coming from the parent of $v_S$ to its first child, one set connecting the
paths coming back from the first child to the second child, and one
set connecting the paths coming back from the second child to the
parent of $v_S$ (see Figure~\ref{fig: DFS-whole}). Even though every pair of
relevant vertex subsets on the interface of $S$ is linked, this
property only guarantees that we can route one such set of paths,
which presents a major technical difficulty in using this approach
directly. 
 
 \begin{figure}[h]
\centering
\scalebox{0.4}{\includegraphics{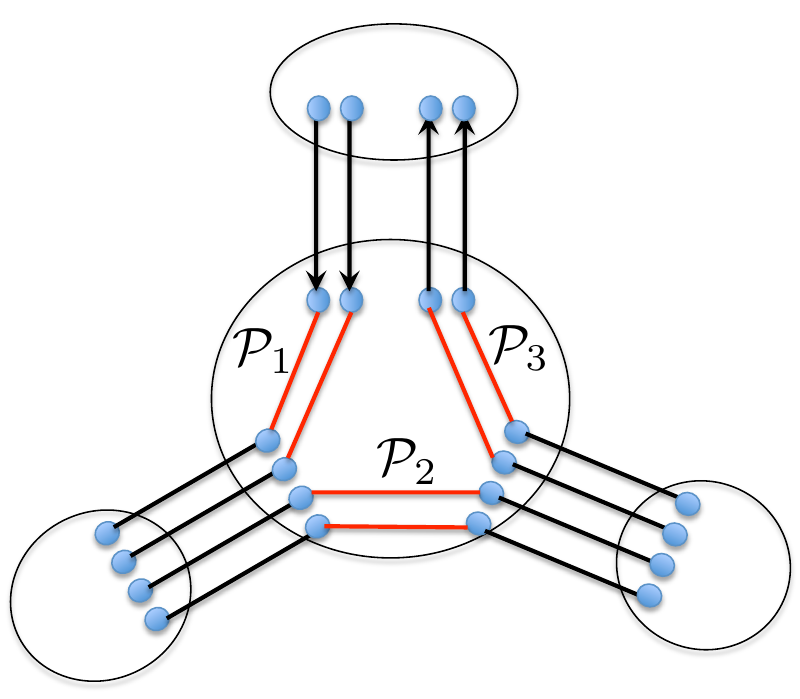}}\label{fig: DFS2}
\caption{Routing paths of the DFS tour inside $S$\label{fig: DFS-whole}}
\end{figure}

 Our algorithm consists of two phases. In the first phase, we build a collection of disjoint paths, connecting the cluster corresponding to the root of the tree $T$ to the clusters in $\lset$, along the root-to-leaf paths in $T$. In the second phase, we build the path-of-sets system by exploiting the paths constructed in Step 1, to simulate the tree tour.
 
\subsection{Step 1}
Let $G'$ be the graph obtained from the union of $G[S]$ for all $S\in
\sset$, and the sets $\pset^*_e$ of paths, for all $e\in E(T)$. We
root $T$ at a degree-$1$ vertex that does not belong to $L$ (since $T$ has at least $\ell^*+1$ leaves and $|L|=\ell^*$, such a vertex exists), and we let $S^*$ be the cluster
corresponding to the root of $T$. The goal of the first step is
summarized in the following theorem.
 
 \begin{theorem}
There is an efficient algorithm to compute, for each $S\in \lset$, a collection $\qset_S$ of $\floor{\tilde{w}/\ell^*}$ paths  in graph $G'$, that have the following properties:
 
 \begin{itemize}
 \item Each path $Q\in \qset_S$ starts at a vertex of $S^*$ and terminates at a vertex of $S$; its inner vertices are disjoint from $S$ and $S^*$.
 
 \item For each path $Q\in \qset_S$, for each cluster $S'\in \sset$, such that $v_{S'}$ lies on the path connecting $v_{S^*}$ to $v_S$ in $T$, $Q\cap G[S']$ is a (non-empty) path. For all other clusters $S'\in \sset$, $Q\cap G[S']=\emptyset$.
 
 \item The paths in ${\qset}=\bigcup_{S\in \lset}\qset_S$ are vertex-disjoint.
 \end{itemize}
 \end{theorem}
 
 Notice that from the structure of graph $G'$, if $P$ is the path connecting $v_{S^*}$ to $v_S$ in the tree $T$, then every path in $\qset_S$ visits every cluster $S'$ with $v_{S'}\in P$ exactly once, in the order in which they appear on $P$, and it does not visit other clusters of $\sset$.

\begin{proof}
Recall that for every vertex $v_S\in V(T)$, and for each edge $e$ incident to $v_S$, we have defined a subset $\delta_S(e)\subseteq S$ of vertices that serve as endpoints of the paths in $\pset^*_e$.
  For each cluster $S'\in \sset$, let $n(S')$ be the number of the
  descendants of $v_{S'}$ in the tree $T$ that belong to $L$. If
  $S'\neq S^*$, then let $e$ be the edge of the tree $T$ connecting
  $v_{S'}$ to its parent, and denote $\delta_{S'}=\delta_{S'}(e)$. We process the tree in top to bottom order, while
  maintaining a set $\qset$ of disjoint paths.  We ensure that the
  following invariant holds throughout the algorithm. Let $S,S'\in \sset$ be
  a pair of clusters, such that $v_S$ is the parent of $v_{S'}$ in
  $T$. Assume that so far the algorithm has processed $v_S$ but it has
  not processed $v_{S'}$ yet. Then there is a collection
  $\qset_{S'}\subseteq\qset$ of $n(S')\cdot \floor{\tilde w/\ell^*}$ paths
  connecting $S^*$ to $S'$ in $\qset$. Each such path does not share
  vertices with $S'$, except for its last vertex, which must belong to
  $\delta_{S'}$. Moreover, for every path $Q\in \qset_{S'}$, for every
  cluster $S''\in \sset$, such that $v_{S''}$ lies on the path connecting $v_{S^*}$
  to $v_{S'}$ in $T$, $Q\cap G[S'']$ is a (non-empty) path, and for every other cluster $S''$, $Q\cap
  G[S'']=\emptyset$.

  In the first iteration, we start with the root vertex $v_{S^*}$. Let
  $v_S$ be its unique child, and let $e=(v_{S^*},v_S)$ be the corresponding
  edge of $T$. We let $\qset_{S}$ be an arbitrary subset of $n(S)\cdot
  \floor{\tilde w/\ell^*}$ paths of $\pset^*_{e}$, and we set
  $\qset=\qset_S$. (Notice that $|L|\cdot \floor{\tilde w/\ell^*}\leq
  \tilde w=|\pset^*_{e}|$, since $|L|=\ell^*$, so we can always find such a
  subset of paths).

  Consider now some non-leaf vertex $v_S$, and assume that its parent has
  already been processed. We assume that $v_S$ has two children. The case where
  $v_S$ has only one child is treated similarly. Let $\qset_S\subseteq\qset$ be the subset of
  paths currently connecting $S^*$ to $S$, and let $\Gamma'\subseteq
  \delta_S$ be the endpoints of these paths that belong to $S$. Let
  $v_{S'},v_{S''}$ be the children of $v_S$ in $T$, and let
  $e_1=(v_S,v_{S'}),e_2=(v_S,v_{S''})$ be the corresponding edges of
  $T$.  We need the following claim.

  \begin{claim}\label{claim: routing in sets}
    We can efficiently find a subset $\Gamma_1\subseteq\delta_S(e_1)$ of
    $n(S')\cdot \floor{\tilde w/\ell^*}$ vertices and a subset
    $\Gamma_2\subseteq\delta_S(e_2)$ of $n(S'')\cdot \floor{\tilde w/\ell^*}$
    vertices, together with a set $\rset$ of $|\Gamma'|$ disjoint
    paths contained in $G[S]$, where each path connects a vertex of
    $\Gamma'$ to a distinct vertex of $\Gamma_1\cup \Gamma_2$.
\end{claim}

\begin{proof}
  We build the following flow network, starting with $G[S]$. Set the
  capacity of every vertex in $S$ to $1$. Add a sink $t$, and connect
  every vertex in $\Gamma'$ to $t$ with a directed edge. Add a new
  vertex $s_1$ of capacity $n(S')\cdot \floor{\tilde w/\ell^*}$ and connect it
  with a directed edge to every vertex of
  $\delta_{S}(e_1)$. Similarly, add a new vertex $s_2$ of capacity
  $n(S'')\cdot \floor{\tilde w/\ell^*}$ and connect it with a directed edge to
  every vertex of $\delta_{S}(e_2)$. Finally, add a source $s$ and
  connect it to $s_1$ and $s_2$ with directed edges.
  From the integrality of flow, it is enough to show that there is an
  $s$-$t$ flow of value $|\Gamma'|=n(S)\cdot
  \floor{\tilde w/\ell^*}=(n(S')+n(S''))\cdot \floor{\tilde w/\ell^*}$ in this flow
  network. Since $\Gamma'$ and $\delta_S(e_1)$ are linked, there is a
  set $\pset_1$ of $|\Gamma'|$ disjoint paths connecting the vertices
  of $\Gamma'$ to the vertices of $\delta_S(e_1)$. We send
  $n(S')/n(S)$ flow units along each such path. Similarly, there is a
  set $\pset_2$ of $|\Gamma'|$ disjoint paths connecting vertices of
  $\Gamma'$ to vertices of $\delta_S(e_2)$. We send $n(S'')/n(S)$ flow
  units along each such path. It is immediate to verify that this
  gives a feasible $s$-$t$ flow of value $|\Gamma'|$ in this network.
\end{proof}

Let $\pset_1\subseteq \pset^*(e_1)$ be the subset of paths whose
endpoints belong to $\Gamma_1$, and define $\pset_2\subseteq
\pset^*(e_2)$ similarly for $\Gamma_2$.  Concatenating the paths in
$\qset_S$, $\rset$, and $\pset_1\cup \pset_2$, we obtain two
collections of paths: set $\qset_{S'}$ of $n(S')\cdot \floor{\tilde w/\ell^*}$
paths, connecting $S^*$ to $S'$, and set $\qset_{S''}$ of $n(S'')\cdot
\floor{\tilde w/\ell^*}$ paths, connecting $S^*$ to $S''$, that have the
desired properties. We delete the paths of $\qset_S$ from $\qset$, and
add the paths in $\qset_{S'}$ and $\qset_{S''}$ instead.

Once all non-leaf vertices of the tree $T$ are processed, we obtain
the desired collection of paths.
\end{proof}

\subsection{Step 2}
 In this step, we process the tree $T$ in the bottom-up order,
 gradually building the path-of-sets system. We will imitate the
 depth-first-search tour of the tree, and exploit the sets
 $\set{\qset_S\mid S\in \lset}$ of paths constructed in Step 1 to
 perform this step.
  
 For every vertex $v_S$ of the tree $T$, let $T_{v_S}$ be the subtree
 of $T$ rooted at $v_S$. Define a subgraph $G_S$ of $G'$ to be the
 union of all clusters $G'[S']$ with $v_{S'}\in V(T_{v_S})$, and all
 sets $\pset^*_e$ of paths with $e\in E(T_{v_S})$. We also define
 $L_S\subseteq L$ to be the set of all descendants of $v_S$ that
 belong to $L$, and $\lset_S=\set{S'\mid v_{S'}\in L_{S}}$ the
 collection of the corresponding clusters.

 We process the tree $T$ in a bottom to top order, maintaining the
 following invariant. Let $v_S$ be a vertex of $T$, and let $\ell_S$
 be the length of the longest simple path connecting $v_S$ to 
 its descendant in $T$. Once vertex $v_S$ is processed, we have
 computed a path-of-sets system $(\lset_S,\pset^S)$ of width $w^*$
 and length $|\lset_S|$, that is completely contained in $G_S$. (That
 is, the path-of-sets system is defined over the collection $\lset_S$
 of vertex subsets - all subsets $S'\in \lset$ where $v_{S'}$ is a
 descendant of $v_S$ in $T$). Let $X,Y\in \lset_S$ be the first and
 the last set on the path-of-sets system. Then we also compute subsets
 $\qset'_X\subseteq\qset_X$, $\qset'_Y\subseteq\qset_Y$ of paths of
 cardinality at least $\floor{\frac{\tilde w}{2\ell^*}}-8\ell_S\cdot w^*$, such that
 the paths in $\qset_X'\cup \qset_Y'$ are completely disjoint from the
 paths in $\pset^S$ (see Figure~\ref{fig: invariant}). Note that
 $\qset_X,\qset_Y$ are the sets of paths computed in Step 1, so the
 paths in $\qset_X\cup\qset_Y$ are also disjoint from $\bigcup_{S'\in
   \lset}S'$, except that one endpoint of each such path must belong
 to $X$ or $Y$. We note that since the tree height is less than
 $\ell^*$, $\floor{\frac{\tilde w}{2\ell^*}}-8\ell_S\cdot w^*>
 \floor{\frac{\tilde w}{2\ell^*}}-8\ell^*\cdot w^*\geq0$ where the latter inequality
 is based on the assumption that $\tilde{w} > 16w^* (\ell^*)^2+1$.

\begin{figure}[h]
\begin{center}
\scalebox{0.4}{\includegraphics{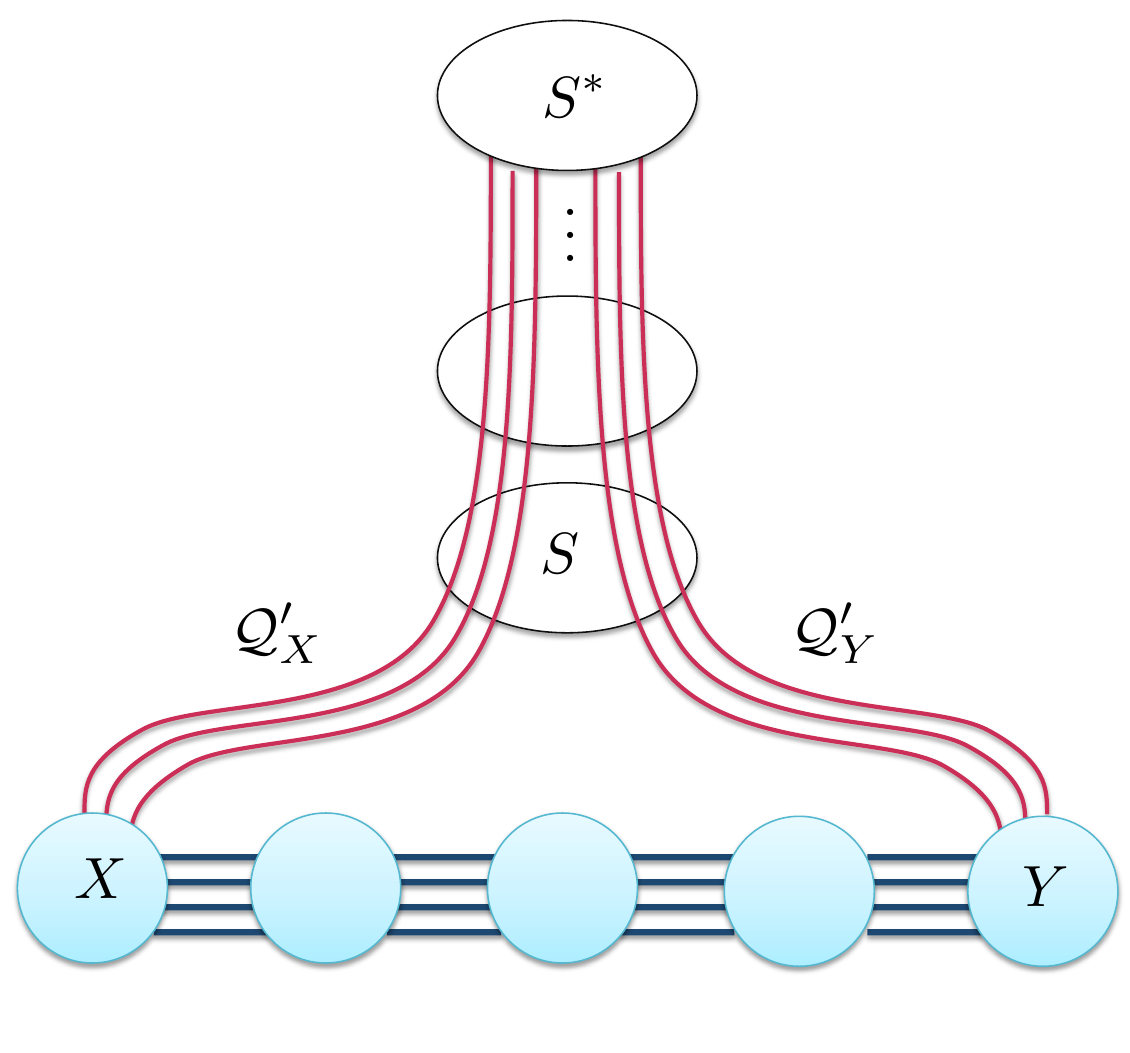}}\caption{Invariant for Step 2.\label{fig: invariant}}
\end{center}
\end{figure}

Clearly, once all vertices of the tree $T$ are processed, we obtain
the desired path-of-sets system $(\lset,\pset, A_1,B_{\ell^*})$ of length $\ell^*$ and
width $w^*$. We now describe the algorithm for processing each
vertex.

If $v_S$ is a leaf of $T$, then we do nothing. If $v_S\in L$, then the
path-of-sets system consists of only $\sset=\set{S}$, with $X=Y=S$. We let $\qset'_X,\qset'_Y$ be an arbitrary pair of disjoint subsets of $\qset_S$ containing $\floor{\frac{\tilde w}{2\ell^*}}$ paths each. If $v_S$ is a degree-2 vertex of $T$, then we also
do nothing. The path-of-sets system is inherited from its child, and
the corresponding sets $\qset'_X,\qset'_Y$ remain unchanged. Assume
now that $v_S$ is a degree-3 vertex, and let $v_{S'},v_{S''}$ be its
two children. Consider the path-of-sets systems that we computed for
its children: $(\lset_{S'},\pset^{S'})$ for $S'$ and
$(\lset_{S''},\pset^{S''})$ for $S''$. Let $X_1,Y_1$ be the first and
the last cluster of the first system, and $X_2,Y_2$ the first and the
last cluster of the second system (see Figure~\ref{fig:
  step2-start}). The idea is to connect the two path-of-sets systems
into a single system, by joining one of $\set{X_1,Y_1}$ to one of
$\set{X_2,Y_2}$ by $w^*$ disjoint paths. These paths are constructed by
concatenating sub-paths of some paths from $\qset'_{X_1}\cup \qset'_{Y_1}\cup
\qset'_{X_2}\cup \qset'_{Y_2}$, and additional paths contained in
$G[S]$.

Consider the paths in $\qset'_{X_1}$ and direct these paths from $X_1$
towards $S^*$. For each such path $Q$, let $v_Q$ be the first vertex
of $Q$ that belongs to $S$. Let $\Gamma_1=\set{v_Q\mid Q\in
  \qset'_{X_1}}$. We similarly define $\Gamma_2$,
$\Gamma'_1,\Gamma'_2$ for $\qset'_{Y_1}$, $\qset'_{X_2}$ and
$\qset'_{Y_2}$, respectively. Denote $\Gamma=\Gamma_1\cup \Gamma_2$,
and $\Gamma'=\Gamma_1'\cup \Gamma_2'$.  For simplicity, we denote the
portions of the paths in $\qset'_{X_1}\cup \qset'_{Y_1}$ that are
contained in $G[S]$ by $\pset$, and the portions of paths in
$\qset'_{X_2}\cup \qset'_{Y_2}$ that are contained in $G[S]$ by $\pset'$
(see Figure~\ref{fig: step2-middle}). That is,

\[\pset=\set{P\cap G[S]\mid P\in \qset'_{X_1}\cup \qset'_{Y_1}}; \quad\quad\quad \pset'=\set{P\cap G[S]\mid P\in \qset'_{X_2}\cup \qset'_{Y_2}}\]

Our goal is to find a set $\rset$ of $4w^*$ disjoint paths in $G[S]$
connecting $\Gamma$ to $\Gamma'$, such that the paths in $\rset$
intersect at most $8w^*$ paths in $\pset$, and at most $8w^*$ paths in
$\pset'$. Notice that in general, since sets $\Gamma,\Gamma'$ are
linked in $G[S]$, we can find a set $\rset$ of $4w^*$ disjoint paths
in $G[S]$ connecting $ \Gamma$ to $\Gamma'$, but these paths may
intersect many paths in $\pset\cup \pset'$. We start from an arbitrary
set $\rset$ of $4w^*$ disjoint paths connecting $\Gamma$ to $\Gamma'$
in $G[S]$. We next re-route these paths, using Lemma~\ref{lemma: re-routing of vertex-disjoint paths}.

\begin{figure}[h]
\centering
\subfigure[The beginning]{\scalebox{0.4}{\includegraphics{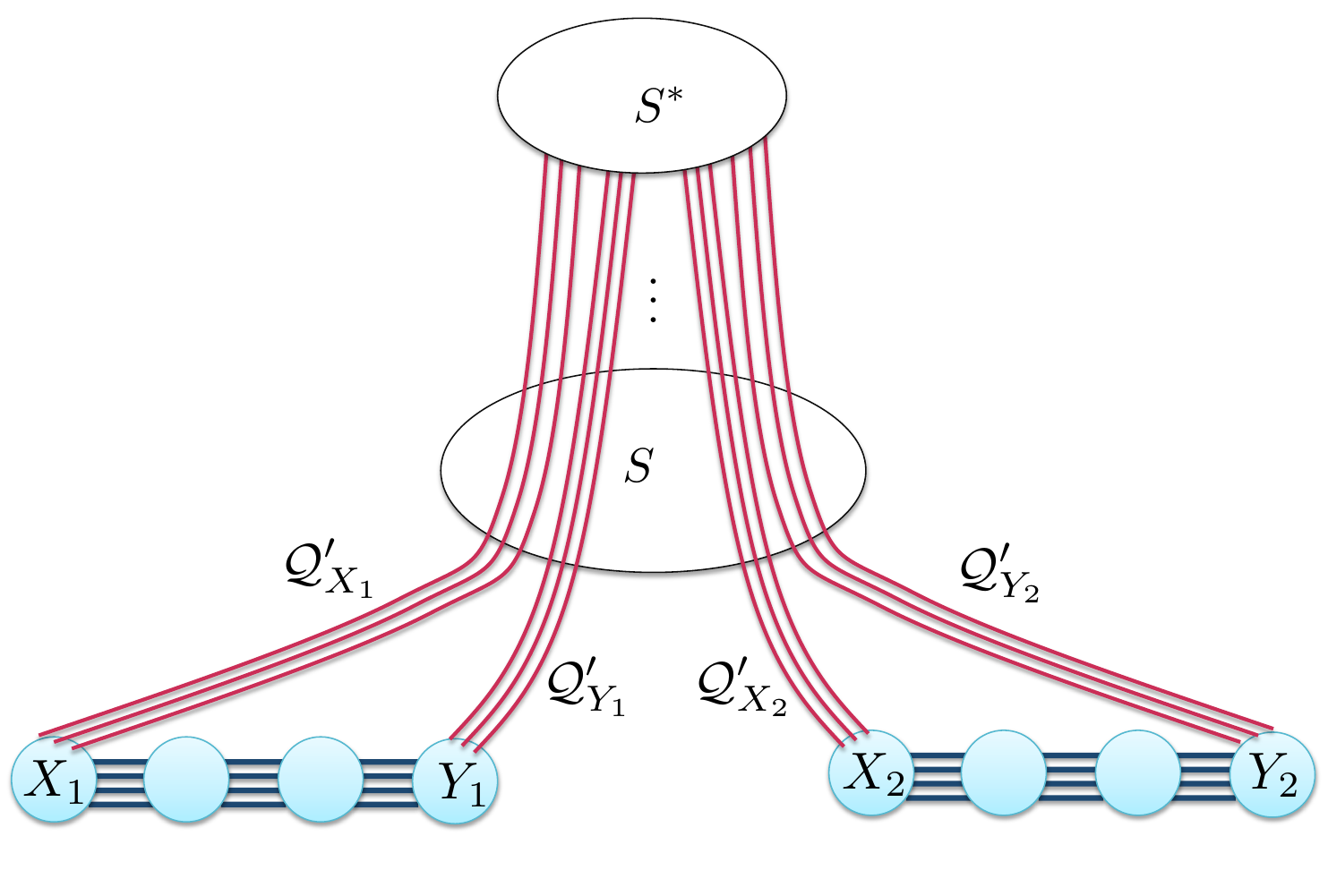}}\label{fig: step2-start}}
\subfigure[Finding the set $\rset$ of paths]{
\scalebox{0.5}{\includegraphics{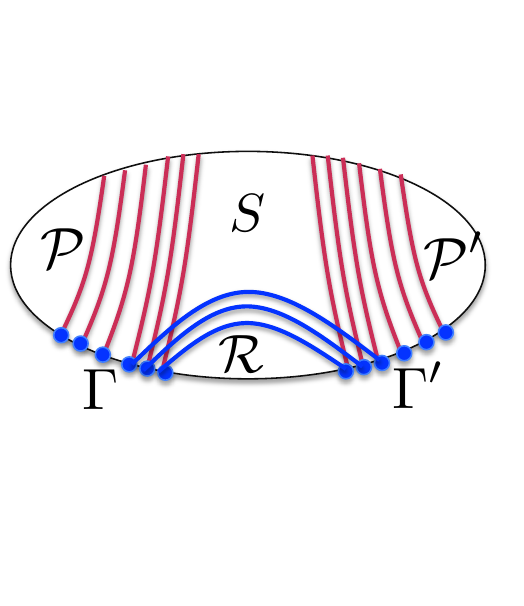}}\label{fig: step2-middle}}
\hspace{1cm}
\subfigure[The end]{\scalebox{0.4}{\includegraphics{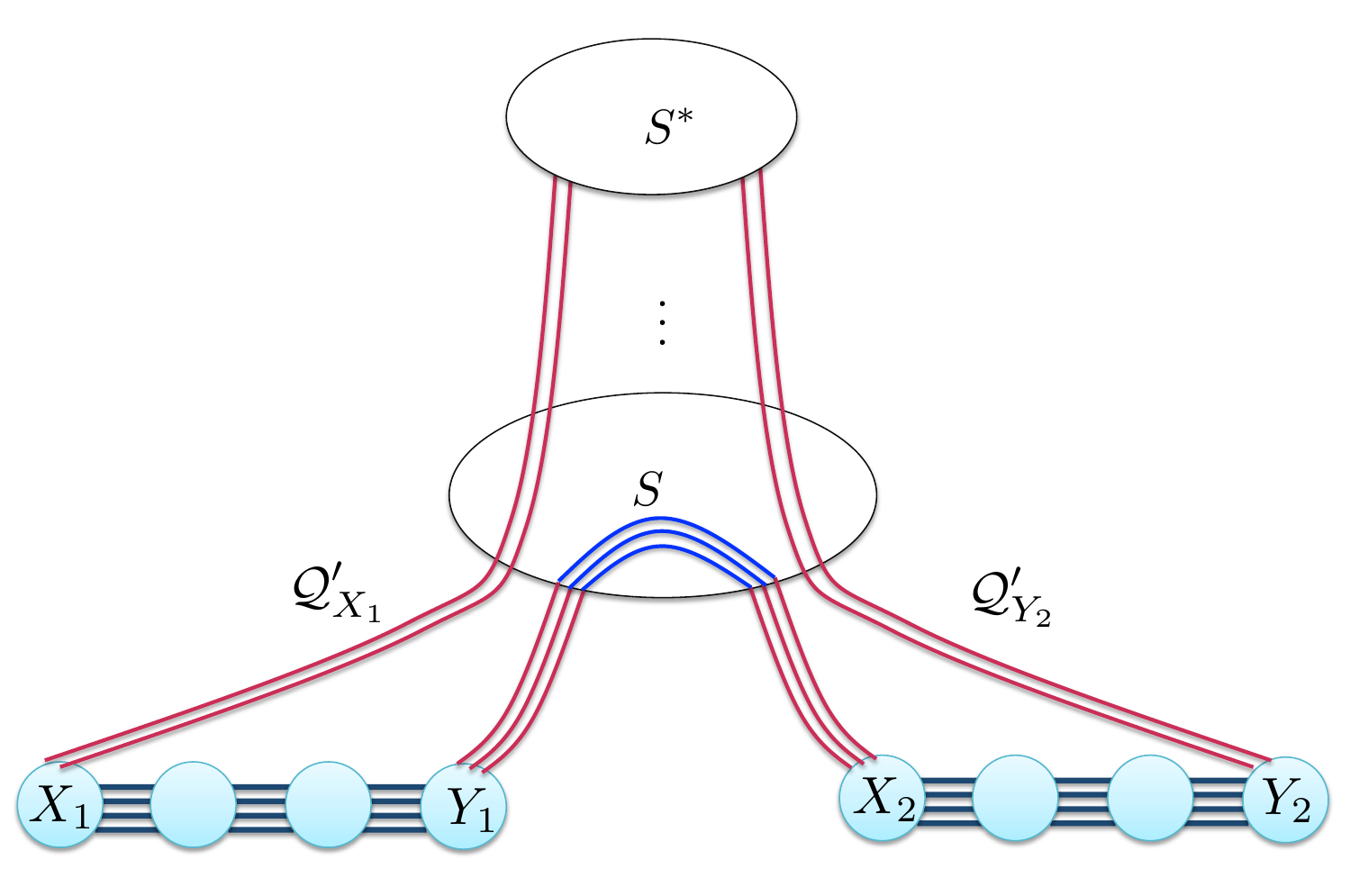}}\label{fig: step2-end}}
\caption{Processing a degree-3 vertex $v_S$. \label{fig: step 2}}
\end{figure}

We apply Lemma~\ref{lemma: re-routing of vertex-disjoint paths}
twice. First, we unify all vertices of $\Gamma$ into a single vertex
$s$, and direct the paths in $\pset$ and the paths in $\rset$ towards
it. We then apply Lemma~\ref{lemma: re-routing of vertex-disjoint
  paths} to the two sets of paths, with $\pset$ as $\xset_1$ and $\rset$
as $\xset_2$. Let $\tilde{\pset}\subseteq \pset$, $\rset'$ be the two
resulting sets of paths. We discard from $\tpset$ paths that share
endpoints with paths in $\rset'$ (at most $|\rset'|$ paths). Then
$|\tpset|\geq |\pset|-2|\rset|=|\pset|-8w^*$, and $\rset'$ contains
$4w^*$ disjoint paths connecting vertices in $\Gamma$ to vertices in
$\Gamma'$. Moreover, the paths in $\tpset\cup \rset'$ are completely
disjoint.

Next, we unify all vertices in $\Gamma'$ into a single vertex $s$, and
direct all paths in $\pset'$ and $\rset'$ towards $s$. We then apply
Lemma~\ref{lemma: re-routing of vertex-disjoint paths} to the two
resulting sets of paths, with $\pset'$ serving as $\xset_1$ and $\rset'$
serving as $\xset_2$. Let $\tpset'\subseteq\pset'$ and $\rset''$ be the two
resulting sets of paths. We again discard from $\tpset'$ all paths
that share an endpoint with a path in $\rset''$ -- at most $|\rset''|$
paths. Then $|\tpset'|\geq |\pset'|-2 |\rset''| \ge |\pset'|-8w^*$, and
the paths in $\tpset'\cup \rset''$ are completely disjoint from each
other. Notice also that the paths in $\rset''$ remain disjoint from
the paths in $\tpset$, since the paths in $\rset''$ only use vertices
that appear on the paths in $\rset'\cup \pset'$, which are disjoint
from $\tpset$.

Consider now the final set $\rset''$ of paths. The paths in $\rset''$
connect the vertices of $\Gamma_1\cup \Gamma_2$ to the vertices of
$\Gamma_1'\cup \Gamma_2'$. There must be two indices
$i,j\in\set{1,2}$, such that at least a quarter of the paths in
$\rset''$ connect vertices of $\Gamma_i$ to vertices of
$\Gamma'_j$. We assume without loss of generality that $i=2,j=1$, so at least $w^*$ of
the paths in $\rset''$ connect vertices of $\Gamma_2$ to
vertices of $\Gamma_1'$. Let $\rset^*\subseteq\rset''$ be the set of
these paths. We obtain a collection $\pset^*$ of $w^*$ paths
connecting $Y_1$ to $X_2$, by concatenating the prefixes of the paths
in $\qset'_{Y_1}$, the paths in $\rset''$, and the prefixes of the
paths in $\qset'_{X_2}$ (see Figure~\ref{fig: step2-end}). Notice that
the paths in $\pset^*$ are completely disjoint from the two
path-of-sets systems, except for their endpoints that belong to $Y_1$
and $X_2$. This gives us a new path-of-sets system, whose collection
of vertex sets is $\sset=\lset_S$. The first and the last sets in this system
are $X_1$ and $Y_2$, respectively. In order to define the new set
$\qset_{X_1}'$, we discard from $\qset_{X_1}'$ all paths that share
vertices with paths in $\rset''$ (as observed before, there are at
most $8w^*$ such paths). Since at the beginning of the current
iteration, $|\qset_{X_1}'|\geq \floor{\frac{\tilde w}{2\ell^*}}-8w^*\ell_{S'}\geq
\floor{\frac{\tilde w}{2\ell^*}}-8w^*(\ell_{S}-1)$, at the end of the current iteration,
$|\qset_{X_1}'|\geq \floor{\frac{\tilde w}{2\ell^*}}-8w^*\ell_{S}$ as required. The new
set $\qset'_{Y_2}$ is defined similarly.  From the construction, the
paths in $\qset_{X_1}'\cup \qset_{Y_2}'$ are completely disjoint from
the paths in $\rset^*$, and hence they are completely disjoint form
all paths participating in the new path-of-sets system. 

Notice that each vertex $v_i\in L$ is only incident on one edge $e\in E(T)$, and from the definition of strong tree-of-sets system, $\delta_{S_i}(e)$ is node-well-linked in $G[S_i]$. These are the only vertices of $S_i$ that may participate in the paths $\pset_j$ of the path-of-sets system, so we obtain a strong path-of-sets system.
\end{proofof}

In order to
complete the proof of Theorem~\ref{thm: main}, it now suffices to
prove Theorem~\ref{thm: meta-tree}

\label{------------------------------------sec: finding the tree----------------------------------------------}
\section{Proof of Theorem~\ref{thm: meta-tree}}
\label{sec: finding the tree}

This part mostly follows the algorithm of~\cite{ChuzhoyL12}. The main
difference is a change in the parameters, so that the number of
clusters in the tree-of-sets system is polynomial in $k$ and not
polylogarithmic, and extending the arguments of~\cite{ChuzhoyL12} to
handle vertex connectivity instead of edge connectivity. We also
improve and simplify some of the arguments of~\cite{ChuzhoyL12}.  Some
of the proofs and definitions are identical to or closely follow those
in~\cite{ChuzhoyL12} and are provided here for the sake of
completeness.  For simplicity, if $(\sset,T,\bigcup_{e\in
  E(T)}\pset_e)$ is a tree-of-sets system in $G$, with parameters
$w,\ell,\alphawl$ as in the theorem statement, and for each $S_i\in
\sset$, $S_i\cap \tset=\emptyset$, then we say that it is a \emph{good
  tree-of-sets system}.
  
\subsection{High-Level Overview}  
In this subsection we provide a high-level overview and intuition for
the proof of Theorem~\ref{thm: meta-tree}. We also describe a
non-constructive proof of the theorem, which is somewhat simpler than
the constructive proof that appears below. This high-level description
oversimplifies some parts of the algorithm for the sake of
clarity. This subsection is not necessary for understanding the
algorithm and is only provided for the sake of intuition. A formal
self-contained proof appears in the following subsections.

Recall that the starting point is a graph $G=(V,E)$ and a set 
 $\tset \subseteq V$ of $k$ terminals, such that $\tset$ is node-well-linked in
$G$. Set $\tset$ certifies that $G$ has treewidth $\Omega(k)$. There can
be portions of the graph that are not well-connected to $\tset$ and
hence are irrelevant to its well-linkedness property.  We can assume
without loss of generality that $G$ is edge-minimal subject to
satisfying the condition that $\tset$ is node-well-linked.  However,
there is no easy structural or algorithmic way to characterize this
minimality condition. For this reason, in various parts of the proof,
we will delete or suppress irrelevant portions of the graph.  Recall
that the goal is to prove that given $G$ and $\tset$, there is a
tree-of-sets system with appropriate parameters.  Loosely speaking, a
tree-of-sets system with parameters $\ell,w, \alphawl$ consists of $\ell$
vertex-disjoint subgraphs with vertex sets $S_1,\ldots, S_{\ell}$ stitched
together with collections of paths in a tree-like fashion. From the
definition we note that each $S_i$ has the property that $G[S_i]$ contains
a well-linked vertex set of size $\Omega(\alphawl \cdot w/\Delta)$.  Thus, we
need as a building block, a procedure that allows us to take a graph
$G$ with a well-linked set of size $k$ and decomposes $G$ into $\ell$
disjoint subgraphs each of which has a well-linked set of size
$\Omega(\alphawl \cdot w/\Delta)$.  The fact that this can be done was
first shown in \cite{Chuzhoy11}, and stated explicitly with additional
refinements in \cite{ChekuriC13}. We make the discussion more precise below.

The proof uses two main parameters: $\ell_0=\ell^2$, and $w_0=w\cdot
\poly(\ell\cdot \Delta\cdot\log k)$. We say that a subset $S$ of vertices
of $G$ is a \emph{good router} if and only if the following three conditions
hold: (1) $S\cap \tset=\emptyset$; (2) $S$ has the
$\alphaWL$-bandwidth property; and (3) $S$ can send a large amount of
flow (say at least $w_0/2$ flow units) to $\tset$ with no
edge-congestion in $G$. A collection of $\ell_0$ disjoint good routers is
called a \emph{good family of routers}. Roughly, the proof consists of
two parts. The first part shows how to find a good family of routers,
and the second part shows that, given a good family routers, we can
build a good tree-of-sets system. We start by describing the second
part, which is somewhat simpler.

\subsubsection*{From a Good Family of Routers to a Good Tree-of-Sets System} 
Suppose we are given a good family $\rset=\set{S_1,\ldots,S_{\ell_0}}$ of
routers. We now give a high-level description of an algorithm to
construct a good tree-of-sets system from $\rset$ (a formal proof
appears in Section~\ref{sec: proof of iteration theorem}). The
algorithm consists of two phases. We start with the first phase.

Since every set $S_i\in \rset$ can send $w_0/2$ flow units to the
terminals with no edge-congestion, and the terminals are
$1$-well-linked in $G$, it is easy to see that every pair $S_i,S_j\in
\rset$ of sets can send $w_0/2$ flow units to each other with
edge-congestion at most $3$, and so there are at least
$\frac{w_0}{6\Delta}$ node-disjoint paths connecting $S_i$ to
$S_j$. We build an auxiliary graph $H$ from $G$, by contracting each cluster
$S_i\in \rset$ into a super-node $v_i$. We view the super-nodes
$v_1,\ldots,v_{\ell_0}$ as the terminals of $H$, and denote
$\ttset=\set{v_1,\ldots,v_{\ell_0}}$. We then use standard splitting
procedures in graph $H$ repeatedly, to obtain a new graph $H'$, whose
vertex set is $\ttset$, every pair of vertices remains
$\frac{w_0}{\poly(\Delta)}$-edge-connected, and every edge
$e=(v_i,v_j)\in E(H')$ corresponds to a path $P_e$ in $G$, connecting
a vertex of $S_i$ to a vertex of $S_j$. Moreover, the paths
$\set{P_e\mid e\in E(H')}$ are node-disjoint, and they do not contain
the vertices of $\bigcup_{S\in \rset}S$ as inner vertices. More
specifically, graph $H$ is obtained from $H'$ by first performing a
sequence of edge contraction and edge deletion steps that preserve
element-connectivity of the terminals, and then performing standard
edge-splitting steps that preserves edge-connectivity. Let $Z$ be a
graph whose vertex set is $\ttset$, and there is an edge $(v_i,v_j)$
in $Z$ if and only if there are many (say $\frac{w_0}{\ell_0^2\poly(\Delta)}$)
parallel edges $(v_i,v_j)$ in $H'$. We show that $Z$ is a connected
graph, and so we can find a spanning tree $T$ of $Z$. Since $\ell_0=\ell^2$,
either $T$ contains a path of length $\ell$, or it contains at least $\ell$
leaves. Consider the first case, where $T$ contains a path $P$ of
length $\ell$. We can use the path $P$ to define a tree-of-sets system
(in fact, it will give a path-of-sets system directly, after we apply
Theorem~\ref{thm: grouping} to boost the well-linkedness of the boundaries of the clusters that participate in $P$, and Theorem~\ref{thm: linkedness
  from node-well-linkedness} to ensure the linkedness of the
corresponding vertex subsets inside each cluster). From now on, we
focus on the second case, where $T$ contains $\ell$ leaves. Assume
without loss of generality that the good routers that are associated with the leaves of
$T$ are $\rset'=\set{S_1,\ldots,S_{\ell}}$. We show that we can find, for
each $1\leq i\leq \ell$, a subset $E_i\subseteq \out_G(S_i)$ of
$w_3=w\poly(\ell\cdot \Delta)$ edges, such that for each pair $1\leq
i<j\leq \ell$, there are $w_3$ node-disjoint paths connecting $S_i$ to
$S_j$ in $G$, where each path starts with an edge of $E_i$ and ends
with an edge of $E_j$. In order to compute the sets $E_i$ of edges, we
show that we can simultaneously connect each set $S_i$ to the set
$S^*\in \rset$ corresponding to the root of tree $T$ with many
paths. For each $i$, let $\pset_i$ be the collection of paths
connecting $S_i$ to $S^*$. We will ensure that all paths in
$\bigcup_i\pset_i$ are node-disjoint. The existence of the sets
$\pset_i$ of paths follows from the fact that all sets $S_i$ can
simultaneously send large amounts of flow to $S^*$ (along the
leaf-to-root paths in the tree $T$) with relatively small
congestion. After boosting the well-linkedness of the endpoints of
these paths in $S^*$ using Theorem~\ref{thm: grouping} for each
$\pset_i$ separately, and ensuring that, for every pair
$\pset_i,\pset_j$ of such path sets, their endpoints are linked inside
$S^*$ using Theorem~\ref{thm: linkedness from node-well-linkedness},
we obtain somewhat smaller subsets $\pset'_i\subseteq\pset_i$ of paths
for each $i$. The desired set $E_i$ of edges is obtained by taking the
first edge on every path in $\pset'_i$. We now proceed to the second
phase.

The execution of the second phase is very similar to the execution of
the first phase, except that the initial graph $H$ is built slightly
differently.  We will ignore the clusters in $\rset\setminus
\rset'$. For each cluster $S_i\in \rset'$, we delete all edges in
$\out_G(S_i)\setminus E_i$ from $G$, and then contract the vertices of
$S_i$ into a super-node $v_i$. As before, we consider the set
$\ttset=\set{v_1,\ldots,v_{\ell}}$ of supernodes to be the terminals of the
resulting graph $\tH$. Observe that now the degree of every terminal
$v_i$ is exactly $w_3$, and the edge-connectivity between every pair
of terminals is also exactly $w_3$. It is this additional property
that allows us to build the tree-of-sets system in this phase. As
before, we perform standard splitting operations to reduce graph $\tH$
to a new graph $\tH'$, whose vertex set is $\ttset$. As before, every
edge $e=(v_i,v_j)$ in $\tH'$ corresponds to a path $P_e$ connecting a
a vertex of $S_i$ to a vertex of $S_j$ in $G$; all paths in  $\set{P_e\mid e\in E(\tH')}$ are node-disjoint, and they do not
contain the vertices of $\bigcup_{S\in \rset'}S$ as inner
vertices. However, we now have the additional property that the degree
of every vertex $v_i$ in $\tH'$ is $w_3$, and the edge-connectivity of
every pair of vertices is also $w_3$. We build a graph $\tZ$ on the
set $\ttset$ of vertices as follows: for every pair $(v_i,v_j)$ of
vertices, if there number of edges $(v_i,v_j)$ in $\tH'$ is
$n_{i,j}>w_3/\ell^3$, then we add $n_{i,j}$ parallel edges $(v_i,v_j)$ to
$\tZ$. Otherwise, if $n_{i,j}<w_3/\ell^3$, then we do not add an edge
connecting $v_i$ to $v_j$. We then show that the degree of every
vertex in $\tZ$ remains very close to $w_3$, and the same holds for
edge-connectivity of every pair of vertices in $\tZ$. Note that every
pair $v_i,v_j$ of vertices of $\tZ$ is either connected by many
parallel edges, or there is no edge $(v_i,v_j)$ in $\tZ$. In the final
step, we show that we can construct a spanning tree of $\tZ$ with
maximum vertex degree bounded by $3$. This spanning tree immediately
defines a good tree-of-sets system. The construction of the spanning
tree is performed using a result of Singh and Lau~\cite{Singh-Lau},
who showed an approximation algorithm for constructing a
minimum-degree spanning tree of a graph. Their algorithm is based on
an LP-relaxation of the problem. They show that, given a feasible
solution to the LP-relaxation, one can construct a spanning tree with
maximum degree bounded by the maximum fractional degree plus
$1$. Therefore, it is enough to show that there is a solution to the
LP-relaxation on graph $\tZ$, where the fractional degree of every
vertex is bounded by $2$. The fact that the degree of every vertex,
and the edge-connectivity of every pair of vertices are very close to
the same value allows us to construct such a solution.

An alternative way of seeing that graph $\tZ$ has a spanning tree of
degree at most $3$ is to observe that graph $\tZ$ is $1$-tough (that
is, if we remove $q$ vertices from $\tZ$, there are at most $q$
connected components in the resulting graph, for every $q$). It is known
that a $1$-tough graph has a spanning tree of degree at most
$3$~\cite{graph-toughness}.

 \subsubsection*{Finding a Good Family of Routers}
 One of the main tools that we use in this part is a good clustering
 of the graph $G$ and a legal contracted graph associated with it. We
 say that a subset $C\subseteq V(G)$ of vertices is a small cluster
 if and only if $|\out(C)|\leq w_0$, and we say that it is a large cluster
 otherwise. A partition $\cset$ of $V(G)$ is called a \emph{good
   clustering} if and only if each terminal $t\in \tset$ belongs to a separate
 cluster $C_t\in \cset$, where $C_t=\set{t}$, all clusters in $\cset$
 are small, and each cluster has the $\alphawl$-bandwidth
 property. Given a good clustering $\cset$, the corresponding legal
 contracted graph is obtained from $G$ by contracting every cluster
 $C\in \cset$ into a super-node $v_C$ (notice that terminals are not
 contracted, since each terminal is in a separate cluster). The legal
 contracted graph can be seen as a model of $G$, where we ``hide''
 some irrelevant parts of the graph inside the contracted clusters.
 The main idea of the algorithm is to exploit the legal contracted
 graph in order to find a good family of routers, and, if we fail to
 do so, to construct a smaller legal contracted graph. We start with a
 non-constructive proof of the existence of a good family of routers
 in $G$.

\paragraph{Non-Constructive Proof} 
We assume that $G$ is minimal inclusion-wise, for which the set
$\tset$ of terminals is $1$-well-linked. That is, for an edge $e\in
E(G)$, if we delete $e$ from $G$, then $\tset$ is not $1$-well-linked
in the resulting graph.  Let $\cset^*$ be a good clustering of $V(G)$
minimizing the total number of edges in the corresponding legal
contracted graph (notice that a partition where every vertex belongs
to a separate cluster is a good clustering, so such a clustering
exists).  Consider the resulting legal contracted graph $G'$. The
degree of every vertex in $G'$ is at most $w_0$, and, from the
well-linkedness of the terminals in $G$, it is not hard to show that
$G'\setminus \tset$ must contain at least $\Omega(k)$ edges. Then
there is a partition $\set{X_1,\ldots,X_{\ell_0}}$ of $V(G')\setminus
\tset$, where for each $1\leq i\leq \ell_0$, $|\out_{G'}(X_i)|<O(\ell_0
|E_{G'}(X_i)|)$ (a random partition of $V(G')\setminus \tset$ into
$\ell_0$ subsets will have this property with constant probability. This
is since, if we denote $m=|E(G')\setminus \tset|$, then we expect
roughly $\frac{m+k}{\ell_0}$ edges in set $\out_{G'}(X_i)$, and roughly
$m/\ell_0^2$ edges with both endpoints inside $X_i$.)

For each set $X_i$, let $X'_i\subseteq V(G)\setminus \tset$ be the
corresponding subset of vertices of $G$, obtained by un-contracting
each supernode $v_C$ (that is, $X'_i=\bigcup_{v_C\in X_i}C$). If
$\Gamma_i$ is the interface of $X'_i$ in $G$, then we still have that
$|\Gamma_i|\leq O(\ell_0 |E_{G'}(X_i)|)$. As our next step, we would like to
find a partition $\wset_i$ of the vertices of $X'_i$ into clusters,
such that each cluster $W\in \wset_i$ has the $\alphawl$-bandwidth
property, and the total number of edges connecting different clusters
is at most $O(|\Gamma_i|/\ell_0)<|E_{G'}(X_i)|$. We call this procedure
bandwidth-decomposition. Assume first that we are able to find such a
decomposition. We claim that $\wset_i$ must contain at least one good
router $S_i$. If this is the case, then we have found the desired
family $\set{S_1,\ldots,S_{\ell_0}}$ of good routers. In order to show
that $\wset_i$ contains a good router, assume first that at least one
cluster $S_i\in \wset_i$ is large. The decomposition $\wset_i$ already
guarantees that $S_i$ has the $\alphaWL$-bandwidth property. If $S_i$
is not a good router, then it must be impossible to send large amounts
of flow from $S_i$ to $\tset$ in $G$. In this case, using known
techniques (see appendix of \cite{ChekuriNS13}), we can show that we
can delete an edge from $G[S_i]$, while preserving the
$1$-well-linkedness of the terminals\footnote{The technical statement
  here is that if there is a small cut separating two large
  well-linked sets in a graph then there is an edge that can be
  removed without affecting the well-linkedness of one of the sets.},
contradicting the minimality of $G$. Therefore, if $\wset_i$ contains
at least one large cluster, then it contains a good router. Assume now
that all clusters in $\wset_i$ are small. Then we show a new good
clustering $\cset'$ of $V(G)$, whose corresponding contracted graph
contains fewer edges than $G'$, leading to a contradiction. The new
clustering contains all clusters $C\in \cset^*$ with $C\cap
X_i'=\emptyset$, and all clusters in $\wset_i$. In other words, we
replace the clusters contained in $X_i'$ with the clusters of
$\wset_i$. The reason the number of edges goes down in the legal
contracted graph is that the total number of edges connecting
different clusters of $\wset_i$ is less than $|E_{G'}(X_i)|$.

The final part of the proof that we need to describe is the
bandwidth-decomposition procedure. Given a cluster $X_i'$, we would
like to find a partition $\wset_i$ of $X_i'$ into clusters that have
the $\alphawl$-bandwidth property, such that the number of edges
connecting different clusters is bounded by $O(|\Gamma_i|/\ell_0)$. There
are by now standard algorithms for finding such a decomposition, where
we repeatedly select a cluster in $\wset_i$ that does not have the
desired bandwidth property, and partition it along a sparse cut
\cite{Raecke,CKS}. Unfortunately, since our bandwidth parameter
$\alphaWL$ is independent of $n$, such an approach can only work when
$|\Gamma_i|$ is bounded by $\poly(k)$, which is not necessarily true
in our case. In order to overcome this difficulty, as was done in
\cite{Chuzhoy11}, we slightly weaken the bandwidth condition, and define
a $(k,\alphawl)$-bandwidth property as follows: We say that cluster
$C$ with interface $\Gamma$ has the $(k,\alphawl)$-bandwidth property,
if and only if for every pair $A,B\subseteq \Gamma$ of equal-sized disjoint subsets,
with $|A|,|B|\leq k$, the minimum edge-cut separating $A$ from $B$ in
$G[C]$ has at least $\alphawl\cdot |A|$ edges. Alternatively, we can
send $|A|$ flow units from $A$ to $B$ inside $G[C]$ with
edge-congestion at most $1/\alphawl$. Notice that if $C$ does not have
the $(k,\alphawl)$-bandwidth property, then there is a partition
$(C_1,C_2)$ of $C$, and two disjoint equal-sized subsets $A\subseteq
\Gamma\cap C_1$, $B\subseteq \Gamma\cap C_2$, with $|A|,|B|\leq k$, such that
$|E_G(C_1,C_2)|<\alphawl\cdot |A|$. We call such a partition a
$(k,\alphawl)$-violating cut of $C$. Even if we weaken the definition
of the good routers, and replace the $\alphawl$-bandwidth property
with the weaker $(k,\alphawl)$-bandwidth property, we can still
construct a good tree-of-sets system from a family of good
routers. This is since the construction algorithm only uses the
$\alphawl$-bandwidth property of the routers in the weak sense, by
sending small amounts of flow (up to $k$ units) across the
routers. Given the set $X_i'$, we can now show that there is a
partition $\wset_i$ of $X_i'$ into clusters that have
the $(k,\alphawl)$-bandwidth property, such that the number of edges
connecting different clusters is bounded by $O(|\Gamma_i|/\ell_0)$.

\paragraph{Constructive Proof}
A constructive proof is more difficult, for the following two
reasons. First, given a large cluster $S_i$, that has the
$\alphawl$-bandwidth property, but cannot send large amounts of flow
to the terminals in $G$, we need an efficient algorithm for finding an edge that
can be removed from $G[S_i]$ without violating the $1$-well-linkedness
of the terminals. While we know that such an edge must exist, we do
not have a constructive proof that allows us to find it. The second
problem is related to the bandwidth-decomposition procedure. While we
know that, given $X_i'$, there is a desired partition $\wset_i$ of
$X_i$ into clusters that have the $(k,\alphawl)$-bandwidth property,
we do not have an algorithmic version of this result. In particular,
we need an efficient algorithm that finds a $(k,\alphawl)$-violating cut in a
cluster that does not have the $(k,\alphawl)$-bandwidth property. (An efficient
algorithm that gives a $\poly(\log k)$ approximation, by returning an
$(\Omega(k),\alphawl\cdot \poly\log k)$-violating cut would be sufficient,
but as of now we do not have such an algorithm).

In addition to a good clustering defined above, our algorithm uses a
notion of acceptable clustering. An acceptable clustering is defined
exactly like a good clustering, except that large clusters are now
allowed. Each small cluster in an acceptable clustering must have the
$\alphawl$-bandwidth property, and each large cluster must induce a
connected graph in $G$.

In order to overcome the difficulties described above, we define a
potential function $\phi$ over partitions $\cset$ of $V(G)$. Given
such a partition $\cset$, $\phi(\cset)$ is designed to be a good
approximation of the number of edges connecting different clusters of
$\cset$. Additionally, $\phi$ has the following two useful
properties. If we are given an acceptable clustering $\cset$, a large
cluster $C\in \cset$, and a $(k,\alphawl)$-violating cut $(C_1,C_2)$
of $C$, then we can efficiently find a new acceptable clustering
$\cset'$ with $\phi(\cset')<\phi(\cset)-1/n$. Similarly, if we are
given an acceptable clustering $\cset$, and a large cluster $C\in
\cset$, such that $C$ cannot send $w_0/2$ flow units to the terminals,
then we can efficiently find a new acceptable clustering $\cset'$ with
$\phi(\cset')<\phi(\cset)-1/n$.

The algorithm consists of a number of phases. In every phase, we start
with some good clustering $\cset$, where in the first phase,
$\cset=\set{\set{v}\mid v\in V(G)}$. In each phase, we either find a
good tree-of-sets system, of find a new good clustering $\cset'$, with
$\phi(\cset')\leq\phi(\cset)-1$. Therefore, after $O(|E(G)|)$ phases, we
are guaranteed to find a good tree-of-sets system.

We now describe an execution of each phase. Let $\cset$ be the current
good clustering, and let $G'$ be the corresponding legal contracted
graph. As before, we find a partition $\set{X_1,\ldots,X_{\ell_0}}$ of
$V(G')\setminus \tset$, where for each $1\leq i\leq \ell_0$,
$|\out_{G'}(X_i)|<O(\ell_0 |E_{G'}(X_i)|)$, using a simple randomized
algorithm.  For each set $X_i$, let $X'_i\subseteq V(G)\setminus
\tset$ be the corresponding set of vertices in $G$, obtained by
un-contracting each supernode $v_C\in X_i$. For each $1\leq i\leq
\ell_0$, we also construct an acceptable clustering $\cset_i$, containing
all clusters $C\in \cset$ with $C\cap X'_i=\emptyset$, and all
connected components of $G[X_i]$ (if any such connected component is a
small cluster, we further partition it into clusters with
$\alphawl$-bandwidth property). We show that $\phi(\cset_i)\leq
\phi(\cset)-1$ for each $i$. We then perform a number of iterations.

In each iteration, we are given as input, for each $1\leq i\leq \ell_0$,
an acceptable clustering $\cset_i$, with $\phi(\cset_i)\leq
\phi(\cset)-1$, where each large cluster of $\cset_i$ is contained in
$X'_i$. An iteration is executed as follows. If, for some $1\leq i\leq
\ell_0$, the clustering $\cset_i$ contains no large clusters, then
$\cset_i$ is a good clustering, with $\phi(\cset_i)\leq
\phi(\cset)-1$. We then finish the current phase and return the good
clustering $\cset_i$. Otherwise, for each $1\leq i\leq \ell_0$, there is
at least one large cluster $S_i\in \cset_i$. We treat the clusters
$\set{S_1,\ldots,S_{\ell_0}}$ as a potential good family of routers, and
try to construct a tree-of-sets system using them. If we succeed in
building a good tree-of-sets system, then we are done, and we
terminate the algorithm. Otherwise, we will obtain a certificate that
one of the clusters $S_i$ is not a good router. The certificate is
either a $(k,\alphawl)$-violating partition of $S_i$, or a small cut
(containing fewer than $w_0/2$ edges), separating $S_i$ from the
terminals. In either case, using the properties of the potential
function, we can obtain a new acceptable clustering $\cset'_i$ with
$\phi(\cset'_i)\leq \phi(\cset_i)-1/n$, to replace the current
acceptable clustering $\cset_i$. We then continue to the next
iteration.

Overall, as long as we do not find a good tree-of-sets system, and do
not find a good clustering $\cset'$ with $\phi(\cset')\leq
\phi(\cset)-1$, we make progress by lowering the potential of one of
the acceptable clusterings $\cset_i$ by at least $1/n$. Therefore,
after polynomially-many iterations, we are guaranteed to complete the
phase.

We note that Theorem 6 in~\cite{Chuzhoy11} provides an algorithm,
that, given a cluster $X_i'$, and an access to an oracle for computing
$(k,\alphawl)$-violating cuts, produces a partition $\wset_i$ of
$X'_i$ into clusters that have the $(k,\alphawl)$-bandwidth property,
with the number of edges connecting different clusters suitably
bounded. The bound on the number of edges is computed by using a
charging scheme. The potential function that we use here, whose definition may appear non-intuitive, is
modeled after this charging scheme.

In the following subsections, we provide a formal proof of
Theorem~\ref{thm: meta-tree}. We start by defining the different types
of clusterings that we use and the potential function, and analyze its
properties. We then turn to describe the algorithm itself.

\label{------------------------------------sec: legal contracted graphs---------------------------------------}
\subsection{Vertex Clusterings and Legal Contracted Graphs}
\label{subsec: clusterings}
Let $n=|V(G)|$.  Our algorithm uses a parameter $\ell_0=\ell^2$. We use the
following two parameters for the bandwidth property:
$\alpha=\frac{1}{2^{11}\ell_0\log k}$, used to perform
bandwidth-decomposition of clusters, and
$\alphaWL=\frac{\alpha}{\alphasc(k)}=\Omega\left (\frac 1
  {\ell^2\log^{1.5}k}\right )$ - the value of the bandwidth parameter we
achieve.  Finally, we use a parameter $w_0=\frac{k}{192\ell_0^3\log k}$.
We say that a cluster $C\sse V(G)$ is \emph{large} if and only if $|\out(C)|\geq
w_0$, and we say that it is \emph{small} otherwise.  From the
statement of the theorem, we can assume that $w_0>\Delta$, and:

\begin{equation}
w=O\left (\frac{w_0\alpha^2}{\ell^9\Delta^8\log k}\right ). \label{eq: bound on h}
\end{equation}


Next, we define acceptable and good vertex clusterings and legal
contracted graphs, exactly as in~\cite{ChuzhoyL12}.

\begin{definition}
  Given a partition $\cset$ of the vertices of $V(G)$ into clusters,
  we say that $\cset$ is an \emph{acceptable clustering} of $G$ iff:

\begin{itemize}
\item  Every terminal $t\in \tset$ is in a separate cluster, that is, $\set{t}\in \cset$; 
\item Each small cluster $C\in \cset$ has the the $\alphaWL$-bandwidth property; and 
\item For each large cluster $C\in \cset$, $G[C]$ is connected.
\end{itemize}

An acceptable clustering that contains no large clusters is called a \emph{good clustering}.
\end{definition}

\begin{definition}
  Given a {\bf good} clustering $\cset$ of $G$, a graph $H_{\cset}$ is
  a legal contracted graph of $G$ associated with $\cset$, if and only if we can
  obtain $H_{\cset}$ from $G$ by contracting every $C\in\cset$ into a
  super-node $v_C$. We remove all self-loops, but we do not remove
  parallel edges. (Note that the terminals are not contracted since
  each terminal has its own cluster).
\end{definition}



\begin{claim}\label{claim: legal graph has many edges}
  If $G'$ is a legal contracted graph for $G$, then $G'\setminus
  \tset$ contains at least $k/3$ edges.
\end{claim}

\begin{proof}
  For each terminal $t\in \tset$, let $e_t$ be the unique edge
  adjacent to $t$ in $G'$, and let $u_t$ be the other endpoint of
  $e_t$. We partition the terminals in $\tset$ into groups, where two
  terminals $t,t'$ belong to the same group if and only if $u_t=u_{t'}$. Let
  $\gset$ be the resulting partition of the terminals. Since the
  degree of every vertex in $G'$ is at most $w_0$, each group $U\in
  \gset$ contains at most $w_0$ terminals. Next, we partition the
  terminals in $\tset$ into two subsets $X,Y$, where $|X|,|Y|\geq
  k/3$, and for each group $U\in \gset$, either $U\sse X$, or $U\sse
  Y$ holds. We can find such a partition by greedily processing each
  group $U\in \gset$, and adding all terminals of $U$ to one of the
  subsets $X$ or $Y$, that currently contains fewer
  terminals. Finally, we remove terminals from set $X$ until
  $|X|=k/3$, and we do the same for $Y$. Since the set $\tset$ of
  terminals is node-well-linked in $G$, it is $1$-edge-well-linked in
  $G'$, so we can route $k/3$ flow units from $X$ to $Y$ in $G'$, with
  no edge-congestion. Since no group $U$ is split between the two sets
  $X$ and $Y$, each flow-path must contain at least one edge of
  $G'\setminus \tset$. Therefore, $|E(G'\setminus \tset)|\geq k/3$.
\end{proof}

Given a partition $\cset$ of the vertices of $G$, we define a
potential $\phi(\cset)$ for this clustering, exactly as
in~\cite{ChuzhoyL12}. The idea is that $\phi(\cset)$ will serve as a
tight bound on the number of edges connecting different clusters in
$\cset$. At the same time, the potential function is designed in such
a way, that we can perform a number of useful operations on the
current clustering, without increasing the potential.

Suppose we are given {\bf any} partition $\cset$ of the vertices of
$G$. 
We define $\phi(\cset)$ as $\sum_{e\in E(G)}\phi(\cset,e)$ where
$\phi(\cset,e)$ assigns a potential to each edge $e$; to avoid
notational overload we use $\phi(e)$ for $\phi(\cset,e)$.  If both
endpoints of $e$ belong to the same cluster of $\cset$, then we set
its potential $\phi(e)=0$. Otherwise, if $e=(u,v)$, and $u\in C$ with
$|\out(C)|=z$, while $v\in C'$ with $|\out(C')|=z'$, then we set
$\phi(e)=1+\rho(z)+\rho(z')$ where $\rho$ is a non-decreasing
real-valued function that we define below. We think of $\rho(z)$ as
the contribution of $u$, and $\rho(z')$ the contribution of $v$ to
$\phi(e)$. The function $\rho$ will be chosen to give a small (compared to $1$)
contribution to $\phi(e)$ depending on the out-degree of 
the clusters that $e$ connects.

  For an integer $z>0$, we define a potential $\rho(z)$, as
follows. For $z< w_0$, $\rho(z)=4\alpha\log z$.  In order to define $\rho(z)$
for $z\geq w_0$, we consider the sequence $\set{n_0,n_1,\ldots}$ of
numbers, where $n_i=\left(\frac 3 2\right )^i w_0$. The potentials for
these numbers are $\rho(n_0)=\rho(w_0)=4\alpha\log w_0+4\alpha$, and
for $i>0$, $\rho(n_i)=4\frac{\alpha w_0}{n_i}+\rho(n_{i-1})$. Notice
that for all $i$, $\rho(n_i)\leq 12\alpha+4\alpha\log w_0\leq
8\alpha\log w_0\leq \frac{1}{2^8\ell_0}$.  We now partition all integers
$z>w_0$ into sets $Z_1,Z_2,\ldots$, where set $Z_i$ contains all
integers $z$ with $n_{i-1}\leq z< n_i$. For $z\in Z_i$, we define
$\rho(z)=\rho(n_{i-1})$.  This finishes the definition of $\rho$. 
Clearly, for all $z$, $\rho(z)\leq \frac{1}{2^8\ell_0}$.

\begin{observation}
 For a partition $\cset$ of the vertices of $G$ and for an edge $e=(u,v)\in E(G)$, if $u,v$ belong to the same cluster of $\cset$, then $\phi(e)=0$. Otherwise,  $1\leq \phi(e)\leq 1.1$.   
\end{observation}

Suppose we are given a partition $\cset$ of $V(G)$. The following
theorem allows us to partition a small cluster $C$ into a collection
of sub-clusters, each of which has the $\alphawl$-bandwidth property,
without increasing the overall potential. We call this procedure a
\emph{bandwidth decomposition}.

\begin{theorem}\label{thm: well-linked decomposition of small clusters}
  Let $\cset$ be a partition of $V(G)$, and let $C\in \cset$ be any
  small cluster, such that $G[C]$ is connected. Then there is an
  efficient algorithm that finds a partition $\wset$ of $C$ into small
  clusters, such that each cluster $R\in \wset$ has the
  $\alphaWL$-bandwidth property, and additionally, if $\cset'$ is a
  partition obtained from $\cset$ by removing $C$ and adding the
  clusters of $\wset$ to it, then $\phi(\cset')\leq \phi(\cset)$.
\end{theorem}

\begin{proof}
We maintain a partition $\wset$ of $C$ into small clusters, where at the beginning, $\wset=\set{C}$. We then perform a number of iterations.

In each iteration, we select a cluster $S\in \wset$, and set up the
following instance of the sparsest cut problem. Let $\Gamma$ be the
set of the interface vertices of $S$ in $G$. We then consider the
graph $G[S]$, where the vertices of $\Gamma$ serve as terminals. We
run the algorithm \algsc on the resulting instance of the sparsest cut
problem. If the sparsity of the cut produced by the algorithm is less
than $\alpha$, then we obtain a partition $(X,Y)$ of $S$, with
$|E(X,Y)|<\alpha\cdot\min\set{|\Gamma\cap X|,|\Gamma\cap Y|}\leq
\alpha\cdot\min\set{|\out(S)\cap \out(X)|,|\out(S)\cap \out(Y)|}$. In
this case, we remove $S$ from $\wset$, and add $X$ and $Y$ to $\wset$
instead. Notice that since $S$ is a small cluster, $X$ and $Y$ are also small clusters. The algorithm ends when for every cluster $S\in \wset$,
algorithm $\algsc$ returns a partition of sparsity at least
$\alpha$. We are then guaranteed that every cluster in $\wset$ has the
$\alpha/\alphasc(k)=\alphaWL$-bandwidth property, and it is easy to
verify that all resulting clusters are small.

It now only remains to show that the potential does not increase. Each
iteration of the algorithm is associated with a partition of the
vertices of $G$, obtained from $\cset$ by removing $C$ and adding all
clusters of the current partition $\wset$ of $C$ to it. It suffices
to show that if $\cset'$ is the current partition of $V(G)$, and
$\cset''$ is the partition obtained after one iteration, where a set
$S\in \cset$ was replaced by two sets $X$ and $Y$, then
$\phi(\cset'')\leq \phi(\cset')$.

Assume without loss of generality that $|\out(X)|\leq |\out(Y)|$, so
$|\out(X)|\leq 2|\out(S)|/3$. Let $z=|\out(S)|,z_1=|\out(X)|,
z_2=|\out(Y)|$, and recall that $z,z_1,z_2<w_0$.  The changes to the
potential are the following:

\begin{itemize}
\item The potential of the edges in $\out(Y)\cap \out(S)$ only goes down.

\item The potential of every edge in $\out(X)\cap \out(S)$ goes down
  by $\rho(z)-\rho(z_1)=4\alpha\log z-4\alpha\log z_1=4\alpha\log
  \frac{z}{z_1}\geq 4\alpha\log 1.5\geq 2.3\alpha$, since $z_1\leq
  2z/3$. So the total decrease in the potential of the edges in
  $\out(X)\cap \out(S)$ is at least $2.3\alpha\cdot |\out(X)\cap
  \out(S)|$.

\item The edges in $E(X,Y)$ did not contribute to the potential
  initially, and now contribute $1+\rho(z_1)+\rho(z_2)\leq 2$
  each. Notice that $|E(X,Y)|\leq \alpha \cdot |\out(X)\cap \out(S)|$,
  and so they contribute at most $2\alpha \cdot |\out(X)\cap \out(S)|$
  in total.
\end{itemize}
Clearly, the overall potential decreases.
\end{proof}

Assume that we are given an acceptable clustering $\cset$ of $G$. We now define two operations on $G$,  each of which produces a new acceptable clustering of $G$, whose potential is strictly smaller than $\phi(\cset)$.

{\bf Action 1: Partitioning a large cluster.}
Suppose we are given a large cluster $C$, and let $\Gamma$ be the interface of $C$ in $G$. We say that a partition $(X,Y)$ of $C$ is a $(w_0,\alpha)$-violating partition, if and only if there are two subsets $\Gamma_X\subseteq \Gamma\cap X,\Gamma_Y\sse \Gamma\cap Y$ of vertices, with $|\Gamma_X|+|\Gamma_Y|\leq w_0$, and $|E(X,Y)|<\alpha\cdot \min\set{|\Gamma_X|,|\Gamma_Y|}$. Equivalently, $(X,Y)$ is an $(w_0,\alpha)$-violating partition, if and only if $|E(X,Y)|<\alpha\cdot \min\set{|\Gamma\cap X|,|\Gamma\cap Y|, \floor{w_0/2}}$.

Suppose we are given an acceptable clustering $\cset$ of $G$, a large cluster $C\in \cset$, and an $(w_0,\alpha)$-violating partition $(X,Y)$ of $C$. 
In order to perform this operation, we first replace $C$ with $X$ and $Y$ in $\cset$. If, additionally, either of the clusters, $X$ or $Y$, become small, then we perform a bandwidth decomposition of that cluster using Theorem~\ref{thm: well-linked decomposition of small clusters}, and update $\cset$ with the resulting partition. 
Clearly, the final partitioning $\cset'$ is an acceptable clustering. 
We denote this operation by $\partition(C,X,Y)$.

\begin{claim}\label{claim: bound on potential for partition}
Let $\cset'$ be the outcome of operation $\partition(C,X,Y)$. Then $\phi(\cset')<\phi(\cset)-1/n$. 
\end{claim}

\begin{proof}
  Let $\cset''$ be the clustering obtained from $\cset$, by replacing
  $C$ with $X$ and $Y$. From Theorem~\ref{thm: well-linked
    decomposition of small clusters}, it is sufficient to prove that
  $\phi(\cset'')<\phi(\cset)-1/n$.

Assume without loss of generality that $|\out(X)|\leq |\out(Y)|$.  Let
$z=|\out(C)|$, $z_1=|\out(X)|$, $z_2=|\out(Y)|$, so $z_1<
2z/3$. Assume that $z\in Z_i$. Then either $z_1\in Z_{i'}$ for $i'\leq
i-1$, or $z_1<w_0$. The potential of the edges in $\out(Y)\cap
\out(C)$ does not increase. The only other changes in the potential
are the following: the potential of each edge in $\out(X)\cap \out(C)$
decreases by $\rho(z)-\rho(z_1)$, and the potential of every edge in
$E(X,Y)$ increases from $0$ to at most $1.1$. We consider two cases.

First, if $z_1<w_0$, then $\rho(z)\geq 4\alpha+\rho(z_1)$. So the
potential of each edge in $\out(X)\cap \out(C)$ decreases by at least
$4\alpha$, and the overall decrease in potential due to these edges is
at least $4\alpha|\out(X)\cap \out(C)|$. The total increase in
potential due to the edges in $E(X,Y)$ is bounded by
$1.1|E(X,Y)|<1.1\alpha|\Gamma_X|\leq1.1\alpha|\out(X)\cap \out(C)|$,
so the overall potential decreases by at least $2\alpha|\out(X)\cap
\out(C)|>1/n$

The second case is when $z_1\geq w_0$. Assume that $z_1\in
Z_{i'}$. Then $n_{i'}\leq 3z_1/2$, and, since $i'\leq i-1$ must hold,
$\rho(z)\geq \frac{4\alpha w_0}{n_{i'}}+\rho(n_{i'-1})=\frac{4\alpha
  w_0}{n_{i'}}+\rho(z_1)\geq \frac{8\alpha w_0}{3z_1}+\rho(z_1)$. So
the potential of each edge in $\out(X)\cap \out(C)$ decreases by at
least $\frac{8\alpha w_0}{3z_1}$, and the total decrease in potential
due to these edges is at least $\frac{8\alpha
  w_0}{3z_1}\cdot|\out(X)\cap \out(C)|\geq \frac{4\alpha w_0}{3}$,
since $|\out(X)\cap \out(C)|\geq z_1/2$. The total increase in the
potential due to the edges in $E(X,Y)$ is bounded by
$1.1|E(X,Y)|<0.55\alpha w_0$, since $|E(X,Y)|\leq \alpha
w_0/2$. Overall, the total potential decreases by at least
$\frac{2\alpha w_0}{3}>1/n$.

\end{proof}

{\bf Action 2: Separating a large cluster.}
Let $\cset$ be an acceptable clustering, and let $C\in \cset$ be a large cluster in $\cset$. Assume further that we are given a partition $(A,B)$ of $V(G)$, with $C\sse A$, $\tset\sse B$, and $|E_G(A,B)|< w_0/2$. We perform the following operation, that we denote by $\separate(C,A)$.

Consider a cluster $S\in \cset$. If $S\setminus A\neq \emptyset$, and $|\out(S\setminus A)|>|\out(S)|$, then we modify $A$ by removing all vertices of $S$ from it. Notice that in this case, the number of edges in $E(S)$ that originally contributed to the cut $(A,B)$,  $|E(S\cap A,S\cap  B)|>|\out(S)\cap E(A)|$ must hold, so $|\out(A)|$ only goes down as a result of this modification. We assume from now on that if $|S\setminus A|\neq \emptyset$, then $|\out(S\setminus A)|\leq |\out(S)|$. In particular, if $S$ is a small cluster, and $S\setminus A\neq\emptyset$, then $S\setminus A$ is also a small cluster.


We build a new partition
$\cset'$ of $V(G)$ as follows. First, we add every connected component
of $G[A]$ to $\cset'$. Notice that all these clusters are small, as
$|\out(A)|<w_0/2$. Next, for every cluster $S\in\cset$, such that
$S\setminus A\neq \emptyset$, we add every connected component of
$G[S\setminus A]$ to $\cset'$. Notice that every terminal $t\in \tset$
is added as a separate cluster to $\cset'$. So far we have defined a
new partition $\cset'$ of $V(G)$. This partition may not be
acceptable, since we are not guaranteed that every small cluster of
$\cset'$ has the bandwidth property. In our final step, we perform the
bandwidth decomposition of every small cluster of $\cset'$, using
Theorem~\ref{thm: well-linked decomposition of small clusters}, and
obtain the final acceptable partition $\cset''$ of vertices of $G$.
Notice that if $S\in \cset''$ is a large cluster, then there must be
some {\bf large} cluster $S'$ in the original partition $\cset$ with
$S\sse S'$.

\begin{claim}\label{claim: bound on potential for separation}
 Let $\cset''$ be the outcome of operation  $\separate(C,A)$. Then $\phi(\cset'')\leq \phi(\cset)-1$. 
 \end{claim}

 \begin{proof}
   In order to prove the claim, it is enough to prove that
   $\phi(\cset')\leq \phi(\cset)-1$, since, from Theorem~\ref{thm:
     well-linked decomposition of small clusters}, bandwidth
   decompositions of small clusters do not increase the potential.

We now show that $\phi(\cset')\leq \phi(\cset)-1$.
We can bound the changes in the potential as follows:

\begin{itemize}
\item Every edge in $\out(A)$ contributes at most $1.1$ to the
  potential of $\cset''$, and there are at most $ \frac{w_0-1} 2$ such
  edges. 
\item Every edge in $\out(C)$ contributed at least $1$ to the
  potential of $\cset'$, and there are at least $w_0$ such edges,
  since $C$ is a large cluster.
\end{itemize}

For every other edge $e$, the potential of $e$ does not increase. Indeed, let $e=(u,v)$, where $u\in S_1$, $v\in S_2$, with $S_1,S_2\in \cset'$, and $S_1,S_2\not\subseteq A$. Then there are clusters $S_1',S_2'\in \cset$, with $S_1\subseteq S_1'$ and $S_2\subseteq S_2'$. Notice that $S_1'\neq S_2'$, since $S_1$ and $S_2$ correspond to connected components of $S_1'$ and $S_2'$, respectively, and so no edge can connect them. From our construction of $\cset'$, $|\out(S_1)|\leq |\out(S_1')|$ and $|\out(S_2)|\leq |\out(S_2')|$, so the potential of $e$ cannot increase.
Therefore, the decrease in the potential is at least $w_0-\frac{1.1(w_0-1)}2\geq 1$.
\end{proof}

To summarize, given an acceptable clustering $\cset$ of the vertices
of $G$, let $E'$ be the set of edges whose endpoints belong to
distinct clusters of $\cset$. Then $|E'|\leq \phi(\cset)\leq
1.1|E'|$. So the potential is a good estimate on the number of edges
connecting the different clusters. We have also defined two actions on
large clusters of $\cset$: $\partition(C,X,Y)$ can be performed if we are given a large cluster $C\in \cset$ and a
$(w_0,\alpha)$-violating partition $(X,Y)$ of $C$, and $\separate(C,A)$, where
$(A,V(G)\setminus A)$ is a partition of $V(G)$ with $|\out(A)|< w_0/2$, separating a large cluster $C$
from the terminals. Each such action returns a new acceptable
clustering, whose potential goes down by at least $1/n$. Both operations ensure that if $S$ is a large cluster in the new clustering, then there is some
large cluster $S'$ in the original clustering with $S\subseteq S'$.

\label{------------------------------------------------------sec: alg---------------------------------------------}
\subsection{The Algorithm}\label{sec: Alg}

We maintain, throughout the algorithm, a good clustering $\cset$ of $G$. Initially, $\cset$ is a partition of $V(G)$, where every vertex of $G$ belongs to a distinct cluster, that is, $\cset=\set{\set{v}\mid v\in V(G)}$. Clearly, this is a good clustering, as $\Delta<w_0$. The algorithm consists of a number of phases. In every phase, we start with some good clustering $\cset$ and the corresponding legal contracted graph $H_{\cset}$. The phase output is either a good tree-of-sets system, or another good clustering $\cset'$, such that $\phi(\cset')\leq \phi(\cset)-1$. 
In the former case, we terminate the algorithm, and output the tree-of-sets system. In the latter case, we continue to the next phase. After $O(|E(G)|)$ phases, our algorithm will then successfully terminate with a good tree-of-sets system. It is therefore enough to prove the following theorem.

\begin{theorem}\label{thm: find good crossbar or a smaller contracted graph}
Let $\cset$ be a good clustering of the vertices of $G$, and let $H_{\cset}$ be the corresponding legal contracted graph. Then there is an efficient randomized algorithm that with high probability either computes a good tree-of-sets system, or finds a new good clustering $\cset'$, such that $\phi(\cset')\leq \phi(\cset)-1$.
\end{theorem}

The rest of this subsection is dedicated to proving Theorem~\ref{thm:
  find good crossbar or a smaller contracted graph}.  We assume that
we are given a good clustering $\cset$ of the vertices of $G$, and the
corresponding legal contracted graph $G'=H_{\cset}$.

Let $m=|E(G'\setminus \tset)|$. From Claim~\ref{claim: legal graph has
  many edges}, $m\geq k/3$.  As a first step, we randomly partition
the vertices in $G'\setminus \tset$ into $\ell_0$ subsets
$X_1,\ldots,X_{\ell_0}$, where each vertex $v\in V(G')\setminus \tset$
selects an index $1\leq j\leq \ell_0$ independently uniformly at random,
and is then added to $X_j$. We need the following
claim. 


\begin{claim}\label{claim: random partition into gamma sets} With probability at least $\half$, for each $1\leq j\leq \ell_0$, $|\out_{G'}(X_j)|< \frac{10m}{\ell_0}$, while $|E_{G'}(X_j)|\geq \frac{m}{2\ell_0^2}$.
\end{claim}
 


\begin{proof}
  Let $H=G'\setminus\tset$.  Fix some $1\leq j\leq \ell_0$.  Let
  $\event_1(j)$ be the bad event that $\sum_{v\in X_j}d_{H}(v)>
  \frac{2m}{\ell_0}\cdot \left (1+\frac 1 {\ell_0}\right )$. In order to
  bound the probability of $\event_1(j)$, we define, for each vertex
  $v\in V(H)$, a random variable $x_v$, whose value is
  $\frac{d_{H}(v)}{w_0}$ if $v\in X_j$ and $0$ otherwise. Notice that
  $x_v\in [0,1]$, and the random variables $\set{x_v}_{v\in V(H)}$ are
  pairwise independent. Let $B=\sum_{v\in V(H)} x_v$. Then the
  expectation of $B$, $\mu_1=\sum_{v\in V(H)} \frac{d_{H}(v)}{\ell_0
    w_0}=\frac{2m}{\ell_0 w_0}$. Using the standard Chernoff bound (see
  e.g. Theorem 1.1 in~\cite{measure-concentration}),

\[\prob{\event_1(j)}=\prob{B> \left (1+1/\ell_0\right )\mu_1}\leq e^{-\mu_1/(3\ell_0^2)}=e^{-\frac{2m}{3\ell_0^3 w_0}}<\frac 1 {6\ell_0}\]

since $m\geq k/3$ and $w_0=\frac{k}{192\ell_0^3\log k}$.

For each terminal $t\in \tset$, let $e_t$ be the unique edge adjacent
to $t$ in graph $G'$, and let $u_t$ be its other endpoint. Let
$U=\set{u_t\mid t\in \tset}$. For each vertex $u\in U$, let $w(u)$ be
the number of terminals $t$, such that $u=u_t$. Notice that $w(u)\leq
w_0$ must hold. We say that a bad event $\event_2(j)$ happens iff
$\sum_{u\in U\cap X_j}w(u)\geq \frac k{\ell_0}\cdot \left (1+\frac 1
  {\ell_0}\right )$.  In order to bound the probability of the event
$\event_2(j)$, we define, for each $u\in U$, a random variable $y_u$,
whose value is $w(u)/w_0$ if and only if $u\in X_j$, and it is $0$
otherwise. Notice that $y_u\in [0,1]$, and the variables $y_u$ are
independent for all $u\in U$. Let $Y=\sum_{u\in U} y_u$. The
expectation of $Y$ is $\mu_2=\frac{k}{w_0\ell_0}$, and event
$\event_2(j)$ holds if and only if $Y\geq \frac{k}{w_0\ell_0}\cdot \left (1+\frac 1
  {\ell_0}\right )\geq \mu_2\cdot \left (1+\frac 1 {\ell_0}\right )$. Using
the standard Chernoff bound again, we get that:

\[\prob{\event_2(j)}\leq e^{-\mu_2/(3\ell_0^2)}\leq e^{-k/(3w_0\ell_0^3)}\leq \frac 1 {6\ell_0}\]

since $w_0=\frac{k}{192\ell_0^3\log k}$. Notice that if events $\event_1(j),\event_2(j)$ do not hold, then:

\[|\out_{G'}(X_j)|\leq \sum_{v\in X_j}d_{H}(v)+\sum_{u\in U\cap X_j}w(u)\leq \left (1+\frac 1 {\ell_0}\right )\left (\frac{2m}{\ell_0}+\frac{k}{\ell_0}\right )< \frac{10m}{\ell_0}\]

since $m\geq k/3$.

Let $\event_3(j)$ be the bad event that $|E_{G'}(X_j)|<
\frac{m}{2\ell_0^2}$. We next prove that $\prob{\event_3(j)}\leq \frac 1
{6\ell_0}$. We say that two edges $e,e'\in E(G'\setminus \tset)$ are
\emph{independent} if and only if they do not share endpoints. Our first step
is to compute a partition $U_1,\ldots,U_z$ of the set $E(G'\setminus
\tset)$ of edges, where $z\leq 2w_0$, such that for each $1\leq i\leq
z$, $|U_i|\geq \frac m{4w_0}$, and all edges in set $U_i$ are mutually
independent. In order to compute such a partition, we construct an
auxiliary graph $Z$, whose vertex set is $\set{v_e\mid e\in E(H)}$,
and there is an edge $(v_e,v_{e'})$ if and only if $e$ and $e'$ are not
independent. Since the maximum vertex degree in $G'$ is at most $w_0$,
the maximum vertex degree in $Z$ is bounded by $2w_0-2$. Using the
Hajnal-Szemer\'edi Theorem~\cite{Hajnal-Szemeredi}, we can find a
partition $V_1,\ldots,V_z$ of the vertices of $Z$ into $z\leq 2w_0$
subsets, where each subset $V_i$ is an independent set, and $|V_i|\geq
\frac{|V(Z)|}{z}-1\geq \frac{m}{4w_0}$. The partition $V_1,\ldots,V_z$
of the vertices of $Z$ gives the desired partition $U_1,\ldots,U_z$ of
the edges of $G'\setminus \tset$. For each $1\leq i\leq r$, we say
that the bad event $\event_3^i(j)$ happens if and only if $|U_i\cap E(X_j)|<
\frac{|U_i|}{2\ell_0^2}$. Notice that if $\event_3(j)$ happens, then
event $\event_3^i(j)$ must happen for some $1\leq i\leq z$. Fix some
$1\leq i\leq z$. The expectation of $|U_i\cap E(X_j)|$ is
$\mu_3=\frac{|U_i|}{\ell_0^2}$. Since all edges in $U_i$ are independent,
we can use the standard Chernoff bound to bound the probability of
$\event_3^i(j)$, as follows:

\[\prob{\event_3^i(j)}=\prob{|U_i\cap E(X_j)|<\mu_3/2}\leq e^{-\mu_3/8}=e^{-\frac{|U_i|}{8\ell_0^2}}\].

Since $|U_i|\geq \frac{m}{4w_0}$, $m\geq k/3$,
$w_0=\frac{k}{192\ell_0^3\log k}$, this is bounded by
$\frac{1}{k^2}\leq\frac{1}{12w_0\ell_0}$. We conclude that
$\prob{\event_3^i(j)}\leq \frac{1}{12w_0\ell_0}$, and by using the union
bound over all $1\leq i\leq z$, $\prob{\event_3(j)}\leq
\frac{1}{6\ell_0}$.

Using the union bound over all $1\leq j\leq \ell_0$, with probability at
least $\half$, none of the events
$\event_1(j),\event_2(j),\event_3(j)$ for $1\leq j\leq \ell_0$ happen,
and so for each $1\leq j\leq \ell_0$, $|\out_{G'}(X_j)|<
\frac{10m}{\ell_0}$, and $|E_{G'}(X_j)|\geq\frac{m}{2\ell_0^2}$.
\end{proof}

Given a partition $X_1,\ldots,X_{\ell_0}$, we can efficiently check
whether the conditions of Claim~\ref{claim: random partition into
  gamma sets} hold. If they do not hold, we repeat the randomized
partitioning procedure.  From Claim~\ref{claim: random partition into
  gamma sets}, we are guaranteed that with high probability, after $\poly(n)$
iterations, we will obtain a partition with the desired
properties. Assume now that we are given the partition
$X_1,\ldots,X_{\ell_0}$ of $V(G')\setminus \tset$, for which the
conditions of Claim~\ref{claim: random partition into gamma sets}
hold. Then for each $1\leq j\leq \ell_0$,
$|E_{G'}(X_j)|>\frac{|\out_{G'}(X_j)|}{20\ell_0}$. Let $X'_j\sse
V(G)\setminus \tset$ be the set obtained from $X_j$, after we
un-contract each cluster, that is, for each super-node $v_C\in X_j$,
we replace $v_C$ with the vertices of $C$. Notice that
$\set{X'_j}_{j=1}^{\ell_0}$ is a partition of $V(G)\setminus\tset$.

The plan for the rest of the proof is as follows. For each $1\leq
j\leq \ell_0$, we will maintain an acceptable clustering $\cset_j$ of the
vertices of $G$. That is, for each $1\leq j\leq \ell_0$, $\cset_j$ is a
partition of $V(G)$. In addition to being an acceptable clustering, it
will have the following property:

\begin{properties}{P}
\item If $C\in \cset_j$ is a large cluster, then $C\subseteq X'_j$. \label{property of Cj}
\end{properties}

The initial partition $\cset_j$, for $1\leq j\leq \ell_0$ is obtained as
follows. Recall that $\cset$ is the current good clustering of the
vertices of $G$, and every cluster $C\in \cset$ is either contained in
$X'_j$, or it is disjoint from it. First, we add to $\cset_j$ all
clusters $C\in \cset$ with $C\cap X'_j=\emptyset$. Next, we add to
$\cset_j$ all connected components of $G[X'_j]$. If any of these
components is a small cluster, then we perform the bandwidth
decomposition of this cluster, using Theorem~\ref{thm: well-linked
  decomposition of small clusters}, and update $\cset_j$
accordingly. Let $\cset_j$ be the resulting final partition. Clearly,
it is an acceptable clustering, with property (\ref{property of
  Cj}). Moreover, the following claim shows that $\phi(\cset_j)\leq
\phi(\cset)-1$:

\begin{claim}
  \label{claim:improve-potential}
  For each $1\leq j\leq \ell_0$, $\phi(\cset_j)\leq \phi(\cset)-1$.
\end{claim}

\begin{proof}
  Let $\cset'_j$ be the partition of $V(G)$, obtained as follows: we
  add to $\cset'_j$ all clusters $C\in \cset$ with $C\cap
  X_j'=\emptyset$, and we add all connected components of $G[X_j']$ to
  $\cset'_j$ (that is, $\cset'_j$ is obtained like $\cset_j$, except
  that we do not perform the bandwidth decompositions of the small
  clusters). From Theorem~\ref{thm: well-linked decomposition of small
    clusters}, it is enough to prove that $\phi(\cset'_j)\leq
  \phi(\cset)-1$. The changes of the potential from $\cset$ to
  $\cset'_j$ can be bounded as follows:

\begin{itemize}
\item The edges in $E_{G'}(X_j)$ contribute at least $1$ to
  $\phi(\cset)$ and contribute $0$ to $\phi(\cset'_j)$.
\item The potential of the edges in $\out_G(X'_j)$ may increase. The
  increase is at most $\rho(n)\leq \frac{1}{2^8\ell_0}$ per edge. So the
  total increase is at most $\frac{|\out_{G'}(X_j)|}{2^8\ell_0}\leq
  \frac{|E_{G'}(X_j)|}{4}$. These are the only edges whose potential
  may increase.
\end{itemize}

Overall, the decrease in the potential is at least $\frac{|E_{G'}(X_j)|}{2}\geq\frac{m}{4\ell_0^2}\geq \frac{k}{12\ell_0^2}\geq 1$.
\end{proof}

If any of the partitions $\cset_1,\ldots,\cset_{\ell_0}$ is a good
clustering, then we have found a good clustering $\cset'$ with
$\phi(\cset')\leq \phi(\cset)-1$. We terminate the algorithm and
return $\cset'$. Otherwise, we select an arbitrary large cluster
$S_j\in \cset_j$ for each $j$. We then consider the resulting
collection $S_1,\ldots,S_{\ell_0}$ of large clusters, and try to exploit
them to construct a good tree-of-sets system. Since for each $1\leq
j\leq \ell_0$, $S_j\sse X'_j$, the sets $S_1,\ldots,S_{\ell_0}$ are mutually
disjoint and they do not contain terminals.  Our algorithm performs a
number of iterations, using the following theorem.

 \begin{theorem}\label{thm: iteration}
There is an efficient randomized algorithm, that, given a collection $\set{S_1,\ldots,S_{\ell_0}}$ of disjoint vertex subsets of $G$, where for all $1\leq j\leq \ell_0$, $S_j\cap \tset=\emptyset$, with high probability computes one of the following: either
 
 \begin{itemize}
 \item a good tree-of-sets system in $G$; or
 
 \item a $(w_0,\alpha)$-violating partition $(X,Y)$ of $S_j$, for
   some $1\leq j\leq \ell_0$; or
 
 \item a partition $(A,B)$ of $V(G)$ with $S_j\sse A$, $\tset\sse B$ and
   $|E_G(A,B)|<w_0/2$, for some $1\leq j\leq \ell_0$.
 \end{itemize}
 \end{theorem}

 We provide the proof of Theorem~\ref{thm: iteration} in the following
 subsection, and complete the proof of Theorem~\ref{thm: find good
   crossbar or a smaller contracted graph} here. Suppose we are given
 a good clustering $\cset$ of the vertices of $G$. For each $1\leq
 j\leq \ell_0$, we compute an acceptable clustering $\cset_j$ of $V(G)$
 as described above. If any of the partitions $\cset_j$ is a good
 clustering, then we terminate the algorithm and return
 $\cset_j$. From the above discussion, $\phi(\cset_j)\leq
 \phi(\cset)-1$. Otherwise, for each $1\leq j\leq \ell_0$, we select any
 large cluster $S_j\in \cset_j$, and apply Theorem~\ref{thm:
   iteration} to the current family $\set{S_1,\ldots,S_{\ell_0}}$ of
large clusters. If the outcome of Theorem~\ref{thm:
   iteration} is a good tree-of-sets system, then we terminate the
 algorithm and return this tree-of-sets system, and we say that the
 iteration is successful. Otherwise, we apply the appropriate action:
 $\partition(S_j,X,Y)$, or $\separate(S_j,A)$ to the clustering
 $\cset_j$. As a result, we obtain an acceptable clustering
 $\cset'_j$, with $\phi(\cset'_j)\leq \phi(\cset_j)-1/n$. Moreover, it
 is easy to see that this clustering also has Property~(\ref{property
   of Cj}): if the $\partition$ operation is performed, then we only
 partition $S_j$; if the $\separate$ operation is performed, then for every
 large cluster $S$ in the new partition $\cset'_j$, there is a 
 large cluster $S'\in\cset_j$ with $S\subseteq S'$.

 If all clusters in $\cset'_j$ are small, then we can again terminate
 the algorithm with a good clustering $\cset'_j$, with
 $\phi(\cset'_j)\leq \phi(\cset_j)-1/n \le \phi(\cset)-1 - 1/n$
 (recall that Claim~\ref{claim:improve-potential} shows that $\phi(\cset_j) \le
 \phi(\cset)-1$). Otherwise, we select an arbitrary large cluster $S'_j\in
 \cset'_j$, and continue to the next iteration. Overall, as long as we
 do not complete a successful iteration, and we do not find a good
 clustering $\cset'$ of $V(G)$ with $\phi(\cset')\leq \phi(\cset)-1$,
 we make progress in each iteration by decreasing the potential of one
 of the partitions $\cset_j$ by at least $1/n$, by performing either a
 $\separate$ or a $\partition$ operation on one of the large clusters
 of $\cset_j$. After at most $1.1|E(G)|\cdot n\cdot \ell_0$ iterations we
 are then guaranteed to complete a successful iteration, or find a
 good clustering $\cset'$ with $\phi(\cset')\leq \phi(\cset)-1$, and
 finish the algorithm. Therefore, in order to complete the proof of
 Theorem~\ref{thm: find good crossbar or a smaller contracted graph}
 it is now enough to prove Theorem~\ref{thm: iteration}.

 \label{---------------------------------------------sec: iteration execution--------------------------------}
 \subsection{Proof of Theorem~\ref{thm: iteration}\label{sec: proof of iteration theorem}}
 
  Let $\rset=\set{S_1,\ldots,S_{\ell_0}}$.
  We start by checking that for each $1\leq j\leq \ell_0$, the vertices
  of $S_j$ can send $w_0/2$ flow units in $G$ to the terminals with no
  edge-congestion. If this is not the case for some set $S_j$, then
  there is a partition $(A,B)$ of $V(G)$ with $S_j\sse A$, $\tset\sse
  B$ and $|E_G(A,B)|<w_0/2$. We then return the partition $(A,B)$ of
  $V(G)$ and finish the algorithm.  From now on each
  set $S_j$ can send $w_0/2$ flow units in $G$ to the terminals with
  no edge-congestion. We assume without loss of generality that this flow is integral.

 Since the set $\tset$ of terminals is node-well-linked, every pair $(S_j,S_{j'})$ of vertex subsets from $\rset$ can send $w_0/2$ flow units to each other with edge-congestion at most $3$: concatenate the flows from $S_j$ to a subset $\tset_1$ of the terminals, from $S_{j'}$ to a subset $\tset_2$ of the terminals, and the flow between the two subsets of the terminals.
 Scaling this flow down by factor $3\Delta$ and using the integrality of flow, for each such pair $(S_j,S_{j'})$,  there are at least $\lfloor \frac{w_0}{6\Delta}\rfloor$ node-disjoint paths connecting $S_j$ to $S_{j'}$ in $G$. We can assume that these paths do not contain terminals, as the degree of every terminal in $G$ is $1$. 
 
 The algorithm consists of two phases. In the first phase, we attempt to construct a tree-of-sets system, using the collection $\rset$ of vertex subsets. If we fail to do so, we will either return an $(w_0,\alpha)$-violating cut in some cluster $S_j\in \rset$, or we will identify a subset $\rset'\subseteq \rset$ of $\ell$ clusters, and for each cluster $S_j\in \rset'$, a large subset $E_j\subseteq \out(S_j)$ of edges, such that for each $S_i,S_j\in \rset'$, there are many disjoint paths connecting the edges in $E_j$ to the edges in $E_i$ in $G$. In the second phase, we exploit the clusters in $\rset'$ to build the tree-of-sets system.
 
 Given a graph $\G$ and a subset $\ttset$ of vertices called terminals, we say that a pair $(t,t')$ of terminals is $\lambda$-edge-connected, if and only if there are at least $\lambda$ paths connecting $t$ to $t'$ in $\G$, that are mutually edge-disjoint. Let $\lambda(t,t')$ be the largest value $\lambda$, such that $t$ and $t'$ are $\lambda$-edge-connected, and let $\lambda_{\G}(\ttset)=\min_{t,t'\in \ttset}{\lambda(t,t')}$. We say that a pair $t,t'$ of terminals is $\mu$-element-connected, if and only if there are $\mu$ paths connecting $t$ to $t'$ that are pairwise disjoint in both the edges and the non-terminal vertices of $\G$ (but they are allowed to share terminals). Let $\mu(t,t')$ be the largest value $\mu$, such that $t$ and $t'$ are $\mu$-element-connected, and denote $\mu_{\G}(\ttset)=\min_{t,t'\in \ttset}{\mu(t,t')}$. Clearly, $\lambda_{\G}(\ttset)\geq \mu_{\G}(\ttset)$ always holds. 
 We use the following theorem several times.
 
 \begin{theorem}\label{thm: splitting off}
 There is an efficient algorithm, that, given a graph  $\G$ and a set $\ttset\subseteq V(\G)$  of $\kappa$ vertices called terminals, such that $\mu_{\G}(\ttset)\geq \mu$ for some $\mu\geq 1$, constructs another graph $\H$  with $V(\H)=\ttset$, a partition $\uset$ of $E(\H)$ into groups of size at most $\kappa$, and for each edge $e=(t,t')\in E(\H)$ a path $P_e$ connecting $t$ to $t'$ in $\G$, such that:
 
 \begin{itemize}
 \item $\lambda_{\H}(\ttset)\geq 2\mu$;
 \item for each terminal $t$, $d_{\H}(t)\leq 2d_{\G}(t)$;
 
 \item for each $e\in E(\H)$, path $P_e$ does not contain terminals as inner vertices; 
 
  \item if $e'\in E(\G)$ is incident to a terminal of $\ttset$, then $e'$ belongs to at most $2$ paths in $\set{P_e\mid e\in E(\H)}$; and
 \item if we select, for each group $U\in \uset$, an arbitrary edge $e_U\in U$, then the corresponding paths $\set{P_{e_U}\mid U\in \uset}$ are  node-disjoint in $\G$, except for possibly sharing endpoints.
 \end{itemize}
 \end{theorem}
 
 \begin{proof}
 We use the following theorem of Hind and Oellermann~\cite{element-connectivity} (see also~\cite{ChekuriK09}).
 
 \begin{theorem}\label{thm: splitting for element connectivity}
 Let $\G$ be a graph, $\ttset$ a set of terminals in $\G$, and assume that $\mu_{\G}(\ttset)=\mu$ for some $\mu\geq 0$. Let $(p,q)$ be an edge with $p,q\in V\setminus \ttset$. Let $G_1$ be the graph obtained from $G$ by deleting the edge $(p,q)$, and let $G_2$ be obtained from $G$ by contracting it. Then either $\mu_{G_1}(\ttset)\geq \mu$ or $\mu_{G_2}(\ttset)\geq \mu$.
 \end{theorem}
 
 While our graph $\G$ contains an edge $(p,q)$ connecting two non-terminal vertices $p$ and $q$, we apply Theorem~\ref{thm: splitting for element connectivity} to $\G,\ttset$ and the edge $(p,q)$, and replace $\G$ with the resulting graph, where the edge $(p,q)$ is either deleted or contracted. Let $G'$ be the graph obtained at the end of this procedure. For simplicity, we call the terminal vertices of $G'$ \emph{black vertices}, and the non-terminal vertices \emph{white vertices}. Let $W$ denote the set of all white vertices. Notice that every edge in $G'$ either connects two black vertices, or it connects a white vertex to a black vertex. Moreover, we can assume without loss of generality that for each $t\in \ttset$, $v\in W$, there is at most one edge $(t,v)$ in $G'$: otherwise, if several such parallel edges are present, we delete all but one such edge. This does not affect the element-connectivity of any pair $t,t'$ of terminals, since the paths connecting them are not allowed to share $v$. So we will assume from now on that every such pair $(t,v)$ at most one edge $(t,v)$ is present in $G'$. Notice that for each terminal $t$, $d_{G'}(t)\leq d_{\G}(t)$. 
For every pair $(t,t')$ of terminals, an edge $(t,t')$ is present in $G'$ if and only if it was present in $\G$. Every white vertex $v$ is naturally associated with a connected subgraph $C_v$ of $\G$, containing all edges that were contracted into $v$, and all subgraphs $\set{C_v}_{v\in W}$ are completely disjoint. For each edge $(v,t)$ connecting $v$ to some terminal $t$ in $G'$, there is an edge $(u,t)$ in $\G$, where $u$ is some vertex in $C_v$.  
 Notice that $\lambda_{G'}(\ttset)\geq \mu_{G'}(\ttset)\geq \mu$.

 Next, we replace every edge in $G'$ by two parallel edges, and denote the resulting graph by $G''$, so $G''$ is Eulerian. The degree of every terminal $t$ now becomes at most $2d_{\G}(t)$, and $\lambda_{G''}(\ttset)\geq 2\mu$. We now start constructing the final graph $\H$ and the partition $\uset$ of its edges. We start with $\H=G''$, and for every edge $(t,t')\in E(G')$ connecting a pair $t,t'\in \ttset$ of terminals, we add a new group $U$, containing the two copies of the edge $(t,t')$ in $G''$, to $\uset$. Next, we take care of the white vertices, by using the following edge-splitting operation due to Mader~\cite{edge-connectivity}.
 
 \begin{theorem}\label{thm: edge splitting for edge-connectivity}
 Let $\G$ be an undirected multi-graph, $s$ a vertex of $\G$ whose degree is not $3$, such that $s$ is not incident to a cut edge of $G$. Then $s$ has two neighbors $u$ and $v$, such that the graph $\G'$, obtained from $\G$ by replacing the edges $(s,u)$ and $(s,v)$ with the edge $(u,v)$, satisfies $\lambda_{\G'}(x,y)=\lambda_{\G}(x,y)$ for all $x,y\neq s$.
 \end{theorem}
 
 We process the white vertices $v\in W$ one-by-one. Consider some such
 vertex $v$. Recall that there are at most $2\kappa$ edges incident to
 $v$ in $G''$. We apply Theorem~\ref{thm: edge splitting for
   edge-connectivity} to vertex $v$ repeatedly, until it becomes an
 isolated vertex (since the degree of $v$ is even due to the doubling
 of all edges, and the terminals are $2\mu$-edge-connected, the
 conditions of the theorem are always satisfied). Let $U_v$ be the set
 of all resulting new edges in graph $\H$. We add $U_v$ to
 $\uset$. Notice that $|U_v|\leq \kappa$. Once all vertices $v\in W$
 are processed, we obtain the final graph $\H$.  It is easy to see
 that the degree of every terminal $t\in \ttset$ is at most
 $2d_{\G}(t)$, since the edge-splitting operation does not change the
 degrees of the terminals. Every edge $e=(t,t')$ in $\H$ is naturally
 associated with a path $P_e$ connecting $t$ and $t'$ in $\G$: if edge
 $e=(t,t')$ is present in $\G$, then $P_e=e$. Otherwise, edge $e$ was
 obtained by replacing a pair $(t,v),(v,t')$ of vertices in $G''$ with
 edge $(t,t')$. In this case, there must be vertices $u,u'\in C_v$
 (possibly $u=u'$), with edges $e_1=(u,t),e_2=(u',t')\in E(\G)$. We
 let $P$ be a path connecting $u$ to $u'$ in $C_v$, and set
 $P_e=(e_1,P,e_2)$.  Given two edges $e,e'\in E(\H)$, the only
 possibility that the paths $P_e$ and $P_{e'}$ share inner vertices or
 edges is when $e,e'$ are two copies of the same edge connecting some
 pair of terminals in $G''$, or both edges belong to some set $U_v$,
 for some $v\in W$. Therefore, choosing at most one edge from each
 group $U\in \uset$ ensures that the resulting paths are internally
 node- and edge-disjoint.
 
 Consider now some edge $e'\in E(\G)$, such that $e'$ is incident to some terminal $t$. Notice that exactly two copies of $e'$ were present in $\H$ before the edge-splitting procedure. It is then immediate to verify that $e'$ may belong to at most two paths in set $\set{\pset_e\mid e\in E(\H)}$.
 \end{proof}

 \subsubsection{Phase 1}
 
 Let $H$ be the graph obtained from $G$ by contracting each cluster
 $S_i\in \rset$ into a super-node $v_i$. Let
 $\ttset=\set{v_1,\ldots,v_{\ell_0}}$ be the resulting set of super-nodes
 that we will refer to as terminals in this phase. As observed above,
 every pair $S_i,S_j\in \rset'$ of clusters has at least
 $\mu=\floor{\frac{w_0}{6\Delta}}$ node-disjoint paths connecting them
 in $G$. Therefore, $\mu_{H}(\ttset)\geq \mu$. We apply
 Theorem~\ref{thm: splitting off} to graph $H$, set $\ttset$ of
 terminals and the value $\mu$. Let $H'$ denote the resulting graph,
 and $\uset$ the resulting partition of the edges of $H'$, where each
 group $U\in\uset$ contains at most $\ell_0$ edges. Recall that each edge
 $e=(v_i,v_j)$ in $H'$ corresponds to a path $P_e$ connecting $v_i$ to
 $v_j$ in $H$, where $P_e$ does not contain vertices of $\ttset$,
 except for its endpoints. In turn, path $P_e$ defines a path $P'_e$
 in graph $G$, connecting a vertex of $S_j$ to a vertex of $S_i$ {\bf
   directly}, that is, $P'_e$ does not contain the vertices of
 $\bigcup_{S\in \rset}S$ as inner vertices.
 
 Let $w_1=\lfloor \frac{w_0}{2\Delta \ell_0^2}\rfloor$. We define a new
 graph $Z$, whose vertex set is $\set{v_1,\ldots,v_{\ell_0}}$. We add an
 edge $(v_i,v_j)$ to $Z$ if and only if there are at least $w_1$ parallel edges
 connecting $v_i$ to $v_j$ in $H'$.  It is easy to verify that graph
 $Z$ is connected: indeed, assume otherwise. Let $A$ be some connected
 component of $Z$, and let $B$ contain the rest of the vertices. Let
 $v_j\in A$, $v_{j'}\in B$. Since there are at least $2\mu\geq \lfloor
 \frac{w_0}{6\Delta}\rfloor $ edge-disjoint paths connecting $v_j$ and
 $v_{j'}$ in $H'$, $|E_{H'}(A,B)|\geq \lfloor
 \frac{w_0}{6\Delta}\rfloor $ must hold. Since $|A|+|B|=\ell_0$, at least
 one pair $(v_j,v_{j'})$ with $v_j\in A$, $v_{j'}\in B$, has at least
 $w_1$ parallel edges connecting them in $H'$.

Let $T$ be a spanning tree of $Z$ that is rooted at some arbitrary
node. We say that Case 1 happens if $T$ contains a root-to-leaf path
of length at least $\ell$, and we say that Case 2 happens
otherwise. Since $|V(T)|=\ell_0=\ell^2$, and every vertex of $T$ lies on
some root-to-leaf path, if Case 2 happens, then $T$ contains at least
$\ell$ leaves (see Claim~\ref{claim:path-leaves-in-tree}). We now consider each of the two cases separately. For Case
1, we build a good tree-of-sets system directly, and we only apply
Phase 2 of the algorithm if Case 2 happens.
 
\paragraph{Case 1} Let $P$ be a path of $T$ of length exactly
$\ell$. Assume without loss of generality that
$P=\set{v_1,\ldots,v_{\ell}}$. Let $\rset'=\set{S_1,\ldots,S_{\ell}}$ be the set
of corresponding clusters. We build a tree-of-sets system
$(\rset',T^*,\bigcup_{e\in T^*}\pset^*_e)$. The tree $T^*$ is just a
path connecting $v_1,\ldots,v_{\ell}$ in this order.  In order to define
the sets $\pset^*_e$ of edges, we perform the following
procedure. Recall that each edge $(v_i,v_{i+1})$ in $P$ corresponds to
a collection $E_i$ of at least $w_1$ edges in graph $H'$. For each
group $U\in \uset$, we select one edge $e_U\in U$ uniformly at
random. Let $E^*=\set{e_U\mid U\in \uset}$ be the set of the selected
edges. For each $1\leq i<\ell$, let $E'_i=E_i\cap E^*$. Since the size of
each group $U$ is at most $\ell_0$, the expected number of edges in
$E'_i$ is at least $\frac{w_1}{\ell_0}>w\geq \log k$. Using the standard
Chernoff bound, with high probability, for each $1\leq i<\ell$,
$|E'_i|\geq \floor{\frac{w_1}{2\ell_0}}$. Let
$w_2=\floor{\frac{w_1}{2\ell_0}}$. If $E'_i$ contains more than $w_2$
edges, we discard arbitrary edges from $E'_i$ until $|E'_i|=w_2$
holds. Let $\tilde \pset_i=\set{P_e\mid e\in E'_i}$, and let
$\tilde \pset=\bigcup_{i=1}^{\ell-1}\tilde \pset_i$. Then from Theorem~\ref{thm:
  splitting off}, the paths in $\tilde \pset$ are internally node-disjoint in
$H$, and they do not contain terminals as inner vertices. Therefore,
the corresponding paths in graph $G$ are also internally
node-disjoint, and they do not contain the vertices of
$\bigcup_{i=1}^{\ell}S_i$ as inner vertices.
For each $1\leq i\leq \ell$, let $\pset_i$ be the set of paths in graph $G$, corresponding to the paths in $\tpset_i$, and let $\pset=\bigcup_{i=1}^{\ell-1}\pset_i$. Then the paths in $\pset$ are direct and internally node-disjoint.
 Only two issues
remain. First, it is possible that some paths in $\pset$ share
endpoints, and second, we need to ensure that the clusters $S_i$ have
the bandwidth property in the corresponding graph.
 
In order to solve the first problem, we will define, for each $1\leq
i< \ell$, a subset $\pset_i'\subseteq\pset_i$ of $w$ paths, such that every
vertex of $\bigcup_{i=1}^{\ell}S_i$ belongs to at most one path in
$\bigcup_{i'=1}^{\ell-1}\pset'_{i'}$ (or in other words, the paths in
$\bigcup_{i'=1}^{\ell-1}\pset'_{i'}$ have distinct endpoints).

We start with the set $\pset_1$. Using Claim~\ref{claim: large matching}, we can compute a subset
$\pset_1'\subseteq\pset_1$ of $\floor{\frac{w_2}{2\Delta}}$ paths that
do not share endpoints in $S_1$ or $S_2$.  We then consider the set
$\pset_2$ of paths. We delete from $\pset_2$ all paths that share
endpoints with paths in $\pset_1'$. Since
$|\pset_1'|=\floor{\frac{w_2}{2\Delta}}$, at most
$\frac{w_2}{2}$ paths of $\pset_2$ may share endpoints with paths in
$\pset_1'$, so at least half the paths in $\pset_2$ remain. As before,
we select a subset $\pset_2'$ of $\floor{\frac{w_2}{2\Delta}}$ such
paths that do not endpoints using Claim~\ref{claim: large matching}. We continue this process for all
$1\leq i<\ell$, until all paths in set $\pset'=\bigcup_i\pset'_i$ are
mutually disjoint. The sets $\pset'_i$ of paths are then used to
define the sets $\pset^*_e$ of paths in the tree-of-sets system.
Notice that the size of each set is
$\floor{\frac{w_2}{2\Delta^2}}\geq\frac{w_1}{8\ell_0\Delta^2}\geq
\frac{w_0}{32\Delta^3\ell_0^3}>w$ from Equation~(\ref{eq: bound on h}).

Let $G^*$ be the subgraph of $G$ obtained by the union of $G[S]$ for
$S\in \rset'$ and $\bigcup_{e\in T^*}\pset^*_e$.  We need to verify
that each set $S_i$ has the $\alphaWL$-bandwidth property in
$G^*$. Let $\Gamma_i$ be the interface of the set $S_i$ in $G^*$. We
set up an instance of the sparsest cut problem with the graph $G[S_i]$ and the
set $\Gamma_i$ of terminals, and apply algorithm \algSC to it. If the
outcome is a cut of sparsity less than $\alpha$, then, since
$|\Gamma_i|<w_0$, we obtain an $(w_0,\alpha)$-violating cut of $S_i$
in graph $G$. We return this cut as the outcome of the algorithm. If
\algSC returns a cut of sparsity at least $\alpha$ for each set $S_i$,
for $1\leq i\leq \ell$, then we are guaranteed that each such set has the
$\alphaWL$-bandwidth property in $G^*$, and we have therefore
constructed a good tree-of-sets system. (We are guaranteed that
$S_i\cap \tset=\emptyset$ for each $i$, since each set $S_i\in \rset$
only contains non-terminal vertices).

\paragraph{Case 2}
If Case 2 happens, then we  need to execute the second phase of
the algorithm, but first we need to establish the following useful
fact.  Recall that we have found a tree $T$ in graph $Z$ containing at
least $\ell$ leaves. Let $\rset'\sse \rset$ be an arbitrary subset of $\ell$
clusters, corresponding to the leaves of $T$. For simplicity of
notation, we assume that $\rset'=\set{S_1,\ldots,S_{\ell}}$.  We say that a
path $P$ connects $S_i$ to $S_j$ \emph{directly}, for $S_i,S_j\in
\rset'$ if and only if no inner vertex of $P$ belongs to $\bigcup_{S\in \rset'}S$
(but they may belong to clusters $S\in \rset\setminus\rset'$).

\begin{theorem}\label{thm: case 2 step 1}
  There is an efficient randomized algorithm, that with high probability either
  computes a $(w_0,\alpha)$-violating cut of some set $S\in \rset$,
  or finds, for each set $S_i\in \rset'$, a subset $E_i\subseteq
  \out_G(S_i)$ of edges, and for each $1\leq i\neq
  j\leq \ell$  a collection $\pset_{i,j}$ of paths in $G$ that satisfy the following
  properties:
  \begin{itemize}
  \item $|E_i| = w_3=\Omega\left (\frac{\alpha^2 w_1}{\ell\Delta^7\log k}\right )$
  \item $\pset_{i,j}$ is a collection of $w_3$ node-disjoint paths that 
    directly connect $S_i$ to $S_j$
  \item each path $P\in \pset_{i,j}$ contains an edge of $E_i$ and an
    edge of $E_j$ as its first and last edges, respectively.
  \end{itemize}
\end{theorem}

Notice that the theorem implies that the edges in each set $E_i$ do not share endpoints.

\begin{proof}
  Recall that each edge $e\in E(H')$ corresponds to a path $P_e$ in
  $H$. Since for any choice of edges $e_U\in U$ for each $U\in \uset$,
  we are guaranteed that the corresponding set $\set{P_{e_U}\mid U\in
    \uset}$ of paths is edge-disjoint in $H$, and $|U|\leq \ell_0$ for
  all $U\in \uset$, the total edge congestion in $H$ due to the paths
  in $\set{P_e\mid e\in E(H')}$ is at most $\ell_0$. Each such path $P_e$
  naturally defines a direct path $P'_e$ in $G$, and the total edge
  congestion due to the paths in $\set{P'_e\mid e\in E(H')}$ is at
  most $\ell_0$ in $G$.

  Let $v^*$ be the root of the tree $T$, and let $S^*\in \rset$ be the
  corresponding vertex subset. We use the following claim.
  
  \begin{claim} For each set $S_j\in \rset'$, there is a flow $F_j$ of value $\ceil{\frac{w_1}{2\Delta}}$, originating from the
  vertices of $S^*$ and terminating at the vertices of $S_j$ in $G$, with
  edge-congestion at most  $\frac{2}{\alpha}$, such that the flow-paths in $F_j$ do not contain
 the vertices of $\bigcup_{S\in \rset'}S$ as inner vertices.
  \end{claim}
  
  \begin{proof}
   Consider the path $(v^*=v_{i_1},v_{i_2},\ldots,v_{i_x}=v_j)$ in the
  tree $T$, connecting $v^*$ to $v_j$.  For each edge
  $e_z=(v_{i_z},v_{i_{z+1}})$ on this path, there are $w_1$ parallel
  edges corresponding to $e_z$ in graph $H'$. Let $\qset_z$ be the
  corresponding set of $w_1$ direct paths connecting the vertices of
  $S_{i_z}$ to the vertices of $S_{i_{z+1}}$ in $G$.  
  Recall that from Theorem~\ref{thm: splitting off}, every edge incident to a vertex of $S_{i_z}\cup S_{i_{z+1}}$ belongs to at most two paths in set $\qset_z$. Therefore, every vertex of $S_{i_z}\cup S_{i_{z+1}}$ serves as an endpoint of at most $2\Delta$ paths. We can construct a bipartite graph, whose vertex set is $(S_{i_z}\cup S_{i_{z+1}})$, and there is an edge $(u,v)$ for $u\in S_{i_z}$, $v\in S_{i_{z+1}}$ for every path in $\qset_z$ connecting $u$ to $v$. Then the resulting bipartite graph has maximum vertex degree at most $2\Delta$, and so from Claim~\ref{claim: large matching}, we can find a matching of cardinality $\ceil{\frac{|\qset_z|}{2\Delta}}$ in this graph. This matching defines a subset $\qset_z'\subseteq \qset_z$ of $\ceil{\frac{w_1}{2\Delta}}$ paths, whose endpoints are all distinct.
  Denote

\[\Gamma^2_z=\set{v\in S_{i_z}\mid v\mbox{ is the first vertex on some path in }\qset'_z},\]

 and similarly 
 
 \[\Gamma^1_{z+1}=\set{v\in S_{i_{z+1}}\mid v\mbox{ is the last vertex on some path in }\qset'_z}.\]
 
 For each $1<z<x$, we have now defined two subsets
 $\Gamma^1_z,\Gamma^2_z\sse S_{i_z}$ of vertices of cardinality $\ceil{\frac{w_1}{2\Delta}}<w_0$
 each. Let $n_z=|\Gamma^1_z\setminus \Gamma^2_z|=|\Gamma^2_z\setminus \Gamma^1_z|$. We now try to send flow in $G[S_{i_z}]$, where every vertex in $\Gamma^1_z\setminus \Gamma^2_z$ sends one flow unit, and every vertex in $\Gamma^2_z\setminus \Gamma^1_z$ receives one flow unit,
 with edge-congestion at most $1/\alpha$. If such a flow does not
 exist, then the minimum edge-cut separating these two vertex subsets
 defines a $(w_0,\alpha)$-violating cut of $S_{i_z}$. We then
 terminate the algorithm and return this cut. We assume therefore that
 such a flow exists, and denote it by $F'_z$.
  
 Concatenating the flows
 $(\qset'_1,F'_1,\qset'_2,F'_2,\ldots,F'_{x-1},\qset'_{x-1})$, we obtain
 the desired flow $F_j$ of value $\ceil{\frac{w_1}{2\Delta}}$. The total congestion caused by
 paths in $\bigcup_{z=1}^{x-1}\qset_z$ is at most $\ell_0$, while each
 flow $F'_z$ causes congestion at most $1/\alpha$ in graph
 $G[S_{i_z}]$. Therefore, the total congestion due to flow $F_j$ is
 bounded by $\frac{1}{\alpha}+\ell_0\leq \frac{2}{\alpha}$.
\end{proof}

 Scaling all flows $F_j$, for $S_j\in \rset'$ down by factor
 $2\ell\Delta/\alpha$, we obtain a new flow $F$, where every set $S_j\in
 \rset'$ sends at least $\floor{\frac{w_1\alpha}{4\ell\Delta^2}}$ flow units to $S^*$, and
 the total vertex-congestion due to $F$ is at most $1$.  The flow-paths of $F$ do not contain the vertices of $\bigcup_{j=1}^{\ell}S_j$ as inner vertices. We now use the following claim.
 
 \begin{claim}
 There is an efficient algorithm to compute a collection $\set{\pset_j}_{j=1}^{\ell}$ of
 path sets, where for each $1\leq j\leq \ell$, set $\pset_j$ contains
 $\floor {\frac{w_1\alpha}{4\ell\Delta^2}}$ paths connecting $S_j$ to $S^*$,
 the paths in $\bigcup_{j=1}^{\ell}\pset_j$ are node-disjoint, and they do
 not contain the vertices of $\bigcup_{j=1}^{\ell}S_j$ as inner
 vertices.
 \end{claim}
 \begin{proof}
 We set up a flow network $\nset$: start with the graph $G$ and bi-direct all its edges, setting the capacity of every vertex to $1$. Delete all edges entering the clusters $S_j$, and all edges with both endpoints in $S_j$, for all $1\leq j\leq \ell$. Delete all edges leaving the cluster $S^*$, and every edge with both endpoints in $S^*$.
 For each $1\leq j\leq \ell$, add a vertex $s_j$ of capacity $\floor{\frac{w_1\alpha}{4\ell\Delta^2}}$, and connect $s_j$ to every vertex of $S_j$ with a directed edge.
 Finally, add a source vertex $s$ of infinite capacity, connecting it to every vertex in $\set{s_1,\ldots,s_{\ell}}$, and add a destination vertex $t$ of infinite capacity, to which every vertex of $S^*$ connects. Notice that flow $F$ defines a feasible $s$-$t$ flow of value $\ell\cdot \floor{\frac{w_1\alpha}{4\ell\Delta^2}}$ in this network. From the integrality of flow, there is an integral flow of this value in $\nset$. It is immediate to verify that this integral flow defines the desired collections $\pset_j$ of paths, for $1\leq j\leq \ell$.
 \end{proof}

 For each $1\leq j\leq \ell$, let $A_j\sse S^*$ be the set of vertices
 that serve as endpoints of paths in $\pset_j$, and let
 $A=\bigcup_{j=1}^{\ell}A_j$. Notice that $|A|=\ell\cdot \ceil{
   \frac{w_1\alpha}{4\ell\Delta^2}}<w_0$. We set up an instance of the
 sparsest cut problem in graph $G[S^*]$, where the vertices of $A$ act
 as terminals, and apply algorithm \algSC to this problem. If the
 algorithm returns a cut whose sparsity is less than $\alpha$, then we
 have found a $(w_0,\alpha)$-violating cut in $S^*\in \rset$. We
 return this cut, and terminate the algorithm. Otherwise, we are
 guaranteed that the set $A$ is $\alphaBW$-well-linked in $G[S^*]$.

 We apply Corollary~\ref{cor: ultimate boosting} to graph $G[S^*]$ and the sets $A_1,\ldots,A_{\ell}$ of vertices, obtaining, for each $1\leq j\leq \ell$, a subset $A^*_j\subseteq A_j$ of $\Omega\left (\frac{\alphaWL |A_j|}{\Delta^5\sqrt{\log k}}\right )=\Omega\left (\frac{\alpha^2 w_1}{\ell \Delta^7 \log k}\right )=w_3$ vertices, such that for all $1\leq i\neq j\leq \ell$, $A^*_j$ and $A^*_i$ are linked in $G[S^*]$. Let $\qset_{i,j}$ be the
 set of $w_3$ node-disjoint paths connecting $A^*_j$ to $A^*_i$ in
 $G[S^*]$.

 For each $1\leq j\leq \ell$, let $\pset'_j\subseteq\pset_j$ be the subset
 of paths whose endpoint belongs to $A^*_j$, and let $E_j\subseteq
 \out(S_j)$ be the set of edges $e$, where $e$ is the first edge on
 some path of $\pset'_j$, so $|E_j|=w_3$. Consider a pair $1\leq
 i<j\leq \ell$ of indices. The desired set of paths connecting $S_i$ to
 $S_j$ is obtained by concatenating the paths in
 $\pset'_i,\qset_{i,j}$ and $\pset'_j$.\end{proof}

To summarize, we have found a collection $\rset'=\set{S_1,\ldots,S_{\ell}}$
of $\ell$ disjoint vertex subsets, and for each set $S_j$, a collection
$E_j\subseteq \out_G(S_j)$ of $w_3$ edges, such that for each $1\leq
j\neq i\leq \ell$, there is a set of $w_3$ node-disjoint paths in $G$, connecting $S_i$ to $S_j$ directly,
such that each path contains an edge of $E_j$ and an edge of $E_i$ as
its first and last edges, respectively.

\subsubsection{Phase 2}

We construct a new graph $\tH$, obtained from $G$ as follows. First,
for each $1\leq j\leq \ell$, we delete all edges in $\out(S_j)\setminus
E_j$ from $G$. Let $B_j\subseteq S_j$ be the subset of vertices containing the
endpoints of the edges in $E_j$ that belong to $S_j$. We delete all
vertices of $S_j\setminus B_j$ and add a new super-node $v_j$ that
connects to every vertex in $B_j$ with an edge. Recall that from the
above discussion, $|B_j|=w_3$, so the degree of every super-node $v_j$
in $\tH$ is $w_3$, and every pair $v_j,v_i$ of super-nodes are
connected by $w_3$ paths, that are completely disjoint, except for
sharing the first and the last vertex. We will think of the super-node
$v_j$ as representing the set $S_j\in \rset'$. In this phase, the
vertices of $\set{v_1,\ldots,v_{\ell}}$ are called terminals, and a path
$P$ connecting a vertex of $S_i$ to a vertex of $S_{i'}$ in $G$ is
called direct if and only if it does not contain the vertices of
$\bigcup_{j=1}^{\ell}S_j$ as inner vertices.

We apply Theorem~\ref{thm: splitting off} to graph $\tH$, with the set
$\set{v_1,\ldots,v_{\ell}}$ of terminals and $\mu=w_3$. Let $\tH'$ denote
the resulting graph, and $\uset$ the resulting partition of the edges
of $\tH$ into groups of size at most $\ell$. Recall that the degree of
every vertex of $\tH'$ is at most $2w_3$, and every pair of vertices
is $(2w_3)$-edge-connected. This can only happen if the degree of
every vertex is exactly $2w_3$. The main difference between graph
$\tH'$ and the graph $H'$ that we computed in Phase 1 is that now the
degree of every terminal, and the edge-connectivity of every pair of terminals, are the same. It
is this property that allows us to build the tree-of-sets system.  To
simplify notation, denote $h=2w_3$.

Suppose we choose, for each group $U\in \uset$, some edge $e_U\in
U$. Then Theorem~\ref{thm: splitting off} guarantees that all paths in
$\set{P_{e_U}\mid U\in \uset}$ are node-disjoint in $\tH$, except for
possibly sharing endpoints, and they do not contain terminals as inner
vertices. For each such edge $e_U=(v_i,v_j)\in E(\tH')$, path
$P_{e_U}$ in $\tH$ naturally defines a direct path $P'_{e_U}$,
connecting a vertex of $S_i$ to a vertex of $S_j$ in $G$. Moreover,
from the definition of graph $\tH$, the paths in $\set{P'_{e_U}\mid
  U\in \uset}$ are completely node-disjoint in $G$. For each edge $e\in E(\tilde H')$, let $P'_e$ denote the path in graph $G$ corresponding to the path $P_e$ in $\tilde H$.

We build an auxiliary undirected graph $\tZ$ on the set
$\set{v_1,\ldots,v_{\ell}}$ of vertices, as follows. For each pair
$v_j,v_{j'}$ of vertices, there is an edge $(v_j,v_{j'})$ in graph
$\tZ$ if and only if there are at least $h/\ell^3$ edges connecting $v_j$ and
$v_{j'}$ in $\tH'$. If edge $e=(v_j,v_{j'})$ is present in graph
$\tZ$, then its capacity $c(e)$ is set to be the number of edges
connecting $v_j$ to $v_{j'}$ in $\tH'$. For each vertex $v_j$, let
$C(v_j)$ denote the total capacity of edges incident on $v_j$ in graph
$\tZ$.  We need the following simple observation.

\begin{observation}\label{observation: graph Z}
\ 
\begin{itemize}

\item For each vertex $v\in V(\tZ)$, $(1-1/\ell^2)h \leq C(v)\leq h$. 
\item For each pair $(u,v)$ of vertices in graph $\tZ$, we can send at
  least $(1-1/\ell)h$ flow units from $u$ to $v$ in $\tZ$ without
  violating the edge capacities.
\end{itemize}
\end{observation}
\begin{proof}
  In order to prove the first assertion, recall that each vertex in
  graph $\tH'$ has $h$ edges incident to it. So $C(v)\leq h$ for
  all $v\in V(\tZ)$.  Call a pair $(v_j,v_{j'})$ of vertices bad iff
  there are fewer than $h/\ell^3$ edges connecting $v_j$ to $v_{j'}$
  in $\tH'$. Notice that each vertex $v\in V(\tZ)$ may participate in
  at most $\ell$ bad pairs, as $|V(\tZ)|=\ell$. Therefore, $C(v)\geq
  h-\ell h/\ell^3=h(1-1/\ell^2)$ must hold.

  For the second assertion, assume for contradiction that it is not
  true, and let $(u,v)$ be a violating pair of vertices. Then there is
  a cut $(A,B)$ in $\tZ$, with $u\in A$, $v\in B$, and the total
  capacity of edges crossing this cut is at most $(1-1/\ell)h$. Since
  $u$ and $v$ were connected by $h$ edge-disjoint paths in graph
  $\tH'$, this means that there are at least $h/\ell$ edges in graph
  $\tH'$ that connect bad pairs of vertices. But since we can only
  have at most $\ell^2$ bad pairs, and each pair has fewer than
  $h/\ell^3$ edges connecting them, this is impossible.
\end{proof}

The following claim allows us to find a spanning tree of $\tZ$ with
maximum vertex degree at most $3$. It relies on the specific
properties of the graph $\tZ$ outlined in the preceding observation.
This low-degree spanning tree will be used to define the tree-of-sets
system.

\begin{claim}\label{claim: small-degree spanning tree}
  There is an efficient algorithm to find a spanning tree $T^*$ of
  $\tZ$ with maximum vertex degree at most $3$.
\end{claim}
\begin{proof}
  We use the algorithm of Singh and Lau~\cite{Singh-Lau} for
  constructing bounded-degree spanning trees. Suppose we are given a
  graph $G=(V,E)$, and our goal is to construct a spanning tree $T$ of
  $G$, where the degree of every vertex is bounded by some integer $B$. For each
  subset $S\sse V$ of vertices, let $E(S)$ denote the subset of edges
  with both endpoints in $S$, and $\delta(S)$ the subset of edges with
  exactly one endpoint in $S$. Singh and Lau consider a natural
  LP-relaxation for the problem. We note that their algorithm works
  for a more general problem where edges are associated with costs,
  and the goal is to find a minimum-cost tree that respects the degree
  requirements; since we do not need to minimize the tree cost, we
  only discuss the unweighted version here. For each edge $e\in E$, we
  have a variable $x_e$ indicating whether $e$ is included in the
  solution. We are looking for a feasible solution to the following
  LP.

\begin{eqnarray}
&\sum_{e\in E}x_e=|V|-1&\label{LP: total edge weights}\\
&\sum_{e\in E(S)}x_e\leq |S|-1&\forall S\subsetneq V \label{LP: sum for subsets}\\
&\sum_{e\in \delta(v)}x_e\leq B&\forall v\in V \label{LP: degree constraints}\\
&x_e\geq 0&\forall e\in E
\end{eqnarray}

Singh and Lau~\cite{Singh-Lau} show an efficient algorithm, that,
given a feasible solution to the above LP, produces a spanning tree
$T$, where for each vertex $v\in V$, the degree of $v$ is at most
$B+1$ in $T$. Therefore, in order to prove the claim, it suffices to
show a feasible solution to the LP, where $B=2$. Recall that
$|V(\tZ)|=\ell$. The solution is defined as follows. Let $e=(u,v)$ be any
edge in $E(\tZ)$. We set the LP-value of $e$ to be
$x_e=\frac{\ell-1}{\ell}\cdot \left
  (\frac{c(e)}{C(v)}+\frac{c(e)}{C(u)}\right )$.  We say that
$\frac{\ell-1}{\ell}\cdot \frac{c(e)}{C(v)}$ is the contribution of $v$ to
$x_e$, and $\frac{\ell-1}{\ell}\cdot \frac{c(e)}{C(u)}$ is the contribution
of $u$. We now verify that all constraints of the LP hold.

First, it is easy to see that $\sum_{e\in E}x_e=\ell-1$, as
required. Next, consider some subset $S\subsetneq V$ of vertices. Notice
that it suffces to establish Constraint~(\ref{LP: sum for subsets})
for subsets $S$ with $|S|\geq 2$. From Observation~\ref{observation:
  graph Z}, the total capacity of edges in $E_{\tZ}(S,\nots)$ must be
at least $(1-1/\ell)h$. Since for each $v\in S$, $C(v)\leq h$, the
total contribution of the vertices in $S$ towards the LP-weights of
edges in $E_{\tZ}(S,\nots)$ is at least $\frac{\ell-1}{\ell}\cdot
(1-1/\ell)=(1-1/\ell)^2$. Therefore,

\[\sum_{e\in E(S)}x_e\leq \frac{\ell-1}{\ell}|S|-(1-1/\ell)^2=|S|-|S|/\ell-1-1/\ell^2+2/\ell\leq |S|-1\]

since we assume that $|S|\geq 2$. This establishes
Constraint~(\ref{LP: sum for subsets}).  Finally, we show that for
each $v\in V(\tZ)$, $\sum_{e\in \delta_v}x_e\leq 2$. First, the
contribution of the vertex $v$ to this summation is bounded by
$1$. Next, recall that for each $u\in V(\tZ)$, $C(u)\geq
(1-1/\ell^2)h$, while the total capacity of edges in $\delta(v)$ is at
most $h$. Therefore, the total contribution of other vertices to
this summation is bounded by $\frac{h}{(1-1/\ell^2)h}\cdot
\frac{\ell-1}{\ell}\leq \frac{\ell}{\ell+1}\leq 1$. The algorithm of Singh and Lau
can now be used to obtain a spanning tree $T^*$ for $\tZ$ with maximum
vertex degree at most $3$.
\end{proof}

We are now ready to define the tree-of-sets system
$(\rset',T^*,\bigcup_{e\in E(T^*)}\pset^*(e))$. The tree $T^*$ is the
tree computed by Claim~\ref{claim: small-degree spanning tree}. In
order to define the sets $\pset^*(e)$ of paths, recall that each edge
$e$ of $\tilde Z$ (and hence of $T^*$) corresponds to a set $S_e$ of
at least $2w_3/\ell^3$ edges of $\tH'$. For each group $U\in \uset$, we
randomly choose one edge $e_U\in \uset$, and we let $E^*\subseteq
E(\tH)$ be the set of all selected edges. For each edge $e\in E(T^*)$, let
$S'_e=S_e\cap E^*$. The expected size of $S'_e$ is at least
$\frac{2w_3}{\ell^4}$, and using the standard Chernoff bound, with high
probability, for each edge $e\in E(T^*)$, $|S'_e|\geq
\frac{w_3}{\ell^4}$, since $w_3/\ell^4\geq w\geq 4\log k$. This is since
$\frac{w_3}{\ell^4}=\Omega\left
  (\frac{w_1\alpha^2}{\ell^5\Delta^7\log k}\right
)=\Omega\left(\frac{w_0\alpha^2}{\Delta^8\ell^9\log k}\right )\geq w$
from Equation~(\ref{eq: bound on h}).  The final set $\pset^*_e$ of
paths is $\set{P'_{e'}\mid e'\in S'_e}$. Notice that $|\pset^*_e|\geq
w_3/\ell^4\geq w$. We delete paths from $\pset^*_e$ as necessary, until
$|\pset^*_e|=w$. From the definition of the graph $\tH$, and from
Theorem~\ref{thm: splitting off}, all paths in $\bigcup_{e\in
  E(T^*)}\pset^*_e$ are mutually node-disjoint.

Let $G^*$ be the subgraph of $G$ obtained by taking the union of
$G[S_j]$ for $S_j\in \rset'$, and $\bigcup_{e\in E(T^*)}\pset^*_e$.
We need to verify that each set $S_i$ has the $\alphaWL$-bandwidth
property in $G^*$. Let $\Gamma_i$ be the interface of the set $S_i$ in
$G^*$. We set up a sparsest cut problem instance with the graph
$G[S_i]$ and the set $\Gamma_i$ of terminals, and apply algorithm
\algSC to it. If the outcome is a cut of sparsity less than $\alpha$,
then, since $|\Gamma_i|<w_0$, we obtain an $(w_0,\alpha)$-violating
partition of $S_i$ in graph $G$. We return this partition as the outcome of the
algorithm. If \algSC returns a cut of sparsity at least $\alpha$ for
each set $S_i$, for $1\leq i\leq \ell$, then we are guaranteed that each
such set has the $\alphaWL$-bandwidth property in $G^*$, and we have
therefore constructed a good tree-of-sets system.

\label{----------------------------Proof of corollary w paths--------------------------------}
\section{Extensions}\label{sec: extensions}

The following theorem gives a slightly stronger version of
Theorem~\ref{thm: path-of-sets: main}, that we believe will be useful
in designing approximation algorithms for maximum throughput routing
problems such as All-or-Nothing flow and Disjoint Paths 
in node capacited graphs. For brevity, given a collection $\sset$ of vertex
subsets, we denote $V(\sset)=\bigcup_{S_i\in \sset}S_i$.

\begin{theorem}\label{thm: for routing}
  There is a universal constant $\hat c > 1$ and an efficient
  randomized algorithm, that, given as input (i) a graph $G$ with
  maximum vertex degree $\Delta$; (ii) a subset $\tset^*$ of $k^*$
  vertices of $G$ called terminals that have degree $1$ in $G$, such
  that $\tset^*$ is $\alpha^*$-well-linked in $G$ for some
  $0<\alpha^*<1$, and (iii) parameters $\ell^*,w^*>2$, such that
  $\frac{\alpha^* k^*}{\Delta^{23}\log^{8.5}k^*}>\hat c
  w^*(\ell^*)^{48}$, with high probability computes the following:

\begin{itemize}
\item a subgraph $G^*\subseteq G$;

\item a strong path-of-sets system $(\sset^*,\bigcup_{i=1}^{\ell^*-1}\pset_i,A_1,B_{\ell^*})$ in $G^*$ with $\tset^*\cap V(\sset^*)=\emptyset$, such that every set $S\in \sset^*$  has the  $\alphawl^*=\Omega\left(\frac{(\alpha^*)^2}{\Delta^5(\ell^*)^{16}\log^3k^*}\right)$-bandwidth property in $G^*$; and
\item for each $S\in \sset^*$, a set $\qset(S)$ of paths in $G^*$, connecting every terminal in $\tset^*$ to some vertex of $S$, such that the paths in $\qset(S)$ cause edge-congestion at most $\eta^*=O\left(\frac{\Delta^5(\ell^*)^{16}\log^{3}k^*}{(\alpha^*)^2}\right )$.
\end{itemize}\end{theorem}

We note that in approximation algorithms for routing problems, such as
Edge-Disjoint Paths and Node-Disjoint Paths, the typical setting of
parameters is $\ell^*=\Theta(\log^2k)$ (roughly equal to the number of
rounds in the cut-matching game), $\Delta=\poly\log k^*$, and
$w^*=k^*/\poly\log k^*$, where $k^*$ is the number of terminals.

We briefly explain the relevance of the additional properties
guaranteed by the preceding theorem for routing problems. A
path-of-sets system can be used to embed an expander in $G^*$ and this
in turn can be used as a crossbar routing structure; we refer the reader to
prior work \cite{CKS,Chuzhoy11,ChuzhoyL12} for more details on this
approach. However, a technical issue that arises in using the crossbar
for connecting the given input pairs is the following: we need to
connect the input pairs to the interface of the crossbar. To avoid
additional congestion in the routing, we would like the paths
connecting the terminals to the interface to be disjoint from the
crossbar itself. Theorem~\ref{thm: for routing} helps in addressing
this technical issue. We mention that the following simple
approach does not work to yield the desired properties.  We could
start with a subgraph $G'$ containing a path-of-sets system and then
try to add paths from the terminals $\tset^*$ to each $S \in \sset^*$
by using the well-linkedness properties of the terminals. However,
these paths may add new edges and alter the boundaries of the sets in
$\sset^*$ and hence a set $S \in \sset^*$ which was previously
boundary well-linked in $G'$ may not have the property in the new
subgraph $G''$ obtained from $G'$ by addding the paths from the
terminals to the path-of-sets system.

The remainder of this section is devoted to the proof of
Theorem~\ref{thm: for routing}.  Using Theorem~\ref{thm: grouping}, we
compute a subset $\tset\subseteq \tset^*$ of
$k=\ceil{\frac{\alpha^*k^*}{32\Delta^4\alphasc(k^*)}}$ terminals, such
that $\tset$ is node-well-linked in $G$. From the assumption in
Theorem~\ref{thm: for routing}, $\frac{k}{\Delta^{19}\log^8k}\geq
\frac{k}{\Delta^{19}\log^8k^*}\geq \hat c w^*(\ell^*)^{48}$ for some
large enough constant $\hat c$.  We set $\ell=3(\ell^*)^2+1$ and
$w=\frac{\hat c}{2^{38}\cdot c}\cdot
w^*(\ell^*)^{10}\Delta^{11}\log^4k$, so $w>4\log k$ holds, where $c$
is the constant from Theorem~\ref{thm: meta-tree}.  Clearly:
 
 \[cw\ell^{19}\Delta^8\leq (\frac{\hat c}{2^{38}} w^*(\ell^*)^{10}\Delta^{11}\log^4k)\cdot (2\ell^*)^{38}\Delta^8<\hat c \cdot w^*(\ell^*)^{48}\Delta^{19}\log^4k\leq\frac{k}{\log^4k}.\]
 
 Therefore, $\frac{k}{\log^4k}>cw\ell^{19}\Delta^8$, and the 
 conditions of Theorem~\ref{thm: meta-tree} hold for $G,\tset$ and
 parameters $k,\ell$ and $w$. The main ingredient of the proof of
 Theorem~\ref{thm: for routing} is the following generalization of
 Theorem~\ref{thm: meta-tree} for the construction of a tree-of-sets system.

\begin{theorem}\label{thm: meta-tree and paths}
  There is an efficient randomized algorithm, that, given input as in
  Theorem~\ref{thm: for routing}, the set $\tset\subseteq \tset^*$ of
  terminals, and parameters $\ell,w,\alphawl^*$ and $\eta^*$ as above, with high
  probability computes:
 
 \begin{itemize}
 \item a subgraph $G^*$ of $G$; 
 \item a tree-of-sets system $(\sset,T,\bigcup_{e\in E(T)}\pset_e)$ in
   $G^*$, with parameters $\ceil{\frac{\ell-1}{3}},w$ and $\alphawl=
   \Omega(\frac{1}{\ell^2 \log^{1.5} k})$, such that $\tset^*\cap
   V(\sset)=\emptyset$, and every set $S\in \sset$ has the
   $\alphawl^*$-bandwidth property in $G^*$; and
 \item for each $S\in \sset$, a set $\qset(S)$ of paths in $G^*$,
   connecting every terminal in $\tset^*$ to some vertex of $S$, such
   that the paths in $\qset(S)$ cause edge-congestion at most
   $\eta^*$.
\end{itemize}\end{theorem}

Notice that the definition of the tree-of-sets system only requires
that each set $S_i\in \sset$ has the $\alphawl$-bandwidth property in
the subgraph of $G$ induced by the vertices of the tree-of-sets
system. The preceding theorem requires a slightly stronger property,
that additionally $S_i$ must have the $\alphawl^*$-bandwidth property
in the graph $G^*$, that contains both the tree-of-sets system, and
the set $\qset = \bigcup_{S\in \sset} \qset(S)$ of paths.

We first complete the proof of Theorem~\ref{thm: for routing} assuming
Theorem~\ref{thm: meta-tree and paths}, and then provide a proof of the
latter. This is done exactly as in the proof of Theorem~\ref{thm:
  path-of-sets: main}, by first turning the tree-of-sets system into a
strong one, and then into a path-of-sets system.  Consider the graph
$G^*$, the tree-of-sets system $(\sset,T,\bigcup_{e\in E(T)}\pset_e)$,
and the sets of paths $\set{\qset(S)}_{S\in \sset}$ returned by
Theorem~\ref{thm: meta-tree and paths}.  As before, we use
Lemma~\ref{lem:strong-tree-of-set-system} to convert
$(\sset,T,\bigcup_{e\in E(T)}\pset_e)$ into a strong tree-of-sets
system $(\sset,T,\bigcup_{e\in E(T)}\pset^*_e)$ with parameters $\ell$
and $\tilde{w} = \Omega(\frac{\alphabw^2}{\Delta^{10} (\alphasc(w))^2}
\cdot w)$ using Lemma~\ref{lem:strong-tree-of-set-system}. If $\hat c$ is chosen to be large enough, $\tilde w>
16w^*(\ell^*)^2+1$ must hold. We then apply
Theorem~\ref{thm:tree-of-set-to-path-of-set} to obtain a path-of-sets
system $(\sset^*,\bigcup_{i=1}^{\ell^*-1}\pset_i, A_1, B_{\ell^*})$
with width $w^*$ and length
$\ell^*$. Theorem~\ref{thm:tree-of-set-to-path-of-set} guarantees that
$\sset^*\subseteq \sset$, and hence every set $S\in \sset^*$ still has
the $\alphawl^*$-bandwidth property in $G^*$, and we can use the set
$\qset(S)$ of paths computed for it by
Theorem~\ref{thm: meta-tree and paths}. It remains to prove Theorem~\ref{thm: meta-tree and paths}.

\begin{proofof}{Theorem~\ref{thm: meta-tree and paths}}
%
%
The proof closely follows the proof of Theorem~\ref{thm: meta-tree}, using the parameters $w,\ell,\alphawl$, together with the set $\tset$ of terminals. 
As before, if $(\sset,T,\bigcup_{e\in E(T)}\pset_e)$ is a tree-of-sets system in $G$, with parameters
$w,\ell,\alphawl$, and for each $S_i\in \sset$, $S_i\cap \tset=\emptyset$, then we say that it is a \emph{good tree-of-sets system} (but we allow the sets in $\sset$ to contain terminals of $\tset^*\setminus\tset$).
We define the potential function, acceptable clustering, and good clustering exactly as before, using the parameters $w,\ell,\alphawl$, and we set parameters $\ell_0,w_0,\alpha$ exactly as before, so $\ell_0=\ell^2$ and $w_0=\frac{k}{192\ell_0^3\log k}$. Notice that under this definition of good clustering, the terminals of $\tset^*\setminus\tset$ are treated as regular vertices, and they do not necessarily reside in separate clusters. The algorithm again consists of a number of phases, where the input to every phase is a good clustering $\cset$ of $V(G)$, and the output is either another good clustering $\cset'$ with $\phi(\cset')\leq \phi(\cset)-1$, or a valid output for Theorem~\ref{thm: meta-tree and paths}, that is, a subgraph $G^*$ of $G$, a tree-of-sets system  $(\sset,T,\bigcup_{e\in E(T)}\pset_e)$ in $G^*$, and the sets ${\qset(S)}_{S\in \sset}$ of paths as required. The initial clustering is defined exactly as before: $\set{\set{v}\mid v\in V(G)}$.

We now proceed to describe each phase. Suppose the input to the current phase is a good clustering $\cset$, and let $G'$ be the corresponding legal contracted graph. We find the partition $\set{X_1,\ldots,X_{\ell_0}}$ of $V(G')\setminus \tset$, and compute, for each $1\leq j\leq \ell_0$, an acceptable clustering $\cset_j$ exactly as before. Our only departure from the proof of Theorem~\ref{thm: meta-tree} is that we replace Theorem~\ref{thm: iteration} with the following theorem.

 \begin{theorem}\label{thm: iteration-generalized}
 There is an efficient randomized algorithm, that, given a collection $\set{S_1,\ldots,S_{\ell_0}}$ of disjoint vertex subsets of $G$, where for all $1\leq j\leq \ell_0$, $S_j\cap \tset=\emptyset$,  with high probability computes one of the following:
 
\begin{itemize}
\item either a $(w_0,\alpha)$-violating partition $(X,Y)$ of $S_j$, for some $1\leq j\leq \ell_0$; or
 
\item a partition $(A,B)$ of $V(G)$ with $S_j\sse A$, $\tset\sse B$ and $|E_G(A,B)|<w_0/2$, for some $1\leq j\leq \ell_0$; or
 
\item a valid output for Theorem~\ref{thm: meta-tree and paths}, that is:
 
\begin{itemize}
\item a subgraph $G^*$ of $G$; 
\item a tree-of-sets system $(\sset,T,\bigcup_{e\in E(T)}\pset_e)$ in $G^*$, with parameters $\ceil{\frac{\ell-1}{3}},w$ and $\alphawl$, such that $\tset^*\cap V(\sset)=\emptyset$, and every set $S\in \sset$  has the  $\alphawl^*$-bandwidth property in $G^*$; and
\item for each $S\in \sset$, a set $\qset(S)$ of paths in $G^*$, connecting every terminal in $\tset^*$ to some vertex of $S$, such that the paths in $\qset(S)$ cause edge-congestion at most $\eta^*$.
\end{itemize}
 \end{itemize}
 \end{theorem}

Just as in the proof of Theorem~\ref{thm: meta-tree}, the proof of Theorem~\ref{thm: meta-tree and paths} follows from the proof of Theorem~\ref{thm: iteration-generalized}: We start with the initial collection $\cset_1,\ldots,\cset_{\ell_0}$ of acceptable clusterings, where for each $1\leq j\leq \ell_0$, $\phi(\cset_j)\leq \phi(\cset)-1$. If any of these clusterings $\cset_j$ is a good clustering, then we terminate the phase and return this clustering. Otherwise, each clustering $\cset_j$ must contain a large cluster $S_j\in \cset_j$. We then iteratively apply Theorem~\ref{thm: iteration-generalized} to clusters $\set{S_1,\ldots,S_{\ell_0}}$. If the outcome is a valid output for Theorem~\ref{thm: meta-tree and paths}, then we terminate the algorithm and return this output. Otherwise, we obtain either a $(w_0,\alpha)$-violating partition of some cluster $S_j$, or a partition $(A,B)$ of $V(G)$ with $S_j\sse A$, $\tset\sse B$ and $|E_G(A,B)|<w_0/2$, for some $1\leq j\leq \ell_0$. We then apply the appropriate action: $\partition(S_j,X,Y)$, or $\separate(S_j,A)$ to the clustering $\cset_j$, and obtain an acceptable clustering $\cset'_j$, with $\phi(\cset'_j)\leq \phi(\cset_j)-1/n$. If $\cset'_j$ is a good clustering, then we terminate the phase and return $\cset'_j$. Otherwise, we select an arbitrary large cluster $S'_j$ in $\cset'_j$, replace $S_j$ with $S'_j$ and continue to the next iteration. As before, we are guaranteed that after polynomially-many iterations, the algorithm will terminate with the desired output.

From now on we focus on proving Theorem~\ref{thm: iteration-generalized}.
Given the input collection $\set{S_1,\ldots,S_{\ell_0}}$ of vertex subsets, we run the algorithm from Theorem~\ref{thm: iteration} on it. If the outcome is  
a $(w_0,\alpha)$-violating partition $(X,Y)$ of $S_j$, for some $1\leq j\leq \ell_0$, or  a partition $(A,B)$ of $V(G)$ with $S_j\sse A$, $\tset\sse B$ and
$|E_G(A,B)|<w_0/2$, for some $1\leq j\leq \ell_0$, then we terminate the algorithm and return this partition.
   
Therefore, we can assume from now on that the algorithm from
Theorem~\ref{thm: iteration} has computed a good tree-of-sets system
$(\sset,T,\bigcup_{e\in E(T)}\pset_e)$ in $G$, where
$\sset=\set{S_1,\ldots,S_{\ell_0}}$. Let $U=V(\sset)$. Recall that the
algorithm also ensures that each set $S_j$ can send $w_0/2$ flow units
to the terminals of $\tset$ with no edge-congestion, since otherwise
we could find a cut separating $S_i$ from $\tset$ and containing fewer
than $w_0/2$ edges. While we are guaranteed that $\tset\cap
U=\emptyset$, it is possible that some terminals of
$\tset^*\setminus\tset$ belong to $U$ --- we take care of this issue
later. The rest of the proof consists of three steps. In the first
step, we construct the sets $\qset(S)$ of paths for all $S\in \sset$,
slightly alter the tree $T$ by discarding parts of it, and define the
graph $G^*$. In the second step, we ensure that every set $S\in \sset$
has the $\gamma$-bandwidth property in $G^*$, for a sufficiently large
value $\gamma$. In the final step, we remove the vertices of
$\tset^*\setminus \tset$ from the clusters $S$, and ensure that the
resulting clusters have the $\alphawl^*$-bandwidth property in
$G^*$. If either of these steps fail, then we return a
$(w_0,\alpha)$-violating partition of some cluster $S\in \sset$.

\paragraph{Step 1: finding the sets $\qset(S)$ of paths.}
The following lemma allows us to compute the sets $\qset(S)$ of paths.

\begin{lemma}\label{lem: build paths}
There is an efficient algorithm, that either computes a $(w_0,\alpha)$-violating partition of some set $S'\in \sset$, or computes, for every set $S\in \sset$ a collection $\qset(S)$ of paths, such that:

\begin{itemize}
\item paths in $\qset(S)$ connect every terminal in $\tset^*$ to some vertex of $S$, and they are internally disjoint from $\tset^*\cup S$;
\item paths in $\qset(S)$ cause edge-congestion at most $\eta^*$ in $G$; and
\item for each $S'\in \sset$, for every path $P\in \qset(S)$, $P\cap G[S']$ has at most three connected components.
\end{itemize}
\end{lemma}

\begin{proof}
  We fix some set $S\in \sset$. Recall that there is a flow $F$ of
  value $w_0/2$ from the terminals in $\tset$ to the vertices of $S$
  with edge-congestion at most $1$, where each terminal sends at most
  one unit of flow. This implies, by converting the fractional flow
  into an integral flow, that there is a subset $\tset_0\subseteq
  \tset$ of at least $w_0/2$ terminals, and a collection $\qset_0(S)$
  of edge-disjoint paths, connecting $\tset_0$ to $S$. We partition
  the terminals of $\tset^*\setminus\tset_0$ into $r\leq 4k^*/w_0$
  subsets $\tset_i$, of cardinality at most $w_0/2$ each. From the
  $\alpha^*$-well-linkedness of the terminals, for each such set
  $\tset_i$, there is a collection $\qset_i(S)$ of paths, connecting
  $\tset_i$ to $\tset_0$ with edge-congestion at most $1/\alpha^*$.

  For $i>0$, we view the paths in $\qset_i(S)$ as directed from the
  vertices of $\tset_i$ to the vertices of $\tset_0$, and we view the
  paths of $\qset_0(S)$ as directed from the vertices of $\tset_0$ to
  the vertices of $S$.  Fix some $0\leq i\leq r$. We now re-route the
  paths in $\qset_i(S)$, to ensure that for each cluster $S'\in
  \sset$, the intersection of each such path with $G[S']$ has at most
  two connected components.

  Consider some set $S'\in \sset$. Let $\qset_i(S,S')\subseteq
  \qset_i(S)$ be the following subset of paths: $P \in \qset_i(S,S')$,
  if $P \in \qset_i(S)$ and $P\cap G[S']$ has more than two connected
  components. For $P \in \qset_i(S,S')$, let $\Sigma(P)$ be the set of
  these connected components. We can order these components based on
  the orientation of $P$.  Note that
  the first component of $\Sigma(P)$ may contain a terminal of
  $\tset^*\setminus \tset$, if it belongs to $S'$. In such a case we
  discard the corresponding component from $\Sigma(P)$. We denote by
  $s_P,t_P$ the first and the last vertices of $P$, respectively, that
  belong to any remaining component of $\Sigma(P)$; it can be seen
  that $s_P,t_P\in \Gamma_G(S')$. We will refer to $s_P$ and $t_P$ as
  the source and the destination vertex, respectively, of the pair
  $(s_P,t_P)$. Our goal is to find, for each $P \in \qset(S,S')$, an
  alternate path from $s_p$ to $t_p$ that is completely contained
  inside $G[S']$; we do this by exploiting the boundary well-linkeness
  properties of $S'$.  Let $\mset_i(S')=\set{(s_P,t_P)\mid P\in
    \qset_i(S,S')}$; in fact this is a multi-set since $(s_P,t_P)$ can
  be the source and destination vertices for multiple paths $P$. Observe
  that $|\mset_i(S')|\leq |\qset_i(S)|\leq w_0/2$, and a vertex of
  $\Gamma_G(S')$ may belong to at most $\Delta/\alpha^*$ pairs in
  $\mset_i(S')$ (since the paths in $\qset_i(S)$ cause
  vertex-congestion at most $\Delta/\alpha^*$). Let $\Gamma'\subseteq
  \Gamma_G(S')$ contain all vertices that participate in the pairs in
  $\mset_i(S')$, hence $|\Gamma'|\leq w_0$.

  We use algorithm \algsc in order to approximately compute the
  sparsest cut $(A,B)$ of $G[S']$ with respect to the set $\Gamma'$ of
  terminals. If the sparsity of the cut is at less than $\alpha$, then
  $(A,B)$ is a $(w_0,\alpha)$-violating partition of $S'$. We then
  return this partition and terminate the algorithm. Otherwise, we are
  guaranteed that the vertices of $\Gamma'$ are $\alphawl$-well-linked
  in $G[S']$. Let $X=\set{s_P\mid P\in \qset_i(S,S')}$ and
  $Y=\set{t_P\mid P\in \qset_i(S,S')}$, where each set is a multi-set,
  that is, a vertex $v$ that serves as a source vertex in $n_v$ pairs
  of $\mset_i(S')$ appears $n_v$ times in $X$, and the same holds for
  vertices of $Y$.

\begin{observation}
  There is a collection $\rset$ of paths in $G[S']$, connecting
  vertices $X$ to vertices $Y$ with edge-congestion at most
  $O\left(\frac{\Delta}{\alphawl\cdot \alpha^*}\right )$, such that if
  $v$ appears $n_v$ times in $X$ then exactly $n_v$ paths of $\rset$
  originate at $v$, and if $v$ appears $n'_v$ times in $Y$, then
  exactly $n'_v$ paths of $\rset$ terminate at $v$.
\end{observation}
\begin{proof}
  Since every vertex of $X\cup Y$ may participate in at most
  $\Delta/\alpha^*$ pairs in $\mset_i(S')$, we can partition the set
  $\mset_i(S')$ into at most $z=\ceil{2\Delta/\alpha^*}$ subsets
  $\nset_1,\ldots,\nset_z$, such that for all $1\leq j\leq z$, every
  vertex of $X\cup Y$ participates in at most one pair in
  $\nset_j$. For each $1\leq j\leq z$, we then denote by $X_j$ and
  $Y_j$ the sets of all source and all destination vertices,
  respectively, of the pairs in $\nset_j$. Since $S'$ has the
  $\alphawl$-bandwidth property, we are guaranteed that $X_j\cup Y_j$
  is $\alphawl$-well-linked in $G[S']$. Therefore, there is a set
  $\rset_j$ of paths in $G[S']$, connecting every vertex of $X_j$ to a
  distinct vertex of $Y_j$, with edge-congestion at most
  $1/\alphawl$. We then set $\rset=\bigcup_j\rset_j$.
\end{proof}

For every path $P\in \qset_i(S,S')$, we discard the segment of the
path between $s_P$ and $t_P$, obtaining two sub-paths $P_1,P_2$ of
$P$. We then use the paths in $\rset$ in order to glue all path
segments in $\set{P_1\mid P\in \qset_i(S,S')}$ and $\set{P_2\mid P\in
  \qset_i(S,S')}$, obtaining a new collection of paths, connecting
each vertex of $\tset_i$ to a distinct vertex of $\tset_0$ (if $i=0$,
then the paths connect every vertex of $\tset_0$ to some vertex of
$S$). These new paths then replace the paths of $\qset_i(S,S')$ in
$\qset_i(S)$. Note that after rerouting, a terminal in $\tset_i$ (if $i > 0$)
may not connect to the same terminal in $\tset_0$ as it did previously.

Once we process every cluster $S'\in \sset$ in this fashion, the final
set $\qset_i(S)$ of paths causes edge-congestion at most
$O\left(\frac{\Delta}{\alphawl\cdot \alpha^*}\right )$, and has the
property that for every $S'\in \sset$ and every path $P\in
\qset_i(S)$, $P\cap G[S']$ contains at most two connected components
(if $P\in \qset_0(S')$, then $P\cap G[S']$ may contain at most one
connected component, as $P$ originates from a vertex of $\tset$, that
cannot lie in $S'$). In our final step, we take the union of all paths
in $\bigcup_{i=1}^r\qset_i(S)$, and concatenate them with $r$ copies
of the paths in $\qset_0(S)$. This final set of paths is denoted by
$\qset(S)$. It is immediate to verify that the paths in $\qset(S)$
connect every vertex of $\tset^*$ to a vertex of $S$, and for all
$S'\in \sset$ and $P\in \qset(S)$, $P\cap G[S']$ has at most three
connected components. The congestion caused by the paths in $\qset(S)$
is bounded by:

\[\begin{split}2r\cdot O\left(\frac{\Delta}{\alphawl\cdot \alpha^*}\right )&=O\left(\frac{k^*\cdot \Delta}{w_0\cdot \alphawl\cdot \alpha^*}\right )\\
&=O\left(\frac{k^*\Delta\ell_0^3\log k }{k \cdot \alpha^*}\cdot \ell_0\log k\alphasc(k)\right )\\
&=O\left(\frac{k^*\Delta\ell^8\log^{2.5} k}{\alpha^*}\cdot\frac{\Delta^4\alphasc(k^*)}{\alpha^*k^*}\right )\\
&=O\left (\frac{\Delta^5\ell^8\log^{3} k^*}{(\alpha^*)^2}\right )\\
&=O\left(\frac{\Delta^5(\ell^*)^{16}\log^{3}k^*}{(\alpha^*)^2}\right )=\eta^*.
\end{split}\]
\end{proof}

Let $\sset=\set{S_1,\ldots,S_{\ell}}$, and let
$V(T)=\set{v_1,\ldots,v_{\ell}}$, where $v_i$ is the vertex
corresponding to the clusters $S_i$. For a cluster $S_i$, let
$\tset(S_i)=\tset^*\cap S_i$, and let $i^* = \argmax_i |\tset(S_i)|$
be the index of the cluster containing the largest number of
terminals. Let $T'$ be the largest connected component of $T\setminus
\set{v_{i^*}}$. Since $T$ has maximum degree $3$, $|V(T')|\geq
\ceil{\frac{\ell-1}{3}}$. If $|V(T')|>\ceil{\frac{\ell-1}{3}}$, 
we discard leaves of $T'$ until the equality holds. We discard from
$\sset$ all clusters except those corresponding to the vertices of
$T'$, obtaining a new tree-of-sets system $(T',\sset,\bigcup_{e\in
  E(T')}\pset(e))$. For simplicity, we will denote $T'$ by $T$ from now
on. The new tree-of-sets system has the property that for each $S_i\in
\sset$, $|\tset(S_i)|\leq k^*/2$. Since the terminals of $\tset^*$ are
$\alpha^*$-well-linked, there is a set $\rset_i$ of paths in $G$,
connecting every terminal in $\tset(S_i)$ to some terminal of
$\tset^*\setminus\tset(S_i)$ with edge-congestion at most
$1/\alpha^*$. By appropriately truncating each path $P$ in $\rset_i$,
we can ensure that it terminates at a vertex $u_P\in \Gamma_G(S_i)$,
and that $P\subseteq G[S_i]$. We denote by $R_i\subseteq
\Gamma_G(S_i)$ the set of endpoints of the resulting paths in
$\rset_i$. Then $|R_i|\leq k^*/2$, and every vertex in $R_i$ serves as
an endpoint of at most $\Delta/\alpha^*$ paths in $\rset_i$.

We are now ready to define the graph $G^*$. This graph is the union of
all subgraphs $G[S_i]$ for $S_i\in \sset$ and paths
$\left(\bigcup_{e\in E(T)}\pset(e)\right )\cup \left(\bigcup_{S_i\in
    \sset}\qset(S_i)\right )$.

\paragraph{Step 2: Ensuring Bandwidth Property of Clusters}

Consider some cluster $S_i\in \sset$, and recall that we have already
defined a subset $R_i\subseteq \Gamma_G(S_i)$ of its vertices. We let
$R'_i=\Gamma_{G^*}(S_i)$. Recall that $|R_i|\leq k^*/2$, and set
$R'_i$ contains, for each path $Q\in \bigcup_{S\in \sset}\qset(S)$,
at most six vertices of $Q$ (as $Q\cap G[S_i]$ contains at most three
connected components). Set $R'_i$ also contains at most $3w$ vertices for
the paths from $\bigcup_{e \in E(T)} \pset_e$ that terminate in $S_i$.
Let $\hat R_i=R_i\cup R'_i$. Then:

\[|\hat R_i|\leq \frac{k^*}{2}+6\ceil{\frac{\ell-1}{3}}\cdot \frac{w_0}{2} + 3w.\]

Since $k^*=\Theta\left(\frac{\Delta^4k\alphasc(k^*)}{\alpha^*}\right
)=\Theta\left(\frac{w_0\Delta^4\ell^6\alphasc(k^*)\log
    k}{\alpha^*}\right)=\Theta\left(\frac{w_0\Delta^4\ell^6\log^{1.5}k^*}{\alpha^*}\right
)$, and $w \le k^*$, we get that $|\hat R_i|=w_0\cdot
O\left(\frac{\Delta^4\ell^6\log^{1.5}k^*}{\alpha^*}\right )$.

We let $\rho=\Theta\left(\frac{\Delta^4\ell^6\log^{1.5}k^*}{\alpha^*}\right )$, so that $|\hat R_i|\leq w_0\rho$, and we denote $\gamma=\frac{\alpha}{\rho\alphasc(k^*)}$. We use the following claim.

 \begin{claim}\label{claim: bw prop of sets}
 There is an efficient algorithm, that, given a cluster $S_i\in \sset$, either certifies that $\hat R_i$ is $\gamma$-well-linked in $G[S_i]$, or returns a $(w_0,\alpha)$-violating partition of $S_i$.
 \end{claim}
 
 \begin{proof}
 We use algorithm \algsc in order to approximately compute the sparsest cut $(A,B)$ of $G[S_i]$ with respect to the set $\hat R_i$ of terminals. If the sparsity of the cut is at least $\gamma\cdot\alphasc(k^*)$, then we are guaranteed that $S$ has the $\gamma$-bandwidth property in $G^*$. Assume now that the sparsity of the cut is less than $\gamma\cdot \alphasc(k^*)$. We claim that in this case, cut $(A,B)$ is a $(w_0,\alpha)$-violating partition of $S$. 
 
Indeed, 

\[\begin{split}
|E(A,B)|&<\gamma\cdot \alphasc(k^*)\cdot \min\set {|A\cap \hat R_i|, |B\cap \hat R_i|}\\
&=\frac{\alpha}{\rho}\cdot  \min\set {|A\cap \hat R_i|, |B\cap \hat R_i|}\\
&\leq \frac{\alpha |\hat R_i|}{2\rho}\\
&\leq \frac{\alpha w_0}{2}.\end{split}\]

In particular, this shows that $|E(A,B)|<\alpha\cdot \min{\set{|A\cap
    \hat R_i|, |B\cap \hat R_i|,w_0/2}}$, so this is indeed a
$(w_0,\alpha)$-violating partition of $S_i$.\end{proof}

If for any set $S_i\in \sset$, Claim~\ref{claim: bw prop of sets}
returns a $(w_0,\alpha)$-violating partition, then we terminate the
algorithm and return this partition. Therefore, we assume from now on
that for each set $S_i$, the vertex set $\hat R_i$ is
$\gamma$-well-linked in $G[S_i]$

So far we have obtained a graph $G^*\subseteq G$, a tree-of sets
system $(T,\sset,\bigcup_{e\in E(T)}\pset(e))$, and the sets
$\set{\qset(S)}_{S\in \sset}$ of paths as required, except that it is
still possible that the terminals of $\tset^*\setminus \tset$ belong
to the vertex sets $S\in \sset$.  We rectify this in our final step.

\paragraph{Step 3: Removing the Terminals from the Clusters.}


In this step, we define a new tree-of-sets system, by replacing every
cluster $S\in \sset$ with cluster $S'=S\setminus \tset^*$. Let
$\sset'$ denote the resulting set of clusters. Recall that the
terminals in $\tset^*$ all have degree $1$ in $G$, and hence they cannot
participate in the paths $\bigcup_{e\in E(T)}\pset(e)$. Therefore,
$(\sset',T,\bigcup_{e\in E(T)}\pset(e))$ remains a valid tree-of-sets
system. Consider a cluster $S'_i\in \sset'$, and let
$\Gamma'_i\subseteq S'_i$ be the set of vertices serving as endpoints
of the paths in $\pset(e)$ for all edges $e\in E(T)$ incident to the
vertex $v_i\in V(T)$ that corresponds to the set $S_i$. Since the
vertices of $\tset^*\cap S_i$ all have degree $1$, their removal from
$S_i$ does not affect the well-linkedness of the set $\Gamma'_i$ of
vertices, and hence $\Gamma'_i$ remains $\alphabw$-well-linked in
$G[S_i']$. Similarly, the set $\hat R_i$ of vertices remains
$\gamma$-well-linked in $G[S'_i]$.

If some vertex $t\in \tset^*$ originally belonged to some cluster
$S_i$, then without loss of generality, its corresponding path in
$\qset(S_i)$ contained a single vertex --- the vertex $t$. We now
replace this path with a path containing a single edge, connecting $t$
to its unique neighbor $u_t$, that must belong to $S'_i$. Notice that
vertex $u_t$ now belongs to the boundary of $S'_i$ in $G^*$, even
though it may not belong to the boundary of $S_i$. The resulting set
$\qset(S_i)$ of paths still connects all vertices of $\tset^*$ to the
vertices of $S'_i$ with edge-congestion at most $\eta^*$, but now we
need to prove that each resulting cluster $S'_i\in \sset'$ has the
$\alphawl^*$-bandwidth property in $G^*$. The following claim will
finish the proof of Theorem~\ref{thm: iteration-generalized}.

\begin{claim} Each set $S'_i\in \sset'$ has the $\alphawl^*$-bandwidth
  property in $G^*$.
\end{claim}
\begin{proof}

  Let $\Gamma'=\Gamma_{G^*}(S'_i)\cup R_i$. It suffices to prove that
  $\Gamma'$ is $\alphawl^*$-well-linked in $G[S_i]$. Consider a
  partition $(A,B)$ of $S'_i$. Let $Z_A=\Gamma'\cap A$ and
  $Z_B=\Gamma'\cap
  B$.

  We partition the vertices of $Z_A$ into two subsets: $Z'_A$ contains
  all vertices that belonged to $\hat R_i$, and $Z''_A$ contains all
  remaining vertices, so each vertex in $Z''_A$ is a neighbor of some
  terminal in $\tset^*\cap S_i$. We define a partition of $Z_B$ into
  $Z'_B$ and $Z''_B$ similarly. We now consider two cases.

  Assume first that both $|Z'_A|\geq \frac{\alpha^* |Z''_A|}{2\Delta}$
  and $|Z'_B|\geq \frac{\alpha^* |Z''_B|}{2\Delta}$. In this case,
  since $\hat R_i$ is $\gamma$-well-linked in $G[S_i]$ (from
  Claim~\ref{claim: bw prop of sets}), we get that:

\[|E(A,B)|\geq \gamma\min\set{|Z'_A|,|Z'_B|}\geq \frac{\alpha^* \gamma}{4\Delta}\cdot \min\set{|Z_A|,|Z_B|}.\]

Recall that $\gamma=\frac{\alpha}{\rho\alphasc(k^*)}=\Theta\left(\frac{\alpha\alpha^*}{\Delta^4\ell^6\log^2k^*}\right )$, while $\alpha=\Omega\left(\frac{1}{\ell^2\log k}\right )=\Omega\left(\frac{1}{\ell^2\log k^*}\right )$, and hence:

\[\begin{split}
|E(A,B)|&\geq\Omega\left ( \frac{\alpha (\alpha^*)^2}{\Delta^5\ell^6\log^2k^*}\right )\cdot \min\set{|Z_A|,|Z_B|}
\\ &\geq \Omega\left(\frac{(\alpha^*)^2}{\Delta^5\ell^{8}\log^3k^*}\right)\cdot\min\set{|Z_A|,|Z_B|}\\
&=\Omega\left(\frac{(\alpha^*)^2}{\Delta^5(\ell^*)^{16}\log^3k^*}\right)\cdot\min\set{|Z_A|,|Z_B|}\\
&\geq \alphawl^*\cdot \min\set{|Z_A|,|Z_B|}.\end{split}\]

Assume now that $|Z'_A|<\frac{\alpha^* |Z''_A|}{2\Delta}$. Let
$\tset''\subseteq \tset^*$ be the set of terminals $t$, such that
$t\in S_i$, and the unique neighbor of $t$ belongs to $Z''_A$. Then
every vertex in $Z''_A$ has a neighbor in $\tset''$, and
$|\tset''|\geq |Z''_A|\geq 2\Delta|Z'_A|/\alpha^*$.  Recall that we
have defined a set $\rset_i$ of paths in $G[S_i]$, connecting
$\tset^*\cap S_i$ to $R_i\subseteq \hat R_i$, with edge-congestion at
most $\alpha^*/\Delta$. Let $\rset'\subseteq\rset_i$ be the set of
paths originating at the vertices of $\tset''$. Then at most half the
paths in $\rset'$ may terminate at the vertices of $Z'_A$, and each
one of the remaining paths must contain an edge of $E(A,B)$. Since the
paths cause edge-congestion at most $\alpha^*$, we conclude that:

\[E(A,B)\geq \frac{\alpha^*\cdot |Z''_A|}{2} \geq \frac{\alpha^*}{4}|Z''(A)|+2\Delta |Z'_A|\geq \frac{\alpha^*}{4}|Z(A)|>\alphawl^*|Z(A)|.\]

The case where $|Z'_B|<\frac{\alpha^*|Z''_B|}{2\Delta}$ is analyzed similarly.
\end{proof}
\end{proofof}

\paragraph{Acknowledgement.} We thank Tasos Sidiropoulos 
for suggesting the problem of improving the bounds in the
grid minor theorem, and in particular for asking whether 
the ideas in routing algorithms are useful for this purpose.
We also thank the two reviewers for detailed comments and
suggestions.

\bibliographystyle{alpha}
\bibliography{improved-GMT}

\begin{thebibliography}{DFHT05}

\bibitem[And10]{Andrews}
Matthew Andrews.
\newblock Approximation algorithms for the edge-disjoint paths problem via
  {Raecke} decompositions.
\newblock In {\em Proceedings of IEEE FOCS}, pages 277--286, 2010.

\bibitem[ARV09]{ARV}
Sanjeev Arora, Satish Rao, and Umesh~V. Vazirani.
\newblock Expander flows, geometric embeddings and graph partitioning.
\newblock {\em J. ACM}, 56(2), 2009.

\bibitem[CC13]{ChekuriC13}
Chandra Chekuri and Julia Chuzhoy.
\newblock Large-treewidth graph decompositions and applications.
\newblock In {\em Proc.\ of ACM STOC}, pages 291--300, 2013.

\bibitem[CC15]{tw-sparsifiers}
Chandra Chekuri and Julia Chuzhoy.
\newblock Degree-3 treewidth sparsifiers.
\newblock In {\em Proceedings of the Twenty-Sixth Annual ACM-SIAM Symposium on
  Discrete Algorithms}, pages 242--255. SIAM, 2015.

\bibitem[CE13]{ChekuriE13}
Chandra Chekuri and Alina Ene.
\newblock Poly-logarithmic approximation for maximum node disjoint paths with
  constant congestion.
\newblock In {\em Proc.\ of ACM-SIAM SODA}, 2013.

\bibitem[CHR03]{CHR}
M.~Conforti, R.~Hassin, and R.~Ravi.
\newblock Reconstructing edge-disjoint paths.
\newblock {\em Operations Research Letters}, 31(4):273--276, 2003.

\bibitem[Chu12]{Chuzhoy11}
Julia Chuzhoy.
\newblock Routing in undirected graphs with constant congestion.
\newblock In {\em Proc.\ of ACM STOC}, pages 855--874, 2012.

\bibitem[Chu15]{GMT2}
Julia Chuzhoy.
\newblock Excluded grid theorem: Improved and simplified.
\newblock In {\em Proceedings of the Forty-Seventh Annual ACM on Symposium on
  Theory of Computing}, STOC '15, pages 645--654, New York, NY, USA, 2015. ACM.

\bibitem[CK09]{ChekuriK09}
Chandra Chekuri and Nitish Korula.
\newblock A graph reduction step preserving element-connectivity and
  applications.
\newblock In {\em Proc.\ of ICALP}, pages 254--265, 2009.

\bibitem[CKS05]{CKS}
Chandra Chekuri, Sanjeev Khanna, and F.~Bruce Shepherd.
\newblock Multicommodity flow, well-linked terminals, and routing problems.
\newblock In {\em Proc.\ of ACM STOC}, pages 183--192, 2005.

\bibitem[CKS13]{ANF}
Chandra Chekuri, Sanjeev Khanna, and F.~Bruce Shepherd.
\newblock The all-or-nothing multicommodity flow problem.
\newblock {\em SIAM Journal on Computing}, 42(4):1467--1493, 2013.
\newblock Preliminary version in {\em Proc.\ of ACM STOC}, 2004.

\bibitem[CL12]{ChuzhoyL12}
Julia Chuzhoy and Shi Li.
\newblock A polylogarithimic approximation algorithm for edge-disjoint paths
  with congestion 2.
\newblock In {\em Proc.\ of IEEE FOCS}, 2012.

\bibitem[CNS13]{ChekuriNS13}
Chandra Chekuri, Guyslain Naves, and F.~Bruce Shepherd.
\newblock Maximum edge-disjoint paths in $k$-sums of graphs.
\newblock {\em CoRR}, abs/1303.4897, 2013.
\newblock Extended abstract in Proc.\ of ICALP, 2013.

\bibitem[DFHT05]{DemaineFHT05}
Erik~D. Demaine, Fedor~V. Fomin, Mohammad~Taghi Hajiaghayi, and Dimitrios~M.
  Thilikos.
\newblock Subexponential parameterized algorithms on bounded-genus graphs and
  {\it h}-minor-free graphs.
\newblock {\em J. ACM}, 52(6):866--893, 2005.

\bibitem[DH08]{DemaineH08}
Erik~D. Demaine and MohammadTaghi Hajiaghayi.
\newblock Linearity of grid minors in treewidth with applications through
  bidimensionality.
\newblock {\em Combinatorica}, 28(1):19--36, 2008.

\bibitem[DHK09]{DemaineHK09}
Erik Demaine, MohammadTaghi Hajiaghayi, and Ken-ichi Kawarabayashi.
\newblock Algorithmic graph minor theory: Improved grid minor bounds and
  wagner’s contraction.
\newblock {\em Algorithmica}, 54:142--180, 2009.

\bibitem[Die12]{Diestel-book}
Reinhard Diestel.
\newblock {\em Graph Theory, 4th Edition}, volume 173 of {\em Graduate texts in
  mathematics}.
\newblock Springer, 2012.

\bibitem[DJGT99]{DiestelJGT99}
Reinhard Diestel, Tommy~R. Jensen, Konstantin~Yu. Gorbunov, and Carsten
  Thomassen.
\newblock Highly connected sets and the excluded grid theorem.
\newblock {\em J. Comb. Theory, Ser. B}, 75(1):61--73, 1999.

\bibitem[DP09]{measure-concentration}
Devdatt Dubhashi and Alessandro Panconesi.
\newblock {\em Concentration of Measure for the Analysis of Randomized
  Algorithms}.
\newblock Cambridge University Press, New York, NY, USA, 1st edition, 2009.

\bibitem[HO96]{element-connectivity}
H.~R. Hind and O.~Oellermann.
\newblock Menger-type results for three or more vertices.
\newblock {\em Congressus Numerantium}, 113:179--204, 1996.

\bibitem[HS70]{Hajnal-Szemeredi}
A.~Hajnal and E.~Szemer\'edi.
\newblock Proof of a conjecture of {P. Erdos}.
\newblock {\em Combinatorial Theory and its Applications}, pages 601--623,
  1970.

\bibitem[KK12]{KawarabayashiK-grid}
K.~Kawarabayashi and Y.~Kobayashi.
\newblock {Linear min-max relation between the treewidth of H-minor-free graphs
  and its largest grid minor}.
\newblock In {\em Proc.\ of STACS}, 2012.

\bibitem[KRV09]{KRV}
Rohit Khandekar, Satish Rao, and Umesh Vazirani.
\newblock Graph partitioning using single commodity flows.
\newblock {\em J. ACM}, 56(4):19:1--19:15, July 2009.

\bibitem[KT10]{KreutzerT10}
Stephan Kreutzer and Siamak Tazari.
\newblock On brambles, grid-like minors, and parameterized intractability of
  monadic second-order logic.
\newblock In {\em Proc.\ of ACM-SIAM SODA}, pages 354--364, 2010.

\bibitem[LS15]{LeafS12}
Alexander Leaf and Paul Seymour.
\newblock Tree-width and planar minors.
\newblock {\em J. Comb. Theory Ser. B}, 111(C):38--53, March 2015.

\bibitem[Mad78]{edge-connectivity}
W.~Mader.
\newblock A reduction method for edge connectivity in graphs.
\newblock {\em Ann. Discrete Math.}, 3:145--164, 1978.

\bibitem[OSVV08]{better-CMG}
Lorenzo Orecchia, Leonard~J. Schulman, Umesh~V. Vazirani, and Nisheeth~K.
  Vishnoi.
\newblock On partitioning graphs via single commodity flows.
\newblock In {\em Proceedings of the 40th annual ACM symposium on Theory of
  computing}, STOC '08, pages 461--470, New York, NY, USA, 2008. ACM.

\bibitem[R{\"a}c02]{Raecke}
Harald R{\"a}cke.
\newblock Minimizing congestion in general networks.
\newblock In {\em Proc.\ of IEEE FOCS}, pages 43--52, 2002.

\bibitem[Ree97]{Reed-chapter}
Bruce Reed.
\newblock {\em Surveys in Combinatorics}, chapter Treewidth and Tangles: A New
  Connectivity Measure and Some Applications.
\newblock London Mathematical Society Lecture Note Series. Cambridge University
  Press, 1997.

\bibitem[RS86]{RS-grid}
Neil Robertson and P~D Seymour.
\newblock {Graph minors. V. Excluding a planar graph}.
\newblock {\em Journal of Combinatorial Theory, Series B}, 41(1):92--114,
  August 1986.

\bibitem[RS91]{graphminorsx}
Neil Robertson and Paul~D Seymour.
\newblock Graph minors. x. obstructions to tree-decomposition.
\newblock {\em Journal of Combinatorial Theory, Series B}, 52(2):153--190,
  1991.

\bibitem[RST94]{RobertsonST94}
N~Robertson, P~Seymour, and R~Thomas.
\newblock {Quickly Excluding a Planar Graph}.
\newblock {\em Journal of Combinatorial Theory, Series B}, 62(2):323--348,
  November 1994.

\bibitem[RZ10]{RaoZhou}
Satish Rao and Shuheng Zhou.
\newblock Edge disjoint paths in moderately connected graphs.
\newblock {\em SIAM J. Comput.}, 39(5):1856--1887, 2010.

\bibitem[Sch03]{Schrijver}
Alexander Schrijver.
\newblock {\em Combinatrial Optimization: Polyhedra and Efficiency}, volume~24
  of {\em Algorithms and Combinatorics}.
\newblock Springer, 2003.

\bibitem[Sey]{PS-comm}
Paul Seymour.
\newblock personal communication.

\bibitem[SL15]{Singh-Lau}
Mohit Singh and Lap~Chi Lau.
\newblock Approximating minimum bounded degree spanning trees to within one of
  optimal.
\newblock {\em J. {ACM}}, 62(1):1:1--1:19, 2015.
\newblock Preliminary version in {\em Proc.\ of ACM STOC}, 2007.

\bibitem[Win89]{graph-toughness}
Sein Win.
\newblock On a connection between the existence of k-trees and the toughness of
  a graph.
\newblock {\em Graphs and Combinatorics}, 5:201--205, 1989.

\end{thebibliography}

\newpage

\appendix
\label{---------------------------------------Appendix-------------------------------------}
\label{--------------------------------------------Proofs from Prelims---------------------------------------------------}
\section{Proofs Omitted from Section~\ref{sec: prelims}}

\label{---------------------------------------Grouping thm proof-------------------------------------}

\subsection{Proof of Theorem~\ref{thm: grouping}}
We start with a non-constructive proof, since it is much simpler and
gives better parameters. This proof can be turned into an algorithm
whose running time is $\poly(n)\cdot 2^{\kappa}$. We then show a
constructive proof with running time $\poly(n,\kappa)$.

\subsubsection{A non-constructive proof}
A separation in graph $G$ is two subgraphs $Y,Z$ of $G$, such that
every edge of $G$ belongs to exactly one of $Y,Z$, and $G=Y\cup Z$. The order of the
separation is $|V(Y)\cap V(Z)|$. We say that a separation $(Y,Z)$ is
\emph{balanced with respect to $\tset$}, if and only if $|V(Y)\cap \tset|,|V(Z)\cap \tset|\geq
|\tset|/4$. Let $(Y,Z)$ be a balanced separation of $G$ with respect to $\tset$ of minimum
order, and let $X=V(Y)\cap V(Z)$. Assume without loss of generality
that $|V(Y)\cap \tset|\geq |V(Z)\cap \tset|$, so $|V(Y)\cap \tset|\geq
|\tset|/2$. We claim that $X$ is node-well-linked in graph $Y$.

\begin{claim}
Set $X$ of vertices is node-well-linked in graph $Y$.
\end{claim}

\begin{proof}
  Let $A,B$ be any two equal-sized subsets of $X$, and assume that
  $|A|=|B|=z$. It is enough to show that there is a set $\pset$ of $z$
  disjoint paths connecting $A$ to $B$ in $Y$. Assume otherwise. Then
  there is a set $S$ of at most $z-1$ vertices separating $A$ from $B$
  in $Y$.

  Let $\cset$ be the set of all connected components of $Y\setminus
  S$. We partition $\cset$ into three subsets: $\cset_1$ contains all
  components containing the vertices of $A$; $\cset_2$ contains all
  components containing the vertices of $B$, and $\cset_3$ contains all
  remaining components (notice that all three sets of clusters are
  pairwise disjoint). Let $R_1=\bigcup_{C\in \cset_1}V(C)$, and define
  $R_2$ and $R_3$ for $\cset_2$ and $\cset_3$, respectively. Assume
  without loss of generality that $|R_1\cap \tset|\geq |R_2\cap
  \tset|$. We define a new separation $(Y',Z')$, as follows. The set
  of vertices $V(Y')=R_1\cup R_3\cup S$, and $V(Z')=V(Z)\cup R_2\cup
  S$. Let $X'=V(Y')\cap V(Z')$. The edges of $Y'$ include all edges of
  $G$ with both endpoints in $V(Y')\setminus X'$, and all edges of $G$
  with one endpoint in $V(Y')\setminus X'$ and the other endpoint in
  $X'$. The edges of $Z'$ include all edges with both endpoints in
  $Z'$.

  We claim that $(Y',Z')$ is a balanced separation with respect to $\tset$. Clearly,
  $|V(Z')\cap \tset|\geq |\tset|/4$, since $V(Z)\subseteq V(Z')$, and
  $|V(Z)\cap \tset|\geq |\tset|/4$.  We next claim that $|V(Y')\cap
  \tset|\geq |\tset|/4$. Assume otherwise. Then, from our assumption,
  $|R_2\cap \tset|<|\tset|/4$, and so $|V(Y)\cap \tset|=|R_2\cap \tset|+|V(Y')\cap \tset|<|\tset|/2$, a
  contradiction. Therefore, $(Y',Z')$ is a balanced
  separator with respect to $\tset$. Finally, we claim that its order is less than $|X|$,
  contradicting the minimality of $X$. Indeed, $|V(Y')\cap V(Z')|\leq
  |X|-|B|+|S|<|X|$.
\end{proof}

Let $\tset_1=\tset\cap V(Z)$, $\tset_2=\tset\cap V(Y)$, and let $\tset_1'\subseteq \tset_1$, $\tset_2'\subseteq \tset_2$ be two disjoint subsets containing $\ceil{\kappa /4}$ vertices each. From Observations~\ref{obs: wl-properties} and~\ref{obs: sparsest cut to flow}, there is a flow $F$ from $\tset_1'$ to $\tset_2'$, such that every vertex in $\tset_1'$ sends one flow unit, every vertex in $\tset_2'$ receives one flow unit, and the congestion on every edge is at most $1/\alpha$. We now bound the vertex-congestion caused by the flow $F$. For every vertex $v\in V(G)$, let $F_1(v)$ be the total amount of flow on all paths that originate or terminate at $v$, and let $F_2(v)$ be the total amount of flow on all paths that contain $v$ as an inner vertex. It is immediate to verify that $F_1(v)\leq 1$, while $F_2(v)\leq \frac{\Delta}{2\alpha}$, since every flow-path $P$ that contains $v$ as an inner vertex contributes flow $F(P)$ to two edges incident to $v$. Therefore, the total flow through $v$ is at most $\frac{\Delta}{2\alpha}+1\leq \frac{5\Delta}{6\alpha}$, as $\Delta\geq 3$ and $\alpha\leq 1$.
By sending $\frac{6\alpha}{5\Delta}\cdot F(P)$ flow units via every path $P$, we obtain a flow of value at least $\frac{\kappa}{4}\cdot \frac{6\alpha}{5\Delta}=\frac{3\alpha \kappa}{10\Delta}$ from vertices of $\tset_1'$ to vertices of $\tset_2'$, that causes vertex-congestion at most $1$. From the integrality of flow, there is a set 
%
%
%
$\pset'$ of $\kappa'=\ceil{\frac{3\alpha \kappa}{10\Delta}}$
node-disjoint paths connecting terminals in $\tset_1'$ to terminals in
$\tset_2'$ in $G$. Each such path has to contain a vertex of $X$. For each
path $P'\in \pset'$, we truncate the path $P'$ to the first vertex of
$X$ on $P'$ (where the path is directed from $\tset_1$ to
$\tset_2$). Let $\pset$ be the resulting set of truncated paths. Then
$\pset$ is a set of $\kappa'$ disjoint paths, connecting  vertices
of $\tset_1'$ to vertices of $X$; every path in $\pset$ is
completely contained in graph $Z$, and is disjoint from $X$ except for
its last endpoint that belongs to $X$.

Let $\tset''\subseteq \tset_1'$ be the set of terminals from which the
paths in $\pset$ originate, and let $X'\subseteq X$ be the set of
vertices where they terminate. We claim that $\tset''$ is
node-well-linked in $G$. Indeed, let $A,B\subseteq\tset''$ be any pair of
equal-sized subsets of terminals.  Let $U=A\cap B$, $A'=A\setminus U$
and $B'=B\setminus U$. 

We define the set $\tilde{A}'\subseteq X'$ as follows: for each terminal
$t\in A'$, let $P_t\in \pset$ be the path originating at $t$, and let
$x_t$ be its other endpoint, that belongs to $X$. We then set
$\tilde{A}'=\set{x_t\mid t\in A'}$. We define a set $\tB'\subseteq X$
similarly for $B'$. Let $\pset_A\subseteq \pset$ be the set of paths
originating at the vertices of $A'$, and let $\pset_B\subseteq \pset$
be the set of paths originating at the vertices of $B'$.  Notice that
both sets of paths are contained in $Z$, and are internally disjoint
from $X$. The paths in $\pset_A\cup \pset_B$ are also mutually
disjoint, and they avoid $U$.

Let $U'=U\cap X$, and consider the two subsets $\tA=\tA'\cup U'$ and
$\tB=\tB'\cup U'$ of vertices of $X$. Denote $|\tA|=|\tB|=z$.  Since
$X$ is node-well-linked in $Y$, there is a set $\qset$ of $z$ disjoint
paths connecting $\tA$ to $\tB$ in $Y$. The paths in $\qset$ are then
completely disjoint from the paths in $\pset_1,\pset_2$ (except for
sharing endpoints with them). The final set of paths connecting $A$ to
$B$ is obtained by concatenating the paths in $\pset_1,\qset,\pset_2$, and adding a collection $\qset'$ of paths that contains, for every vertex $v\in U\setminus U'$, a path $P_v$ consisting of only the vertex $v$ itself.



\subsubsection{A Constructive Proof}

We assume that $\kappa\geq 32\Delta^4\alphasc(\kappa)/\alpha$, since otherwise we can return a set $\tset'$ consisting of a single terminal.
For every subset $C\subseteq V$ of vertices, let $\tset_C=C\cap
\tset$.  We say that a partition $(A,B)$ of $V$ is \emph{balanced}
with respect to $\tset$, if $|\tset_A|,|\tset_B|\geq
\frac{\kappa}{2\Delta}$. We need the following lemma that follows from the well-linkedness
of $\tset$.
\begin{lemma}
  \label{lem:tB-tA-paths}
  Let $(A,B)$ be any balanced partition of $V$ with respect to $\tset$. There is an efficient
  algorithm that computes a collection $\pset$ of node-disjoint paths
  from $\tset_B$ to $\tset_A$ where $|\pset| \ge \ceil{\frac{\kappa
    \alpha}{2\Delta^2}}$.
\end{lemma}
\begin{proof}
Assume without loss of generality that $|\tset_B|\leq |\tset_A|$.
Since $\tset$ is $\alpha$-well-linked in $G$, from Observations~\ref{obs: wl-properties} and \ref{obs: sparsest cut to flow}, there is a flow $F$ in $G$, where every vertex of $\tset_B$ sends one flow unit, every vertex in $\tset_A$ receives at most one flow unit, and the edge-congestion is at most $1/\alpha$. Therefore, the amount
  of flow through any vertex is at most
  $\Delta/\alpha$. Scaling this flow down by factor $\Delta/\alpha$,
  we obtain a $\tset_B$-$\tset_A$ flow of value at least $\frac{\kappa
    \alpha}{2\Delta^2}$ and vertex-congestion at most $1$.  From the
  integrality of flow, there is a set of $\ceil{\frac{\kappa
      \alpha}{2\Delta^2}}$ disjoint paths connecting terminals in
  $\tset_B$ to terminals in $\tset_A$.
\end{proof}

Given a balanced partition $(A,B)$ of $V$ with respect to $\tset$ and a collection of paths $\pset$
from $\tset_B$ to $\tset_A$. For each path $P \in \pset$, let $v(P)$ be the first vertex of $P$ that lies in $A$, and let $\Gamma_A(\pset)=\set{v_P\mid P\in \pset}$.
We now show an algorithm to construct a balanced partition with some useful properties.

\begin{theorem}\label{thm: find good balanced partition}
  There is an efficient algorithm to compute a balanced partition $(A,B)$
  of $V$ with respect to $\tset$ and a collection $\pset$ of $\ceil{\frac{\kappa
      \alpha}{2\Delta^2}}$ node-disjoint paths from $\tset_B$ to $\tset_A$
  such that $G[B]$ is connected, and set $\Gamma_A(\pset)$ is $1/\alphasc(\kappa)$-well-linked in $G[A]$.
\end{theorem}

\begin{proof}
  We say that a balanced partition $(A,B)$ of $V$ with respect to $\tset$ is \emph{good} if and only if both
  $G[A]$ and $G[B]$ are connected. We start with some initial good
  balanced partition $(A,B)$ and apply Lemma~\ref{lem:tB-tA-paths} to
  find a collection of paths $\pset$, and then perform a number of
  iterations. In every iteration, we will either find a new good
  balanced partition $(A',B')$ with $|E(A',B')|<|E(A,B)|$, or we will
  establish that the current partition has the required properties (after possibly switching $A$ and $B$). In the former
  case, we continue to the next iteration, and in the latter case we
  terminate the algorithm and return the current partition $(A,B)$ and the set
  $\pset$ of paths. Clearly, after at most $|E|$ iterations, our algorithm is
  guaranteed to terminate with the desired output.

  The initial partition $(A,B)$ is computed as follows. Let $T$ be any
  spanning tree of $G$, rooted at any vertex. Let $v$ be the lowest
  vertex of $T$ whose subtree contains at least
  $\frac{\kappa}{2\Delta}$ terminals. Since the degree of every vertex
  is at most $\Delta$, the subtree of $T$ rooted at $v$ contains at
  most $\frac{\kappa} 2+1$ terminals. We let $A$ contain all vertices
  in the subtree of $T$ rooted at $v$ (including $v$), and we let $B$
  contain all remaining vertices. Then both $A$ and $B$ contain at
  least $\frac{\kappa}{2\Delta}$ terminals, and both $G[A]$ and $G[B]$
  are connected.

  Given any good balanced partition $(A,B)$ of $V$, we perform an
  iteration as follows. Assume without loss of generality that
  $|\tset_A|\geq |\tset_B|$ (otherwise, we switch $A$ and $B$). First, we apply
  Lemma~\ref{lem:tB-tA-paths} to find a collection $\pset$ of
  $\ceil{\frac{\kappa \alpha}{2\Delta^2}}$ disjoint paths from
  $\tset_B$ to $\tset_A$.  Let $S =\Gamma_A(\pset)$; note that $|S| \le
  \kappa/2$. For a subset $Z\subseteq A$ of vertices, we denote
  $S_Z=Z\cap S$.  We set up an instance of the sparsest cut
  problem in graph $G[A]$ with the set $S$ of terminals. Let
  $(X,Y)$ be the partition of $A$ returned by the algorithm \algSC on
  this instance. If
  $\frac{|E(X,Y)|}{\min\set{|S_X|,|S_Y|}}\geq 1$, then we
  are guaranteed that $S$ is $1/\alphasc(\kappa)$-well-linked in
  $G[A]$. We then return $(A,B)$ and $\pset$, that are guaranteed to satisfy the
  requirements of the theorem. We now assume that
  $\frac{|E(X,Y)|}{\min\set{|S_X|,|S_Y|}}=\rho<1$.

  Our next step is to show that there is a partition $(X',Y')$ of $A$,
  such that $G[X']$ and $G[Y']$ are both connected, and the sparsity
  of the cut $(X',Y')$ in $G[A]$ (with respect to $S$) is at most $\rho$. In order to show this, we
  start with the cut $(X,Y)$, and perform a number of iterations.  Let
  $\cset$ be the set of all connected components of $G[A]\setminus
  E(X,Y)$. Each iteration will reduce the number of the connected
  components in $\cset$ by at least $1$, while preserving the sparsity
  of the cut. Let $\cset_1\subseteq \cset$ be the set of all connected
  components contained in $X$, and let $\cset_2\subseteq \cset$ be the
  set of connected components contained in $Y$. Assume without loss of generality that
  $|S_X|\leq |S_Y|$. If there is some component $C\in \cset$ with $|S_C|=0$, then we can move the vertices of $C$ to the opposite side of the partition $(X,Y)$, and obtain a new partition $(X',Y')$ whose sparsity is less than $\rho$, and the number of connected components in $G[A]\setminus E(X',Y')$ is strictly smaller than $|\cset|$. Therefore, we assume from now on that for each $C\in \cset$, $|S_C|>0$.

  Assume first that $|\cset_1|>1$. Then $|E(X,Y)|=\rho\cdot
  |S_X|$, and so there is a connected component $C\in \cset_1$
  with $|E(C,Y)|\geq \rho\cdot |S_C|$. Moreover, $|S_X|>|S_C|$, since we have assumed that for each $C'\in \cset$, $|S_{C'}|>0$. Consider a new partition
  $(X',Y')$ of $A$, with $X'=X\setminus C$ and $Y'=Y\cup C$.  Notice
  that the number of the connected components in $G[A]\setminus E(X',Y')$
  is strictly smaller than $|\cset|$. We claim that the sparsity of
  the new cut is at most $\rho$. Indeed, the sparsity of the new cut
  is:

  \[\frac{|E(X',Y')|}{|S_{X'}|}=\frac{|E(X,Y)|-|E(C,Y)|}{|S_X|-|S_C|}\leq
  \frac{\rho|S_X|-\rho
    |S_C|}{|S_{X}|-|S_C|}=\rho.\]

  Assume now that $|\cset_2|>1$, and denote
  $|E(X,Y)|/|S_Y|=\rho'$. Then $\rho'\leq \rho$. As before, there
  is a connected component $C\in \cset_2$ with $|E(C,X)|\geq
  \rho'|S_C|$ and $|S_C|<|S_Y|$. Consider a new partition $(X',Y')$ of $A$, where
  $X'=X\cup C$ and $Y'=Y\setminus C$. As before, the number of
  connected components in $G[A]\setminus E(X',Y')$ is strictly smaller
  than $|\cset|$. We now show that the sparsity of the new cut is at
  most $\rho$. If $|S_{Y'}|\leq |S_{X'}|$, then the sparsity
  of the new cut is:

  \[\frac{|E(X',Y')|}{|S_{Y'}|}=\frac{|E(X,Y)|-|E(C,X)|}{|S_Y|-|S_C|}\leq
  \frac{\rho'|S_Y|-\rho'
    |S_C|}{|S_{X}|-|S_C|}=\rho'\leq \rho.\]
 
  Otherwise, $|S_{X'}|< |S_{Y'}|$, and the sparsity
  of the new cut is:
 
  \[\frac{|E(X',Y')|}{|S_{X'}|}=\frac{|E(X,Y)|-|E(X,C)|}{|S_X|+|S_C|}<\frac{|E(X,Y)|}{|S_X|}=\rho.\]

  We continue this procedure, until $|\cset|=2$ holds, so in the end,
  $\cset=\set{G[X],G[Y]}$, and the sparsity of the cut $(X,Y)$ is at
  most $\rho<1$.  Assume without loss of generality that $|X\cap
  \tset_A|\leq |Y\cap \tset_A|$. We obtain a new partition $(A',B')$
  of $V$, by setting $A'=A\setminus X$ and $B'=B\cup X$. Notice that this
  is guaranteed to be a balanced partition with respect to $\tset$, since $|\tset_A|\geq
  \kappa/2$. The number of the edges in the new cut, $|E(A',B')|\leq
  |E(A,B)|-|E(X,B)|+|E(X,Y)|\leq |E(A,B)|-|S_X|+|E(X,Y)|<|E(A,B)|$,
  since $|E(X,Y)|<|S_X|$.  Moreover, since $X$ contains at least one
  vertex of $S$, $G[B\cup X]$ is connected, while we have established
  above that $G[Y]$ is also connected. We then replace the partition
  $(A,B)$ with $(A',B')$ and continue to the next iteration.
\end{proof}

We apply the preceding theorem to find a balanced partition $(A,B)$ of $V$ with respect to $\tset$ and
a corresponding collection $\pset'$ of paths. 

Let $\Gamma \subseteq A$ be the set of vertices of $A$ incident to the
edges of $E(A,B)$ and similarly let $\Gamma'\subseteq B$ be the set of
vertices of $B$ incident to the edges of $E(A,B)$.  Consider some path
$P'\in \pset'$, and let $e=(v,u)$ be the first edge of $P'$ that
belongs to $E(A,B)$, with $v\in \Gamma'$, $u\in \Gamma$. (Recall that the paths are directed from $\tset_B$ to $\tset_A$). We truncate
the path $P'$ at vertex $v$, and we say that $u$ is a \emph{special
  neighbor} of $v$. Let $\pset$ be the resulting set of truncated
paths. Then $\pset$ is a set of $\ceil{\frac{\kappa
    \alpha}{2\Delta^2}}$ node-disjoint paths, connecting the vertices of
$\tset_B$ to the vertices of $\Gamma'$, and every path in $\pset$ is
completely contained in $G[B]$. Moreover, if $\Gamma''\subseteq \Gamma'$
denotes the set of endpoints of the paths in $\pset$, then for each
$v\in \Gamma''$, we have defined a special neighbor $u\in \Gamma$,
such that, if $v\neq v'\in \Gamma''$, then their special neighbors are
distinct. Let $S$ be the set of the special neighbors. Recall that from Theorem~\ref{thm: find good balanced partition}, $S$ is $1/\alphasc(\kappa)$-well-linked in
$G[A]$. Let $q=2\Delta\alphasc(\kappa)$.  We need the following
theorem that allows us to group the paths in $\pset$ inside the graph
$G[B]$.

\begin{theorem}
  There is an efficient algorithm to find a subset $\tpset\subseteq
  \pset$ of at least $|\pset|/2$ paths, and a collection $\cset$ of
  disjoint connected subgraphs of $G[B]$, such that each path $P\in
  \tpset$ is completely contained in some subgraph $C\in \cset$, and
  each such subgraph contains at least $q$ and at most $4\Delta q$
  paths in $\tpset$.
\end{theorem}

\begin{proof}
  Start from $G[B]$ and build a new graph $H$, by contracting every
  path $P\in \pset$ into a super-node $u_P$. Let $U=\set{u_P\mid P\in
    \pset}$ be the resulting set of super-nodes. Let $T$ be any
  spanning tree of $H$, rooted at any vertex $r$. Given a vertex $v\in V(T)$, let $T_v$ be
  the sub-tree of $T$ rooted at $v$. Let $J_v\subseteq V(G)$ be the set of all vertices of
  $T_v$ that do not belong to $U$, and all vertices on paths $P$ with
  $u_P\in T_v$. In other words, $J_v=(V(T_v)\setminus U)\cup \left(\bigcup_{u_P\in V(T_v)\cap U}V(P)\right )$. 
  Denote $G_v=G[J_v]$. 
  Over the course of the
  algorithm, we will delete some vertices of $T$. The notation $T_v$
  and $G_v$ is always computed with respect to the most current tree
  $T$. We start with $\cset=\emptyset,\tpset=\emptyset$, and then iterate.

  Each iteration is performed as follows. If $q\leq |V(T)\cap U|\leq
  4\Delta q$, then we add $G_r$ to $\cset$, and terminate
  the algorithm. If $|V(T)\cap U|<q$, then we also terminate the
  algorithm (we will show later that $\tpset$ must contain at least
  $|\pset|/2$ paths at this point). Otherwise, let $v$ be the lowest
  vertex of $T$ with $|T_v\cap U|\geq q$. If $v\not \in U$, then,
  since the degree of every vertex is at most $\Delta$, $|T_v\cap
  U|\leq \Delta q$. We add $G_v$ to $\cset$, and all paths in
  $\set{P\mid u_P\in T_v}$ to $\tpset$. We then delete all vertices of
  $T_v$ from $T$, and continue to the next iteration.

  Assume now that $v=u_P$ for some path $P\in \pset$. If $|T_v\cap
  U|\leq 4\Delta q$, then we add $G_v$ to $\cset$, and all paths in
  $\set{P'\mid u_{P'}\in T_v}$ to $\tpset$ and continue to the next
  iteration. So we assume that $|T_v\cap U|> 4\Delta q$.

  Let $v_1,\ldots,v_z$ be the children of $v$ in $T$.  Build a new
  tree $T'$ as follows. Start with the path $P$, and add the vertices
  $v_1,\ldots,v_z$ to $T'$. For each $1\leq i\leq z$, let
  $(x_i,y_i)\in E(G[B])$ be any edge connecting some vertex $x_i\in V(P)$ to
  some vertex $y_i\in V(G_{v_i})$; such an edge must exist from the definition
  of $G_{v_i}$ and $T$. Add the edge $(v_i,x_i)$ to $T'$. Therefore,
  $T'$ is the union of the path $P$, and a number of disjoint stars
  whose centers lie on the path $P$, and whose leaves are the vertices
  $v_1,\ldots,v_z$. The degree of every vertex of $P$ is at most
  $\Delta$. The weight of the vertex $v_i$ is defined to be the number
  of paths in $\pset$ contained in $G_{v_i}$. Recall that the weight
  of each vertex $v_i$ is at most $q$, by the choice of $v$. For each
  vertex $x\in P$, the weight of $x$ is the total weight of its
  children in $T'$. Recall that the the total weight of the vertices
  of $P$ is at least $4\Delta q$, and the weight of every vertex is at
  most $\Delta q$. We partition $P$ into a number of disjoint segments
  $\Sigma=(\sigma_1,\ldots,\sigma_{\ell})$ of weight at least $q$ and
  at most $2\Delta q$ each, as follows. Start with $\Sigma=\emptyset$,
  and then iterate. If the total weight of the vertices of $P$ is at
  most $2\Delta q$, we build a single segment, containing the whole
  path. Otherwise, find the shortest segment $\sigma$ starting from
  the first vertex of $P$, whose weight is at least $q$. Since the
  weight of every vertex is at most $\Delta q$, the weight of $\sigma$
  is at most $\Delta q$. We then add $\sigma$ to $\Sigma$, delete it
  from $P$ and continue. Consider the final set $\Sigma$ of
  segments. For each segment $\sigma$, we add a new graph
  $C_{\sigma}$ to $\cset$. Graph $C_\sigma$ consists of the union of
  $\sigma$, the graphs $G_{v_i}$ for each $v_i$ that is connected to a
  vertex of $\sigma$ with an edge in $T'$, and the corresponding edge
  $(x_i,y_i)$. Clearly, $C_{\sigma}$ is a connected subgraph of
  $G[B]$, containing at least $q$ and at most $2\Delta q$ paths of
  $\pset$. We add all those paths to $\tpset$, delete all vertices of
  $T_v$ from $T$, and continue to the next iteration. We note that path $P$ itself is not added to $\tpset$, but all paths $P'$ with $u_{P'}\in V(T_v)$ are added to $\tpset$.

  At the end of this procedure, we obtain a collection $\tpset$ of
  paths, and a collection $\cset$ of disjoint connected subgraphs of
  $G$, such that each path $P\in \tpset$ is contained in some $C\in
  \cset$, and each $C\in \cset$ contains at least $q$ and at most
  $4\Delta q$ paths from $\tpset$. It now remains to show that
  $|\tpset|\geq |\pset|/2$. We discard at most $q$ paths in the last
  iteration of the algorithm. Additionally, when $v=u_P$ is processed,
  if $|T_v\cap U|> 4\Delta q$, then path $P$ is also discarded, but at
  least $4\Delta q$ paths are added to $\tpset$. Therefore, overall,
  $|\tpset|\geq |\pset|-\frac{|\pset|}{4\Delta q+1}-q\geq |\pset|/2$,
  since $|\pset|= \ceil{\frac{\kappa\alpha}{2\Delta^2}}$, while
  $q=2\Delta\alphasc(\kappa)$, and we have assumed that $\kappa\geq
  \frac{32\Delta^4\alphasc(\kappa)}{\alpha}$.
\end{proof}

For each graph $C\in \cset$, we select one path $P_C\in \tpset$ that
is contained in $C$, and we let $t_C$ be the terminal that serves as
an endpoint of $P_C$. Let $\Gamma'_C\subseteq \Gamma'$ be the set of all
vertices of $\Gamma'$ that serve as endpoints of paths of $\tpset$
that are contained in $C$. Then $|\Gamma'_C|\geq q$. We delete vertices
from $\Gamma'_C$ as necessary, until $|\Gamma'_C|=q$ holds. Our final
set $\tset'$ of terminals is $\tset'=\set{t_C\mid C\in
  \cset}$. Observe that $|\tset'|\geq\frac{|\tpset|}{4\Delta q} \geq
\frac{|\pset|}{16\Delta^2\alphasc(\kappa)}\geq
\frac{\kappa\alpha}{32\Delta^4\alphasc(\kappa)}$, as required.

It now only remains to show that $\tset'$ is node-well-linked in
$G$. Let $\tset_1,\tset_2$ be any pair of equal-sized subsets of
$\tset'$. Let $\tset^*=\tset_1\cap \tset_2$,
$\tset_1'=\tset_1\setminus \tset_2$, and $\tset_2'=\tset_2\setminus
\tset_1$. We set up an $s$-$t$ flow network, by adding a source $s$
and connecting it to every vertex of $\tset_1'$ with a directed edge,
and adding a sink $t$, and connecting every vertex of $\tset_2'$ to
it. We also delete all vertices of $\tset^*$ from the graph, and set
all vertex capacities, except for $s$ and $t$, to $1$; the capacities of $s$ and $t$ are infinite. From the
integrality of flow, it is enough to show a valid $s$-$t$ flow of
value $|\tset_1'|=|\tset_2'|$ in this flow network. This flow will be
a concatenation of three flows, $F_1,F_2,F_3$.

We start by defining the flows $F_1$ and $F_3$. Consider some terminal
$t'\in \tset'_1\cup \tset_2'$, and let $C\in \cset$ be the subgraph to
which $t'$ belongs. Let $T_C$ be any spanning tree of $C$. Terminal $t'$
sends one flow unit toward the vertices of $\Gamma'_C$ along the tree
$T_C$, such that every vertex in $\Gamma'_C$ receives $1/q$ flow
units. Let $F_1$ be the union of all these flows for all $t'\in
\tset_1'$, and $F_3$ the union of all these flows for all $t'\in
\tset_2'$ (we will eventually think of the flow in $F_3$ as directed
towards the terminals). Notice that for every vertex $v\in B\setminus
\tset^*$, the total flow that goes through vertex $v$ or terminates at
$v$ is at most $1$. We say that the flow is of type 1 if it originates
at a terminal in $\tset_1'$, and it is of type 2 otherwise.

We now proceed to define flow $F_2$. For every cluster $C\in \cset$,
each vertex $v\in \Gamma'_C$ sends the $1/q$ flow units it receives to
its special neighbor $u\in \Gamma$, along the edge $(v,u)$. Recall
that every vertex $u\in \Gamma$ serves as a special neighbor of at
most one vertex in $\Gamma'$. Let $\Gamma_1\subseteq \Gamma$ be the set
of vertices that receive flow of type $1$, and $\Gamma_2\subseteq\Gamma$
is the set of vertices that receive flow of type $2$. Then
$|\Gamma_1|=|\Gamma_2|$, and we denote $|\Gamma_1|=\kappa^*$.  It is
enough to show that there is a flow $F_2$ in $G[A]$, where every
vertex in $\Gamma_1$ sends $1/q$ flow units, every vertex in
$\Gamma_2$ receives $1/q$ flow units, and the total vertex congestion
due to this flow is at most $1/2$.

In order to define this flow, recall that since $\Gamma_1 \cup \Gamma_2
\subseteq S$ and $S$ is $1/\alphasc(\kappa)$-well-linked in $G[A]$, from Observations~\ref{obs: wl-properties} and~\ref{obs: sparsest cut to flow}, there is a flow in $G[A]$, where every vertex
in $\tset_1$ sends one flow unit, every vertex in $\tset_2$
receives one flow unit, and the
edge-congestion is bounded by $\alphasc(\kappa)$. The total flow through every vertex is then at most
$\Delta \alphasc(\kappa)$. Scaling this flow down by factor
$q=2\Delta\alphasc(\kappa)$, we obtain the flow $F_2$, where every
vertex of $\Gamma_1$ sends $1/q$ flow units, every vertex in $\Gamma_2$
receives $1/q$ flow units, and the total vertex congestion is at most
$1/2$. Combining together the flows $F_1,F_2,F_3$, we obtain the final
flow $F$.
From the integrality of flow, there is a set of $|\tset_1|=|\tset_2|$ disjoint paths connecting the vertices of $\tset_1$ to the vertices of $\tset_2$ in $G$.

\subsection{Proof of Theorem~\ref{thm: many leaves or a long 2-path}}

Let $T$ be any spanning tree of the graph $Z$.  If $T$ contains at least $L$ leaves, then we are done. 
Assume now that $T$ contains fewer than $L$ leaves.
 We will next try to perform some improvement steps in order to increase the number of leaves in $T$. 

Assume first that $T$ contains three vertices $a,b,c$, that have degree $2$ in $T$ each, where $b$ is the unique child of $a$ and $c$ is the unique child of $b$, and assume further that there is an edge $(a,c)$ in $Z$. We can then delete the edge $(b,c)$ and add the edge $(a,c)$ to $T$. It is easy to see that the number of leaves increases, with the new leaf being $b$ (see Figure~\ref{fig: procedure}).

Assume now that $v$ is a degree-$2$ vertex in $T$, such that both its father $v_1$ and grandfather $v_2$ are degree-$2$ vertices. Moreover, assume that the unique child $v_1'$ of $v$ is a degree-$2$ vertex, and so is the unique grandchild $v_2'$ of $v$. Assume that an edge $(v,u)$ belongs to $Z$, where $u\neq v_1,v_1'$. Notice that if $u=v_2$ or $u=v_2'$, then we can apply the transformation outlined above. Therefore, we assume that $u\neq v_2$ and $u\neq v_2'$. Two cases are possible. First, if $u$ is not a descendant of $v$, then we add the edge $(u,v)$ to $T$, and delete the edge $(v_1,v_2)$ from $T$. Notice that the number of leaves increases, as two new vertices become leaves - $v_1,v_2$, while in the worst case at most one vertex stops being a leaf (vertex $u$). The second case is when $u$ is a descendant of $v$. Then we add an edge $(u,v)$ to $T$, and delete the edge $(v_1',v_2')$ from $T$. Again, the number of leaves increases by at least $1$, since both $v_1'$ and $v_2'$ are now leaves.
(See Figure~\ref{fig: procedure}).

\begin{figure}[h]
\centering
\subfigure{\scalebox{0.4}{\includegraphics{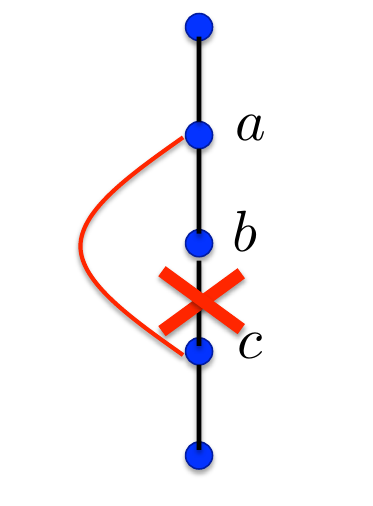}}\label{fig: case0}}
\hspace{1cm}
\subfigure[Case 1]{\scalebox{0.4}{\includegraphics{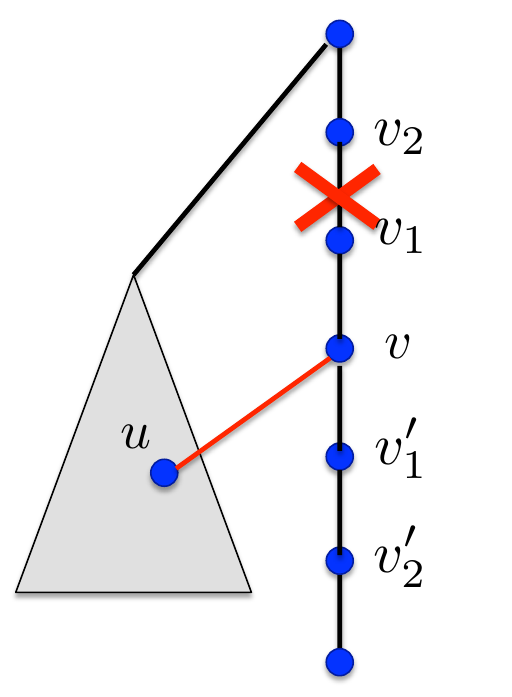}}\label{fig: case1}}
\hspace{1cm}
\subfigure[Case 2]{
\scalebox{0.4}{\includegraphics{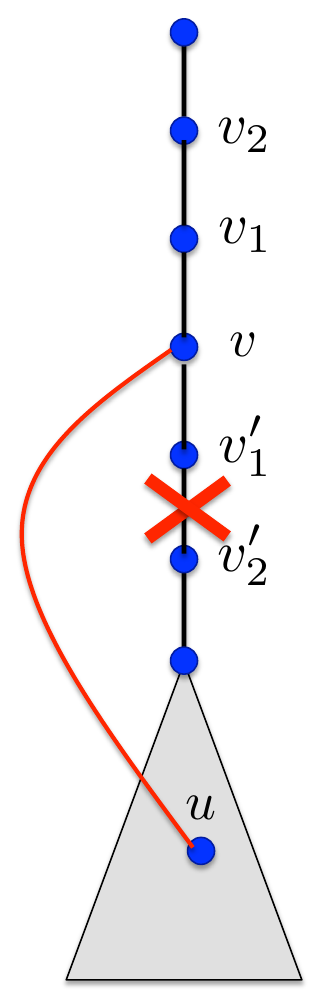}}\label{fig: case2}}
\caption{Improvement steps to increase the number of leaves. \label{fig: procedure}}
\end{figure}

We perform the above improvement step while possible. Let $T$ denote the final tree, where no such operation is possible. If $T$ has at least $L$ leaves, then we are done. Assume therefore that $T$ has fewer than $L$ leaves. Then the number of inner vertices of $T$ whose degree is greater than $2$ in $T$ is at most $L$. If $P$ is a maximal $2$-path in $T$, then the child of its lowermost vertex must be either a leaf or a vertex whose degree is more than $2$ in $T$. Therefore, there are at most $2L$ maximal $2$-paths in $T$, and at least one such path must contain at least $\frac{n-2L}{2L}\geq p+4$ vertices. Let $P'$ be the path obtained from $P$ by deleting the first two and the last two vertices. Then $P'$ contains at least $p$ vertices, and, since no improvement step was possible, $P'$ must be a $2$-path in $Z$.

\subsection{Proof of Lemma~\ref{lemma: re-routing of vertex-disjoint paths}}

The proof we provide here was suggested by Paul Seymour~\cite{PS-comm}. A different proof, using stable matchings, was shown by
Conforti, Hassin and Ravi~\cite{CHR}.

Let $\hat G'\subseteq \hat G$ be obtained from the union of the paths in $\xset_1\cup \xset_2$. Let $U_1'\subseteq U_1$ be the set of vertices where the paths of $\xset_1$ originate, and define $U_2'\subseteq U_2$ for the set $\xset_2$ of paths similarly.
Let $E_1$ be the set of all edges participating in the paths in $\xset_1$. While there is an edge $e\in E(\hat G')\setminus E_1$, such that graph $\hat G'\setminus\set{e}$ contains a set of $\ell_2$ nearly-disjoint $U_2'$--$s$ paths, we delete $e$ from $\hat G'$. At the end of this procedure, the final graph $\hat G'$ has the property that for every edge $e\in E(\hat G')\setminus E_1$, the largest number of nearly-disjoint $U_2'$--$s$ paths in graph $\hat G'\setminus\set{e}$ is less than $\ell_2$. Notice that $\xset_1\subseteq \hat G'$, and graph $\hat G'$ contains $\ell_2$ nearly-disjoint $U_2'$--$s$ paths. We need the following claim.

\begin{claim}\label{claim: augmenting}
There is a set $\xset'$ of $\ell_1$ nearly-disjoint $(U_1'\cup U_2')$--$s$ paths in graph $\hat G'$, such that exactly $\ell_2$ paths of $\xset'$ originate at the vertices of $U_2'$.
\end{claim}

Before we prove Claim~\ref{claim: augmenting}, we show that the set $\xset'$ of paths has the properties required by Lemma~\ref{lemma: re-routing of vertex-disjoint paths}. Let $\xset'_1\subseteq \xset'$ be the set of paths originating from the vertices of $U_1'\setminus U'_2$, and let $\xset'_2=\xset'\setminus\xset'_1$. We only need to show that $\xset'_1\subseteq \xset_1$. Assume otherwise. Then there is some edge $e\in E(\hat G')\setminus E_1$, that lies on some path in $\xset'_1$. But then $\xset'_2\subseteq \hat G'\setminus \set{e}$, and edge $e$ should have been removed from graph $\hat G'$. It now only remains to prove Claim~\ref{claim: augmenting}.

\begin{proofof}{Claim~\ref{claim: augmenting}}
The proof follows standard arguments. We construct a directed node-capacitated flow network $H$: start from graph $\hat G'$, and assign capacity $1$ to each vertex of $\hat G'$, except for vertex $s$, whose capacity is $\ell_1$. We add two new vertices: vertex $t_1$ of capacity $\ell_1-\ell_2$, that connects to every vertex of $U_1'$ with a directed edge, and vertex $t_2$ of capacity $\ell_2$ that connects to every vertex of $U_2'$ with a directed edge. Finally, we add a vertex $t$ of capacity $\ell_1$, that connects to $t_1$ and $t_2$ with directed edges. It is enough to show that there is a valid $t$--$s$ flow of value $\ell_1$ in the resulting flow network: we can then use the integrality of flow to obtain an integral flow of the same value, which in turn immediately defines the desired set $\xset'$ of paths.

Assume for contradiction that there is no $t$--$s$ flow of value $\ell_1$ in $H$. Then there is a set $Z$ of vertices, whose total capacity is less than $\ell_1$, such that $H\setminus Z$ contains no path connecting $t$ to $s$. Since the capacities of $t$ and $s$ are $\ell_1$ each, $s,t\not\in Z$. Also, since the capacities of $t_1$ and $t_2$ sum up to $\ell_1$, both these vertices cannot simultaneously belong to $Z$.

Assume first that $t_1\in Z$. Then, since $t_2\not \in Z$, set $Z$ contains at most $\ell_2-1$ additional vertices, each of which must have capacity $1$. Since there is a set of $\ell_2$ nearly-disjoint $U_2'$--$s$ paths in $\hat G'$, at least one such path $P$ is disjoint from $Z$, and so $H\setminus Z$ must contain a path connecting $t$ to $s$, a contradiction.

Similarly, if $t_2\in Z$, then $t_1\not\in Z$, and set $Z$ contains at most $\ell_1-\ell_2-1$ additional vertices, whose capacities must be all unit. But then at least one path in $\xset_1$ is disjoint from $Z$, giving a path connecting $t$ to $s$ in $H\setminus Z$, a contradiction.

Therefore, we assume that all vertices of $Z$ are capacity-$1$ vertices, that belong to $\hat G'$. But then $Z$ contains at most $\ell_1-1$ vertices, so at least one path in $\xset_1$ is disjoint from $Z$, giving again a path connecting $t$ to $s$ in $H\setminus Z$, a contradiction.
\end{proofof}

\subsection{Proof of Theorem~\ref{thm: starting point}}

Since $G$ has treewidth $k$, we can efficiently find a set $X$ of $\Omega(k)$ vertices of $G$ with properties guaranteed by Lemma~\ref{lem: find well-linked set}.

Using the cut-matching game and Theorem~\ref{thm: CMG}, we can embed an expander $H=(X,F)$ into
$G$ as follows. Each iteration $j$ of the cut-matching game requires
the matching player to find a matching $M_j$ between a given partition
of $X$ into two equal-sized sets $Y_j,Z_j$. From Lemma~\ref{lem: find well-linked set}, there exist a collection $\pset_j$ of
paths from $Y_j$ to $Z_j$, that cause congestion at most $1/\alpha^*$ on the vertices of $G$; these paths naturally define
the required matching $M_j$. The game terminates in $\cKRV(|X|)$
steps. Consider the collection of paths $\mP = \bigcup_j \pset_j$ and let
$G'$ be the subgraph of $G$ obtained by taking the union of these
paths. Let $H=(X,F)$ be the expander, whose vertex set is $X$ and edge set is $F=\bigcup_j M_j$. By the construction, for each $j$, a node $v$
of $G$ appears in at most $1/\alpha^*$ paths in $\pset_j$.  Therefore, the maximum
vertex degree in $G'$ is at most $2 \cKRV(|X|)/\alpha^*=O(\log^3k)$, and moreover the vertex- (and
hence also edge-) congestion caused by the set $\pset$ of paths in $G$ is also
upper-bounded by the same quantity. 
We apply the algorithm $\algsc$ to the sparsest cut instance defined by the graph $H$, where all vertices of $H$ serve as terminals. If the outcome is a cut whose sparsity is less than $\alphaKRV(|X|)$, then the algorithm fails; we discard the current graph $H$ and repeat the algorithm again. Otherwise, if the outcome is a cut of sparsity at least $\alphaKRV(|X|)$, then we are guaranteed that $X$ is an $\alphaKRV(|X|)/\alphasc(|X|)=\Omega(\sqrt{\log |X|})$-expander, and in particular, it is an $\alpha$-expander, for $\alpha=\half$. Since each execution of the cut-matching game is guaranteed to succeed with a constant probability, after $|X|$ such executions, the algorithm is guaranteed to succeed with high probability.

Since $H=(X,F)$ is an $\alpha$-expander, $X$ is
$\alpha$-well-linked in $H$. Since
$H$ is embedded into $G'$ with congestion at most $2\cKRV(|X|)/\alpha^*$, $X$ is
$\frac{\alpha\cdot \alpha^*}{2\cKRV(|X|)}=\Omega\left(\frac{1}{\log^3k}\right )$-well-linked in $G'$. Since the maximum
vertex degree in $G'$ is at most $2\cKRV(|X|)/\alpha^*=O(\log^3k)$, we can apply
Theorem~\ref{thm: grouping} to find a subset $X'\subseteq X$ of $\Omega\left(\frac{k}{\log^{15.5}k}\right)$ vertices, such that $X'$ is node-well-linked in $G'$.

\label{-----------------------------sec: proof of thm from path-system to grid----------------------------}
\section{Proof of Theorem~\ref{thm: find grid minor or good linkage}}
A set $\lset$ of $w$ disjoint paths in $G$ that connect the vertices of $A$ to the vertices of $B$   is called an \emph{$A$-$B$ linkage}. Since the sets $A,B$ of vertices are linked in
$G$, such a linkage $\lset$ exists and can be found efficiently. 

Given an $A$-$B$ linkage $\lset$,
we construct a graph $H=H(\lset)$ as follows.  The
vertices of $H$ are $U=\set{u_P\mid P\in \lset}$, and there is an
edge between $u_P$ and $u_{P'}$ if and only if there is a path $\gamma_{P,P'}$ in $G$, whose first vertex belongs to $P$,
last vertex belongs to $P'$, and the inner vertices do not belong to any
paths in $\lset$. Notice that since $G$ is a connected graph, so
is $H(\lset)$ for any $A$-$B$ linkage $\lset$. We say that an $A$-$B$ linkage $\lset$ is \emph{good} if and only if the longest $2$-path in the corresponding graph $H(\lset)$ contains fewer than $8h_1+1$ vertices.

Assume first that we are given a good linkage $\lset$ in $G$. Then Theorem~\ref{thm: many leaves or a long 2-path} guarantees that there is a spanning tree $\tau$ in $H_{\lset}$ with at least $\frac{w}{2(8h_1+5)}\geq h_2$ leaves. We let $\pset$ contain all paths $P\in \lset$ whose corresponding vertex $u_P$ is a leaf of $\tau$. Then $\pset$ contains at least $h_2$ node-disjoint paths, connecting vertices of $A$ to vertices of $B$. Consider any pair $P,P'\in \pset$ of paths with $P\neq P'$, and let $Q$ be the path connecting $u_P$ to $u_{P'}$ in $\tau$. Let $H_{P,P'}\subseteq G$ be the graph consisting of the union of all paths $P''$ with $u_{P''}\in V(Q)$, and paths $\gamma_{P_1,P_2}$ where $(u_{P_1},u_{P_2})$ is an edge of $Q$. Then graph $H_{P,P'}$ contains a path $\beta_{P,P'}$, connecting a vertex of $P$ to a vertex of $P'$, such that all inner vertices of $\beta_{P,P'}$ are disjoint from $\bigcup_{P''\in \pset}V(P'')$.



In order to complete the proof of Theorem~\ref{thm: find grid minor or good linkage}, we show, using the following theorem, that we can either find a model of the $(h_1\times h_1)$-grid minor in $G$, or compute a good $A$-$B$ linkage $\lset$.

\begin{theorem}\label{thm: improvement step}
There is an efficient algorithm, that, given an $A$-$B$ linkage $\lset$, such that $\lset$ is not a good linkage, returns one of the following:

\begin{itemize}
\item either a model of the $\left (h_1\times h_1\right)$-grid minor in $G$; or
\item a new $A$-$B$ linkage $\lset'$, such that the number of the degree-$2$ vertices in $H(\lset')$ is strictly smaller than the number of the degree-$2$ vertices in $H(\lset)$.
\end{itemize}
\end{theorem}

The proof of Theorem~\ref{thm: find grid minor or good linkage} immediately follows from Theorem~\ref{thm: improvement step}.
We start with an arbitrary $A$-$B$ linkage $\lset$, and iterate. While $\lset$ is not a good linkage, we apply Theorem~\ref{thm: improvement step} to it. If the outcome is a model of the $\left(h_1\times h_1\right )$-grid minor, then we terminate the algorithm and return this model.  Otherwise, if $\lset'$ is a good linkage, then we compute a subset $\pset\subseteq \lset'$ of paths as described above and terminate the algorithm.   Otherwise, we replace $\lset$ with $\lset'$ and continue to the next iteration. After $O(w)$ iterations, the number of degree-$2$ vertices in the graph $H(\lset)$ is guaranteed to fall below $8h_1+1$ (unless the algorithm terminates earlier). It now remains to prove Theorem~\ref{thm: improvement step}.

\subsection{Proof of Theorem~\ref{thm: improvement step}}
Since $\lset$ is not a good $A$--$B$ linkage, there is a $2$-path $R^*=(u_{P_0},\ldots,u_{P_{8h_1}})$ of length $8h_1+1$ in the corresponding graph $H=H_{\lset}$. Let $z=2h_1$, and consider the following four subsets of paths: $\pset_1=\set{P_1,\ldots,P_z}$, $\pset_2=\set{P_{z+1},\ldots,P_{2z}}$, $\pset_3=\set{P_{2z+1},\ldots,P_{3z}}$, and $\pset_4=\set{P_{3z+1},\ldots,P_{4z}}$, whose corresponding vertices participate in the $2$-path $R^*$. (Notice that $P_0\not\in \pset_1$, but the degree of $u_{P_0}$ is $2$ in $H(\lset)$ - we use this fact later). Let $X\subseteq A$ be the set of the endpoints of the paths in $\pset_2$ that belong to $A$, and let $Y\subseteq B$ be the set of the endpoints of the paths in $\pset_4$ that belong to $B$ (see Figure~\ref{fig-setup}). Since $A,B$ are linked in $G$, we can find a set $\qset$ of $z$ disjoint paths connecting $X$ to $Y$ in $G$. We view the paths in $\qset$ as directed from $X$ to $Y$


Let $Q\in \qset$ be any such path. Observe that, since $R^*$ is a $2$-path in $H(\lset)$, path $Q$ has to either intersect all paths in $\pset_1$, or all paths in $\pset_3$ before it reaches $Y$. Therefore, it must intersect $P_{z+1}$ or $P_{2z}$. Let $v$ be the last vertex of $Q$ that belongs to $P_{z+1}\cup P_{2z}$. Let $Q'$ be the segment of $Q$ starting from $v$ and terminating at a vertex of $Y$. Assume first that $v\in P_{z+1}$. We say that $Q$ is a type-1 path in this case. Let $u$ be the first vertex on $Q'$ that belongs to $P_0$. (Such a vertex must exist again due to the fact that $R^*$ is a $2$-path.) Let $Q^*$ be the segment of $Q'$ between $v$ and $u$. Then
$Q^*$ intersects every path in $\pset_1\cup\set{P_0,P_{z+1}}$, and does not intersect any other path in $\lset$, while $|V(Q^*)\cap V(P_0)|=|V(Q^*)\cap P_{z+1}|=1$ (see Figure~\ref{fig-setup}).

Similarly, if $v\in P_{2z}$, then we say that $Q$ is a type-2 path. Let $u$ be the first vertex of $Q'$ that belongs to $P_{3z+1}$, and let $Q^*$ be the segment of $Q'$ between $u$ and $v$. Then $Q^*$ intersects every path in $\pset_3\cup\set{P_{2z}\cup P_{3z+1}}$, and does not intersect any other path in $\lset$, while $|V(Q^*)\cap V(P_{2z})|=|V(Q^*)\cap V(P_{3z+1})|=1$. 

\begin{figure}[h]
\scalebox{0.6}{\includegraphics{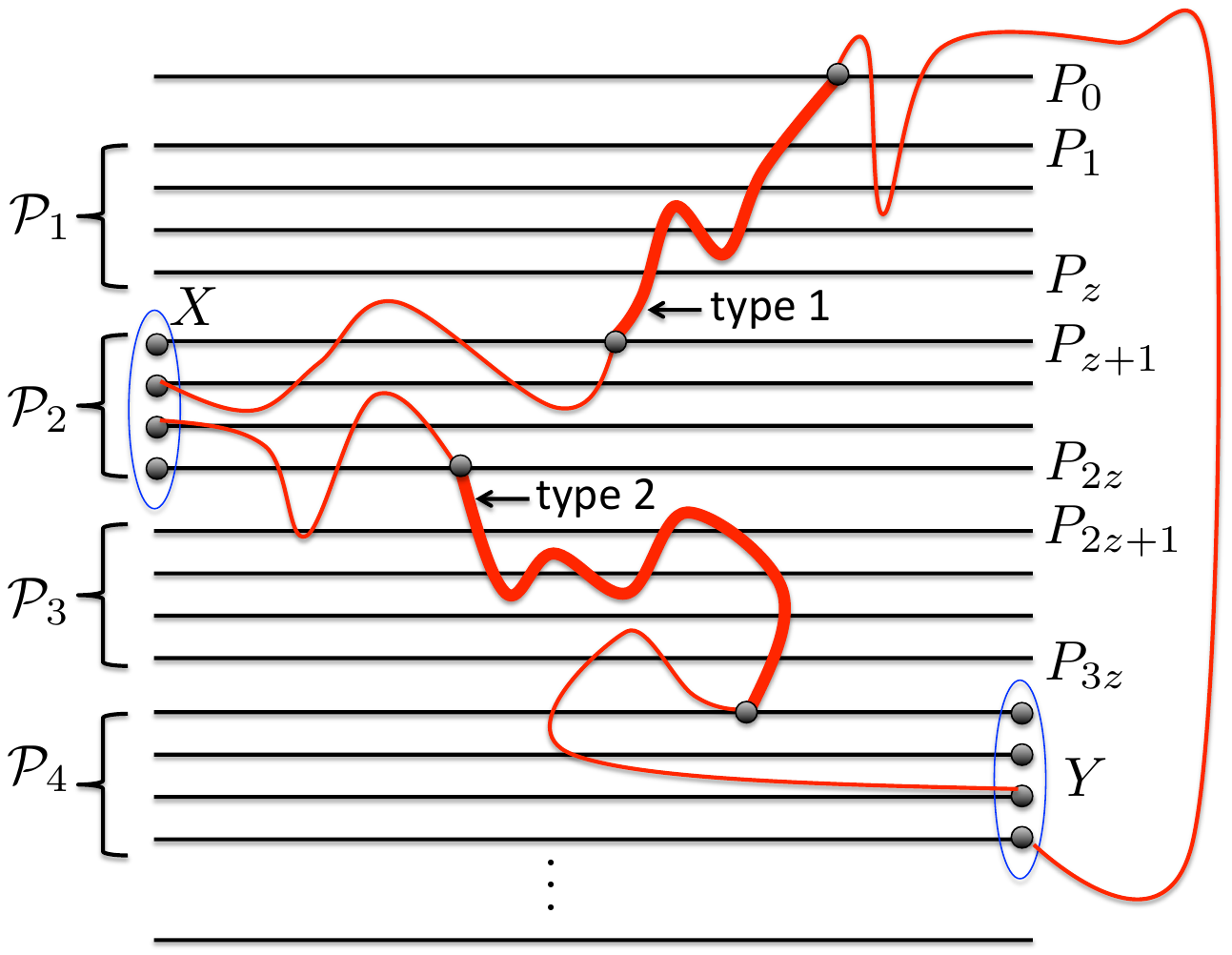}}
\caption{Two examples for paths in $\qset$ - a type-1 and a type-2 path - are shown in red, with the $Q^*$ segment highlighted.\label{fig-setup}}
\end{figure}

 Clearly, either at least half the paths in $\qset$ are type-1 paths, or at least half the paths in $\qset$ are type-2 paths. We assume without loss of generality that the former is true.
 Let $\qset'$ be the set of the sub-paths $Q^*$ for all type-1 paths $Q\in \qset$, that is, $\qset'=\set{Q^*\mid Q\in \qset \mbox{ and $Q$ is type-1}}$. Then $|\qset'|\geq z/2=h_1$.

The rest of the proof is based on the following idea. We will show that either the graph obtained from the union of the paths in $\qset'\cup \pset_1$ is a planar graph, in which case we recover a grid minor directly, or we will find a new $A$-$B$ linkage $\lset'$, such that $H(\lset')$ contains fewer degree-2 vertices than $H(\lset)$. To accomplish this we will iteratively simplify the intersection pattern of the paths in $\qset'$ and $\pset_1'$.

The algorithm performs a number of iterations. Throughout the algorithm, the set $\qset'$ of paths remains unchanged. The input to every iteration consists of a set $\pset_1'$ of paths, such that  the following hold:

\begin{itemize}
\item $\lset'=(\lset\setminus \pset_1)\cup \pset_1'$ is an $A$-$B$ linkage;

\item the graphs $H=H(\lset)$ and $H'=H(\lset')$ are isomorphic to each other, where the vertices $u_P$ for $P\not\in \pset_1$ are mapped to themselves; and

\item every path in $\qset'$ intersects every path in $\pset_1'\cup \set{P_0,P_{z+1}}$, and no other paths of $\lset'$.
 \end{itemize}

 The input to the first iteration is $\pset_1'=\pset_1$. Throughout
 the algorithm, we maintain a graph $\tilde H$ - the subgraph of $G$
 induced by the edges participating in the paths of $\pset_1'\cup
 \qset'$. We define below two combinatorial objects: a bump and a
 cross. We show that if $\tilde H$ has either a bump or a cross, then
 we can find a new set $\pset''_1$ of paths, such that
 $\lset''=(\lset'\setminus \pset_1')\cup \pset_1''$ is an $A$-$B$
 linkage. Moreover, either $H''=H(\lset'')$ contains fewer degree-$2$
 vertices than $H'$, or the two graphs are isomorphic to each
 other. In the former case, we terminate the algorithm and return the
 linkage $\lset''$. In the latter case, we show that we obtain a valid
 input to the next iteration, and $|E(\qset')\cup E(\pset_1')|>
 |E(\qset')\cup E(\pset_1'')|$. In other words, the number of edges in
 the graph $\tilde H$ strictly decreases in every iteration. We also show that,
 if $\tilde H$ contains no bump and no cross, then a large subgraph of
 $\tilde H$ is planar, and contains a grid minor of size $(h_1\times
 h_1)$. Therefore, after $|E(G)|$ iterations the algorithm is
 guaranteed to terminate with the desired output. We now proceed to
 define the bump and the cross, and their corresponding actions. We recall
 the useful observation that for any $A$-$B$ linkage $\lset'$, the
 corresponding graph $H(\lset')$ is a connected graph, since $G$ is
 connected. 

 \paragraph{A bump.} Let $\pset_1'$ be the current set of paths, and
 $\lset'=(\lset\setminus \pset_1)\cup \pset_1'$ the corresponding
 linkage.  We say that the corresponding graph $\tilde H$ contains a
 bump, if there is a sub-path $Q'$ of some path $Q\in \qset'$, whose
 endpoints, $s$ and $t$, both belong belong to the same path $P_j\in
 \pset_1'$, and all inner vertices of $Q'$ are disjoint from all paths
 in $\pset_1'$. (See Figure~\ref{fig: bump}).  Let $a_j\in A,b_j\in B$
 be the endpoints of $P_j$, and assume that $s$ appears before $t$ on
 $P_j$, as we traverse it from $a_j$ to $b_j$. Let $P'_j$ be the path
 obtained from $P_j$, by concatenating the segment of $P_j$ between
 $a_j$ and $s$, the path $Q'$, and the segment of $P_j$ between $t$
 and $b_j$.

\begin{figure}[h]
\scalebox{0.6}{\includegraphics{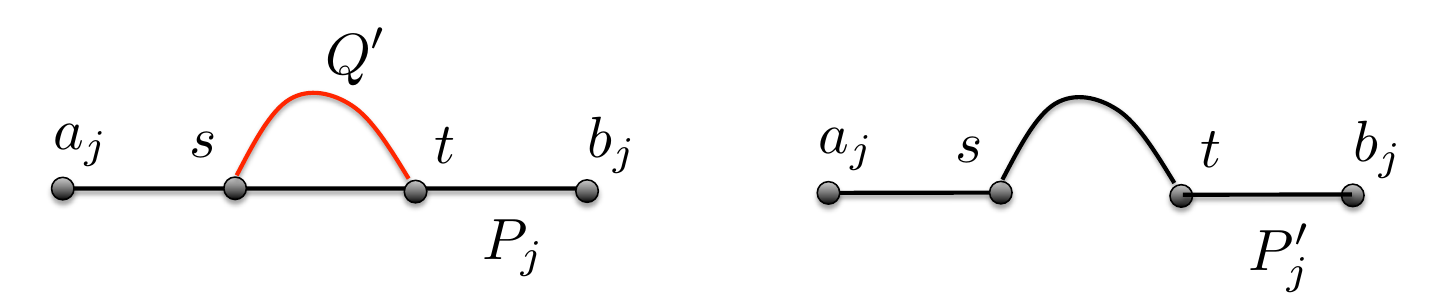}}
\caption{A bump and the corresponding action.\label{fig: bump}}
\end{figure}

Let $\pset_1''$ be the set of paths obtained by replacing $P_j$ with
$P'_j$ in $\pset_1'$, and let $\lset''=(\lset'\setminus \pset_1')\cup
\pset_1''=(\lset \setminus \pset_1)\cup \pset_1''$. It is immediate to
verify that $\lset''$ is an $A$-$B$ linkage. Let $H'=H(\lset')$, and
$H''=H(\lset'')$, and let $E'$ be the set of edges in the symmetric
difference of the two graphs (that is, edges, that belong to exactly
one of the two graphs). Then for every edge in $E'$, both endpoints
must belong to the set $\set{u_{P_{j-1}},u_{P_j},u_{P_{j+1}}}$; this is
because the vertices $\{u_P \mid P \in \pset_1\}$ are part of a $2$-path
in $H(\lset)$. In
particular, the only vertices whose degree may be different in the two
graphs are $u_{P_{j-1}},u_{P_j},u_{P_{j+1}}$. If the degree of any one
of these three vertices is different in $H''$ and $H'$, then, since
their degrees are $2$ in both $H'$ and the original graph $H$, we
obtain a new $A$-$B$ linkage $\lset''$, such that $H(\lset'')$
contains fewer degree-$2$ vertices than $H$. Otherwise, if the degrees
of all three vertices remain equal to $2$, then it is immediate to
verify that $H''$ is isomorphic to $H'$, where each vertex is mapped
to itself, except that we replace $u_{P_j}$ with $u_{P_{j'}}$. It is
easy to verify that all invariants continue to hold in this case. Let
$\tilde H$ be the graph obtained by the union of the paths in
$\pset_1'$ and $\qset'$, and define $\tilde H'$ similarly for
$\pset_1''$ and $\qset'$. Then $\tilde H'$ contains fewer edges than
$\tilde H$, since the portion of the path $P_j$ between $s$ and $t$
belongs to $\tilde H$ but not to $\tilde H'$.

\paragraph{A cross.}
Suppose we are given two disjoint paths $Q_1',Q_2'$, where $Q_1'$ is a
sub-path of some path $Q_1\in \qset'$, and $Q_2'$ is a sub-path of
some path $Q_2\in \qset'$ (with possibly $Q_1=Q_2$). Assume that the
endpoints of $Q_1'$ are $s_1,t_1$ and the endpoints of $Q_2'$ are
$s_2,t_2$. Moreover, suppose that $s_1,s_2$ appear on some path
$P_j\in \pset_1'$ in this order, and $t_2,t_1$ appear on $P_{j+1}\in
\pset_1'$ in this order (where the paths in $\pset_1'$ are directed
from $A$ to $B$), and no inner vertex of $Q_1'$ or $Q_2'$ belongs to
any path in $\pset_1'$. We then say that $Q_1',Q_2'$ are a cross. (See
Figure~\ref{fig: cross}.)

\begin{figure}[h]
\scalebox{0.6}{\includegraphics{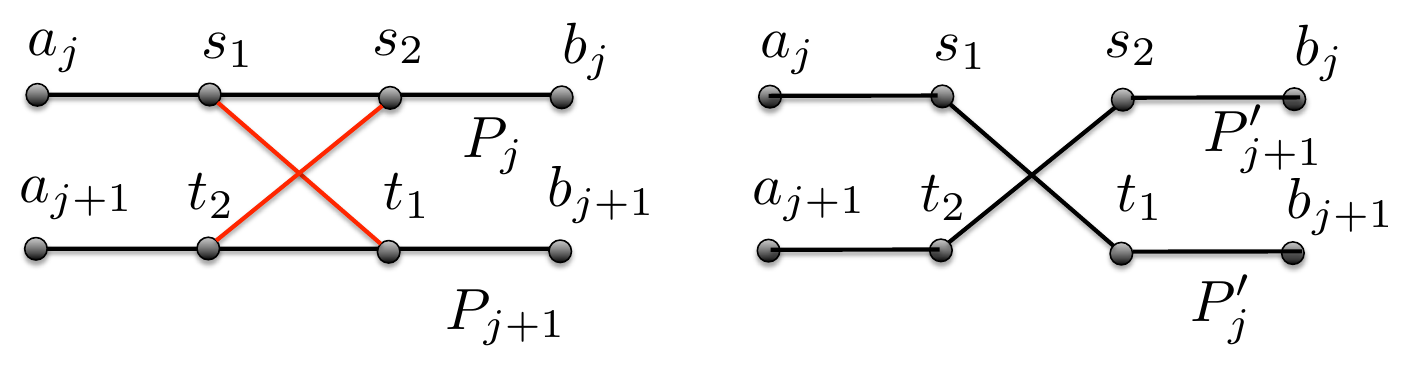}}
\caption{A cross and the corresponding action.\label{fig: cross}}
\end{figure}

Given a cross as above, we define two new paths, as follows. Assume
that the endpoints of $P_j$ are $a_j\in A$, $b_j\in B$, and similarly
the endpoints of $P_{j+1}$ are $a_{j+1}\in A$, $b_{j+1}\in B$. Let
$P'_j$ be obtained by concatenating the segment of $P_j$ between $a_j$
and $s_1$, the path $Q_1'$, and the segment of $P_{j+1}$ between $t_1$
and $b_{j+1}$. Let $P'_{j+1}$ be obtained by concatenating the segment
of $P_{j+1}$ between $a_{j+1}$ and $t_2$, the path $Q_2'$, and the
segment of $P_j$ between $s_2$ and $b_j$. We obtain the new set
$\pset''_1$ of paths by replacing $P_j,P_{j+1}$ with $P'_j,P'_{j+1}$
in $\pset_1'$. Let $\lset''=(\lset'\setminus \pset_1')\cup
\pset_1''=(\lset \setminus \pset_1)\cup \pset_1''$. It is immediate to
verify that $\lset''$ is an $A$-$B$ linkage. As before, let
$H'=H(\lset')$, and $H''=H(\lset'')$, and let $E'$ be the set of edges
in the symmetric difference of the two graphs. Then for every edge in
$E'$, both endpoints must belong to the set
$\set{u_{P_{j-1}},u_{P_j},u_{P_{j+1}},u_{P_{j+2}}}$. Again, this is
because the vertices $\{u_P \mid P \in \pset_1\}$ are part of a
$2$-path in $H(\lset)$.  The only vertices whose degree may be
different in the two graphs are
$u_{P_{j-1}},u_{P_j},u_{P_{j+1}},u_{P_{j+2}}$. If the degree of any
one of these four vertices is different in $H''$ and $H'$, then, since
their degrees are $2$ in both $H'$ and the original graph $H$, we
obtain a new $A$-$B$ linkage $\lset''$, such that $H(\lset'')$
contains fewer degree-$2$ vertices than $H$. Otherwise, if the degrees
of all four vertices remain equal to $2$, then it is immediate to
verify that $H''$ is isomorphic to $H'$, where each vertex is mapped
to itself, except that we replace $u_{P_j}, u_{P_{j+1}}$ with
$u_{P_{j}'}, u_{P_{j+1}'}$ (possibly switching them). It is easy to
verify that all invariants continue to hold in this case. Let $\tilde
H$ be the graph obtained by the union of the paths in $\pset_1'\cup
\qset'$, and define $\tilde H'$ similarly for $\pset_1''\cup
\qset'$. Then $\tilde H'$ contains fewer edges than $\tilde H$, since
the portion of the path $P_j$ between $s_1$ and $s_2$ belongs to
$\tilde H$ but not to $\tilde H'$.

We are now ready to complete the description of our algorithm. We start with $\pset_1'=\pset_1$, and then iterate. In every iteration, we construct a graph $\tilde H$ --- the subgraph of $G$ induced by $\pset_1'\cup \qset'$. If $\tilde H$ contains a bump or a cross, we apply the appropriate action. If the resulting linkage $\lset''$ has the property that $H(\lset'')$ has fewer degree-2 vertices than $H(\lset)$, then we terminate the algorithm and return $\lset''$. Otherwise, we obtain a valid input to the next iteration, and moreover, the number of edges in the new graph $\tilde H$ strictly decreases. Therefore, we are guaranteed that within $O(|E(G)|)$ iterations, either the algorithm terminates with the desired linkage $\lset''$, or the graph $\tilde H$ contains no bump and no cross. We now assume that the latter happens.

Consider the final graph $\tilde H$. For each path $Q\in \qset'$, let $v_Q$ be the first vertex of $Q$ that belongs to $V(\pset_1')$, and let $u_Q$ be the last vertex of $Q$ that belongs to $V(\pset_1')$. Let $\tilde Q$ be the sub-path of $Q$ between $v_Q$ and $u_Q$. Delete from $\tilde H$ all vertices of $V(Q)\setminus V(\tilde Q)$ for all $Q\in \qset'$, and let $\tilde H'$ denote the resulting graph. Let $\tilde \qset=\set{\tilde Q\mid Q\in \qset}$. We need the following claim. 

\begin{claim} If $\tilde H$ contains no cross and no bump, then $\tilde H'$ is planar.
\end{claim}
\begin{proof}
Consider some path $\tilde Q\in \tilde \qset$. Delete from $\tilde Q$ all edges that participate in the paths in $\pset_1'$, and let $\Sigma(\tilde Q)$ be the resulting set of sub-paths of $\tilde Q$. While some path $\sigma\in \Sigma(\tilde Q)$ contains a vertex $v\in V(\pset_1')$ as an inner vertex, we replace $\sigma$ with two sub-paths, where each subpath starts at one of the endpoints of $\sigma$ and terminates at $v$.
Let $\Sigma=\bigcup_{\tilde Q\in \tilde \qset}\Sigma(\tilde Q)$ be the resulting set of paths.
Then for each path $\sigma\in \Sigma$, both endpoints of $\sigma$ belong to $V(\pset_1')$, and the inner vertices are disjoint from $V(\pset_1')$. Moreover, since the paths in $\pset_1'$ induce a $2$-path in the corresponding graph $H(\lset')$, and since there are no bumps, the endpoints of each such path $\sigma$ connect two consecutive paths in $\pset_1'$. Since no crosses are allowed, it is easy to see that the graph $\tilde H'$ is planar.
\end{proof}

We now show how to construct a grid minor in graph $\tilde H'$. We start from the union of the paths in $\pset_1'$ and $\tilde \qset$, and perform the following transformation. We say that a segment $\sigma$ of a path $Q\in \tilde \qset$ is a \emph{hill} if and only if (i) the endpoints $s,t$ of $\sigma$ lie on some path $P_i\in \pset_1'$; (ii) the segment $\sigma'$ of $P_i$ whose endpoints are $s$ and $t$ does not contain any vertex of $V(\tilde\qset\setminus\set{Q})$; and (iii) $\sigma$ intersects $P_{i-1}$ and is internally disjoint from all vertices of $V(\pset_1'\setminus\set{P_{i-1}})$.  While there is a hill in $\pset_1'\cup \tilde \qset$, we modify the corresponding path $Q$ by replacing the segment $\sigma$ with $\sigma'$. If this creates a cycle on $Q$ (this can happen if $\sigma'$ contained a vertex of $Q$), we discard all such cycles until $Q$ becomes a simple path. We continue performing such transformations, until there is no hill in the set $\pset_1'\cup \tilde \qset$ of paths. Notice that this transformation cannot create any bumps. 
We need the following claim:

\begin{claim}\label{claim: grid}
When the above algorithm terminates, for all $P_i\in \pset_1'$ and $Q\in \tilde \qset$, $P_i\cap Q$ is a path.
\end{claim}

Notice that it is now immediate to obtain the $(h_1\times h_1)$-grid minor from the union of the paths in $\pset_1'\cup \tilde \qset$, by first contracting every path $P_i\cap Q$ for all $P_i\in \pset_1'$ and $Q\in \tilde \qset$, and then suppressing all degree-2 vertices, after which we discard the $h_1$ extra rows.  It is now enough to prove Claim~\ref{claim: grid}.

\begin{proof}
Assume otherwise. Then there must be some path $Q\in \qset$, and some segment $\sigma$ of $Q$, whose two endpoints $s,t$ lie on some path $P_i\in \pset_1'$, such that $\sigma$ intersects $P_{i-1}$, and it is internally disjoint from  $V(\pset_1'\setminus\set{P_i})$. Notice that since there are no bumps, the intersection of $Q$ and $P_{i-1}$ is a path. Among all such pairs $(Q,P_i)$, choose the one maximizing $i$. Since $\sigma$ is not a hill, there must be some path $Q'\neq Q$ in $\tilde \qset$ that intersects the segment $\sigma'$ of $P_i$, lying between $s$ and $t$. Let $v$ be any vertex in $Q'\cap \sigma'$, and let $\sigma''$ be the longest contiguous sub-path of $Q'$ contained in $\sigma'$. Let $u$ be the last vertex of $Q'$ before $\sigma''$ that belongs to $V(\pset_1')$, and let $u'$ be the first vertex of $Q'$ after $\sigma''$ that belongs to $V(\pset_1')$. Then it is easy to verify that $u,u'\in V(P_{i+1})$ (and in particular $i\neq 2h_1$). But then we should have chosen the pair $(Q',P_{i+1})$ instead of $(Q,P_i)$, a contradiction.
\end{proof}

\section{Proof of Corollary~\ref{cor: paths from the path-set system}}
We say that a cluster $S_i\in \sset$ is \emph{even} if $i$ is even,
and otherwise we say that $S_i$ is \emph{odd}.  We apply
Theorem~\ref{thm: find grid minor or good linkage} to graph $G[S_i]$
for every even cluster $S_i$, using $A=A_i$ and $B=B_i$. If, for any
even cluster $S_i$, the outcome is the $(h_1\times h_1)$-grid minor,
then we terminate the algorithm and return the model of this
minor. Therefore, we assume that for every even index $i$,
Theorem~\ref{thm: find grid minor or good linkage} returns a
collection $\lset_i$ of $h_2$ node-disjoint paths contained in
$G[S_i]$, that connect some subset $A_i'\subseteq A_i$ of $h_2$
vertices to a subset $B_i'\subseteq B_i$ of $h_2$ vertices, such that
for every pair $P,P'\in \lset_i$ of paths, there is a path
$\beta_i(P,P')$ in $G[S_i]$, connecting a vertex of $P$ to a vertex of
$P'$, where $\beta_i(P,P')$ is internally disjoint from $V(\lset_i)$.

Fix some $1\leq i\leq \floor{\ell/2}$. Let $A'_{2i}\subseteq A_{2i}$
and $B'_{2i}\subseteq B_{2i}$ be the sets of endpoints of the paths in
$\lset_{2i}$. Let $\lset^-_{2i-1}\subseteq \pset_{2i-1}$ be the set of
paths terminating at the vertices of $A'_{2i}$. If $2i<\ell$, then let
$\lset^+_{2i}\subseteq \pset_{2i}$ be the set of paths originating at
the vertices of $B'_{2i}$; otherwise, let $\lset^+_{2i}$ contain $h_2$
paths, each of which consists of a single distinct vertex of
$B_{2i}'$.

We use the odd clusters, via the well-linkedness properties of the
clusters, to connect the path collections from the even clusters into
the desired path collection $\qset$.

Consider now some odd-indexed cluster $S_{i}$. If $i\neq 1$, then let
$A'_{i}\subseteq A_{i}$ be the set of vertices where the paths of
$\lset^+_{i-1}$ terminate, and otherwise let $A'_{i}$ be any set of
$h_2$ vertices of $A_{i}$. If $i<\ell$, then let $B_{i}'\subseteq
B_{i}$ be the set of vertices where the paths of $\lset^-_{i+1}$
originate, and otherwise let $B_i'$ be any set of $h_2$ vertices of
$B_{i}$. Since $A_{i},B_{i}$ are linked in $G[S_i]$, there is a set
$\rset_{i}$ of $h_2$ node-disjoint paths, that are contained in
$G[S_{i}]$, and connect $A_{i}'$ to $B_{i}'$.

We now define a the set $\qset$ of  paths, obtained by the concatenation of all paths $\lset^-_i,\lset_i,\lset^+_i$ where $S_i$
is an even cluster, and paths $\rset_j$, where $S_j$ is an odd cluster.  The resulting set $\qset$ contains $h_2$ disjoint paths, originating at the vertices of $A_1$ and terminating at the vertices of $B_{\ell}$, where for every $1\leq i\leq \ell$, for every path $Q\in \qset$, $Q\cap S_i$ is a path, and $S_1\cap Q,S_2\cap Q,\ldots,S_{\ell}\cap Q$ appear on $Q$ in this order. Moreover, for every even integer $1\leq i\leq \ell$, for every pair $Q,Q'\in \qset$ of paths, there is a path $\beta_i(Q,Q')\subseteq G[S_i]$, that connects a vertex of $Q$ to a vertex of $Q'$, and is internally disjoint form all paths in $\qset$. It is immediate to verify that all paths in $\qset$ are contained in $G'$.

\subsection{Proof of Corollary~\ref{cor: from path-set system to grid minor}}

We apply Corollary~\ref{cor: paths from the path-set system} to the
path-of-sets system, with parameters $h_1=g$ and $h_2=g$, so $w\geq
16g^2+10g$ as required. If the outcome is the $(g\times g)$-grid
minor, then we terminate the algorithm and return its
model. Therefore, we assume that the outcome of Corollary~\ref{cor:
  paths from the path-set system} is a set $\qset$ of $g$ paths
connecting vertices of $A_1$ to vertices of $B_{\ell}$, that we denote
by $\qset=\set{Q_1,\ldots,Q_g}$ (the ordering is arbitrary).



Consider the following graph $G^*$. Start with a grid containing $g$
rows and $g(g-1)$ columns, with the columns indexed
$C_0,C_1,\ldots,C_{g(g-1)-1}$ from left to right. For all $0\leq
i<g(g-1)$, let $t_i=i\mod (g-1)$. We delete from the $i$th column all
edges except for $(t_i+1)$th edge from the top; thus, after this
operation the are exactly $g(g-1)$ vertical edges, one for each of the
columns. Finally, we repeatedly delete degree-$1$ vertices and
suppress degree-$2$ vertices (see Figure~\ref{fig: almost grid}). This
finishes the definition of the graph $G^*$. It is immediate to verify
that $G^*$ contains the $(g\times g)$-grid as a minor. In order to
finish the proof, it is enough to show that graph $G$ contains a
subdivision of $G^*$.  Each row $i$ corresponding to a horizontal path
of $G^*$ is mapped to the path $Q_i$. Each vertical edge of $G^*$ will
be mapped to a path in one of the even clusters as described below.
We denote the vertical edges of $G^*$ by
$e_0,e_1,\ldots,e_{g(g-1)-1}$, where $e_i$ is an edge that was lying
on column $C_i$ of the grid. For each $0\leq i<g(g-1)$, let $v_i,u_i$
be the two endpoints of $e_i$ in $G^*$; let $i_1$ be the index of the
row to which $v_i$ belongs, and assume without loss of generality that
$u_i$ belongs to row $(i_1+1)$. Let $P = Q_{i_1}$ the path that row
$i_1$ is mapped to, let $P' = Q_{i_1+1}$ the path that row $i_1+1$ is
mapped to.  In the graph $G[S_{2(i+1)}]$ there is path
$\beta_{2(i+1)}({P,P'})$ connecting a vertex $a$ lying on $P$ to a
vertex $b$ lying on path $P'$. We map the edge $e_i$ to path
$\beta_{2(i+1)}({P,P'})$, vertex $v_i$ to $a$ and $u_i$ to $b$.  Once
we complete the mapping of all vertices and vertical edges of $G^*$,
for every horizontal edge $e$ that lies say on the $j$th row of $G^*$,
we obtain a natural mapping of $e$ to the segment of $Q_j$ between the
two vertices to which the endpoints of $e$ are mapped.

\begin{figure}[h]
\centering
\scalebox{0.3}{\includegraphics{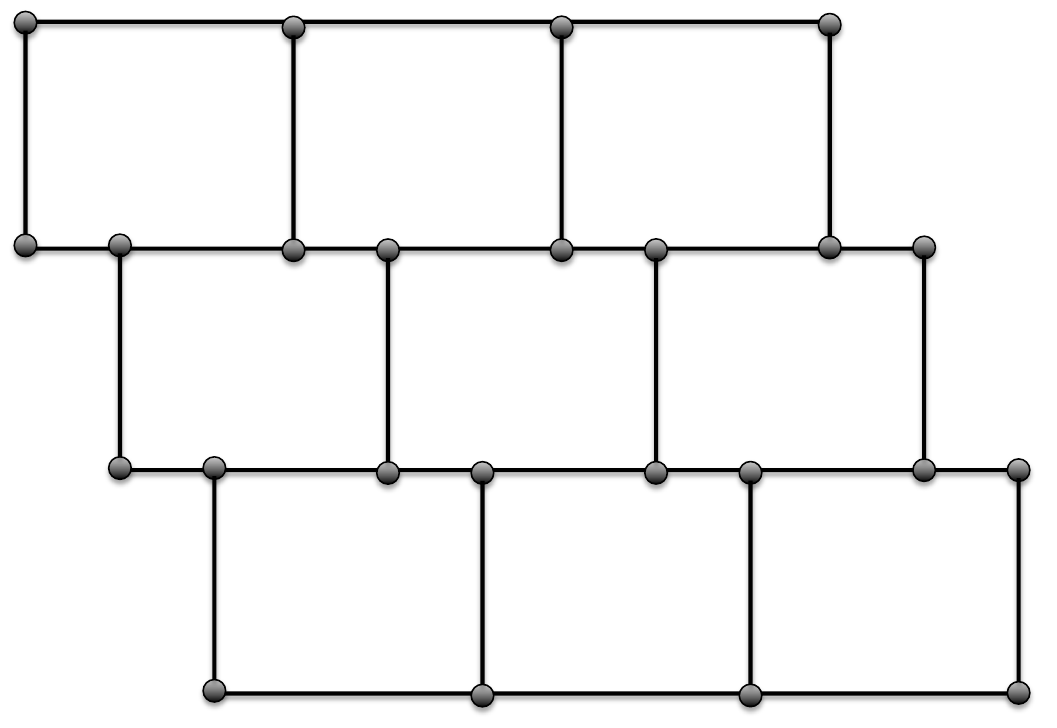}}
\caption{Graph $G^*$ for $g=4$. \label{fig: almost grid}}
\end{figure}

\label{------------------------------------------------------The End-------------------------------------------------------------------} 
\end{document}